\documentclass[12 pt,oneside,british]{amsart}

\usepackage{amsthm,amsmath}
\usepackage{amsfonts}
\usepackage{amsmath}
\usepackage{fancyhdr}
\usepackage[T1]{fontenc}
\usepackage[latin9]{inputenc}
\usepackage{mathabx}
\usepackage{color}
\usepackage{float}
\usepackage{mathrsfs}
\usepackage{layout}
\usepackage{slashed}
\usepackage{amsthm}
\usepackage{amstext}
\usepackage{amssymb}
\usepackage{stmaryrd}
\usepackage{graphicx}
\usepackage{setspace}
\usepackage{esint}
\RequirePackage[colorlinks=true,citecolor=blue,urlcolor=blue]{hyperref}
\usepackage[authoryear]{natbib}
\usepackage{ulem}
\usepackage{enumitem}
\usepackage{ifpdf}
\usepackage{pdflscape}
\usepackage[figure]{hypcap}
\usepackage{fancyhdr}
\usepackage{bbm}
\usepackage[bottom=1in,top=1in]{geometry}
\usepackage[parfill]{parskip}
\theoremstyle{plain}
\newtheorem{thm}{\protect\theoremname}
\theoremstyle{definition}
\newtheorem{example}[thm]{}
\theoremstyle{plain}
\newtheorem{prop}[thm]{Proposition}
\theoremstyle{definition}
\newtheorem*{defn*}{}
\theoremstyle{plain}
\newtheorem{cor}[thm]{Corollary}
\theoremstyle{plain}
\newtheorem{lem}[thm]{Lemma}
\theoremstyle{plain}

\ifpdf 
 \IfFileExists{lmodern.sty}{\usepackage{lmodern}}{}
\fi 
\let\myTOC\tableofcontents
\renewcommand\tableofcontents{  \frontmatter
  \pdfbookmark[1]{\contentsname}{}
  \myTOC
  \mainmatter }

\begin{document}

		\title{Set Identified Dynamic Economies and Robustness to Misspecification}

		\author{Andreas Tryphonides\\ Humboldt University}

		\address{Spandauerstr 1, 10178,Berlin,Germany. }

		\email{andreas.tryphonides@hu-berlin.de}

		\date{10/12/2017}

		\thanks{\tiny{This paper is based on chapter 2 of my PhD thesis (European University Institute). I thank Fabio Canova for his advice and help with earlier versions of this paper, and the rest of the thesis committee:  Peter Reinhard Hansen, Frank Schorfheide and Giuseppe Ragusa for comments and suggestions. This paper benefited also from discussions with George Marios Angeletos,  Eric Leeper, Ellen McGrattan, Juan J. Dolado, Thierry Magnac, Manuel Arellano, Eleni Aristodemou, Raffaella Giacomini, and comments from  participants at the 2nd IAAE conference, the 30th EEA Congress, T2M -BdF workshop, the ES Winter Meetings, the seminars at the University of Cyprus, Royal Holloway, Humboldt University, Universidad Carlos III, University of Glasgow, University of Southampton, the Econometrics and Applied Macro working groups at the EUI. Any errors are my own}}

	\setcounter{page}{1}
	\pagenumbering{arabic}
	\pagestyle{plain}
	\thispagestyle{empty}

\singlespacing
\begin{abstract}
We propose a new inferential methodology for dynamic economies that is robust to misspecification of the mechanism generating frictions. Economies with frictions are treated as perturbations of a frictionless economy that are consistent with a variety of mechanisms. We derive a representation for the law of motion for such economies and we characterize parameter set identification. We derive a link from model aggregate predictions to distributional information contained in qualitative survey data and specify conditions under which the identified set is refined. The latter is used to semi-parametrically estimate distortions due to frictions in macroeconomic variables. Based on these estimates, we propose a novel test for complete models. Using consumer and business survey data collected  by the European Commission, we apply our method to estimate distortions due to financial frictions in the  Spanish economy. We investigate the implications of these estimates for the adequacy of the standard model of financial frictions SW-BGG (\citet{10.1257/aer.97.3.586}, \citet{Bernanke_Gertler}). \newline \newline
\bigskip 
{\small 	Keywords:  Frictions, Heterogeneity, Set Identification, Surveys, Testing \newline JEL Classification: C13, C51, C52, E10, E60} 
\end{abstract}

	\maketitle

\newpage
\normalsize
\doublespacing


\section{Introduction}

Incorporating frictions in dynamic macroeconomic models is undoubtedly essential, both in terms of their theoretical and empirical implications. They are being routinely used to explain both the observed persistence of macroeconomic time series and several stylized facts obtained from firm and household level data. However, at the macroeconomic level, the selection between different types of
frictions and the specification of their corresponding mechanism is a complicated process
featuring arbitrary aspects. This arbitrariness implies that different studies
may find support for alternative types of frictions, depending on other
assumptions made. For example, the choice of nominal rigidities might depend
on which real rigidities are included in the model or whether there is
variable capital utilization or not\footnote{%
See the analysis of \citet{christiano2005nominal} and their
comparison to the results of \citet{2000}.}. More importantly, unless frictions are
micro-founded, policy conclusions may become whimsical as different
mechanisms imply different relationships between policy parameters and
economic outcomes.

In this paper we propose a new inferential methodology that is robust to
misspecification of the mechanism generating frictions in a dynamic stochastic economy. The approach treats
economies with frictions as perturbations of a frictionless economy, which are not uniquely pinned down, and are consistent with different structural specifications.  From a hypothesis testing point of view, the set of alternatives considered are no longer arbitrary. 

The paper makes several contributions. First, we derive a characterization for the law of motion of an economy with frictions that imposes identifying restrictions on the
solution of the model. \citet{ECTA:ECTA768}, CKM hereafter, identify wedges
in the optimality conditions of a frictionless model that produce the same
equilibrium allocations in economies with specific parametric choices for
the frictions. We also take a frictionless model as a benchmark but contrary
to CKM we construct a general representation for the wedge in the decision
rules. We illustrate through examples that the sign of the conditional mean
of the decision rule wedge is typically known, even when the exact mechanism
generating the friction is unknown. We utilize knowledge of the sign of the conditional mean
to set-identify the parameters of the benchmark model. 

Moment inequality restrictions
have been used to characterize frictions in specific markets, see for
example \citet{1996} and \citet{ECTA:ECTA1214}. To our knowledge, we are the first to characterize
such restrictions in dynamic stochastic macroeconomic models and to show their
relationship with the literature on wedges in equilibrium models. We are also the first to characterize set identification 
in this class of models. The econometric theory we develop can accommodate inequality restrictions of general form.

Due to set identification, many models with frictions are likely to be consistent with the robust identifying restrictions. Thus, additional data other than macroeconomic time series can be potentially useful in order to further constrain the set of admissible models. As a second contribution, our methodology permits the introduction of additional restrictions. We show how qualitative survey data can be useful 
for this purpose, since they provide distributional information which is highly relevant in models featuring  frictions and heterogeneity. More specifically, qualitative survey data are linked to current and future beliefs of agents in the model. If survey data reflect subjective conditional expectations, they
contain important indications about agents' actions through the behavioral equations. In the literature\footnote{%
We mainly refer to the treatment of survey data in the most recent "modern"
DSGE literature. The treatment of survey data in Rational Expectations
models date back to the work of \citet{Pesaran} and others.}, survey data are
typically linked to model based random variables using additional observation equations
with additive measurement error (see e.g. \citet{DelNegro2013}). However, this paper follows a different approach as the model is incomplete and there is no well defined model quantity that can be linked to the data. This is where the law of motion representation we derive proves useful, as we can link the macroeconomic distortions to aggregated qualitative surveys through additional moment inequality restrictions. 

Regarding possible applications of the methodology, we show how it contributes to model selection, where the
object of interest is not the set-identified parameter vector itself, but the implied semi-parametric estimates of frictions. To that end we propose a novel Wald test that compares the estimates of distortions implied by a candidate model to the robust set obtained using our method.   
We also derive large sample theory for this type of test.  Beyond restricting the set of observationally equivalent distortions, the use of additional data i.e. surveys prevents the test statistic from degenerating when the data generating process is such that the incomplete model nests the candidate model. 

We apply the methodology to estimate the distortions present in the Spanish
economy due to financial frictions. We estimate a small open economy version
of the \citet{10.1257/aer.97.3.586} model using qualitative survey data
collected by the Commission on the financial constraints of firms and consumers. We contrast the implied estimates of distortions to macroeconomic aggregates due to these frictions to those identified using a complete model that incorporates the \citet{Bernanke_Gertler} financial accelerator. 

We have already referred to how our paper relates to some strands of literature, namely the literature
on wedges and frictions (i.e. \citet{ECTA:ECTA768}, %
\citet{christiano2005nominal}, \citet{2000}) and the literature
on including survey data in DSGE models (i.e. \citet{DelNegro2013}). We contribute to the literature that deals with partial identification in
structural macroeconomic models (e.g. \citet{lubik,coroneo2011testing}) and the literature on applications of moment inequality models, see \citet{ECTA:ECTA1480} and references therein. 
We postpone illustrating more detailed connections to the vast partial identification literature for later sections, when we deal with estimation and inference using our methodology.

The rest of the paper is organized as follows. Section 2 provides a motivating example of the methodology in a partial equilibrium context, while section 3 presents  the prototype economy which will be used as an experimental lab to illustrate our methodology. Section 4 examines the distortions
present in the decision rules and their observable aggregate implications.
Section 5 provides the formal treatment of identification without additional information while Section 6 analyzes the case with additional information, including qualitative survey data. Section 7 discusses inference issues and provides the test. Section 8 contains the application
to Spanish data and Section 9 concludes and provides avenues for future
work. Appendix A contains proofs and the empirical results. Appendix B (online) contains more examples of moment inequality restrictions, details on the relation to the model uncertainty literature and the algorithm to compute frictions, an example on identification using our method, an illustration of the validity of bootstrap for the proposed test and details on the survey data used in the application. 

Throughout the paper we refer to three different conditional probability measures. The
objective probability measure, $\mathbb{P}_{t}$, the probability measure
determined by the frictionless model $M_{f}$, $\mathbb{P}^{M_{f}}_{t}(.|,.)$,
and the subjective probability measure $\mathcal{P}_{t,i}$ where $i$
identifies agent $i$ and $t$ the timing of the conditioning set. All three
are absolutely continuous with respect to Lebesgue measure. The
corresponding conditional expectation operators are $\mathbb{E}_{t}(.)$, $%
\mathbb{E}_{t}^{M_{f}}(.|.)$ and $\mathcal{E}_{i,t}(.)$. We distinguish
between the first and the second conditional expectation as the model will be correctly specified
only if the econometrician employs the right DGP. Moreover, $T$ denotes the length of both
the aggregate and average survey data, and $L$ the number of agents. We
denote by $\theta \in \Theta $ the parameters of interest with $\Theta^{CS}_{\alpha}$ the corresponding $1-\alpha$ level confidence set, and by $%
q_{j}(;\theta )$ a measurable function. Bold capital letters e.g. $\mathbf{Y}%
_{t}$ denote a vector of length $\tau$ containing $\{Y_{j}\}_{j\leq \tau}$. The
operator $\rightarrow _{p}$ signifies convergence in probability and the
operator $\rightarrow _{d}$ convergence in distribution; $\mathcal{N}(.,.)$
is the Normal distribution whose cumulative is $\Phi(.,.)$; $||.||$ is the Euclidean norm unless otherwise
stated; $\perp $ signifies the orthogonal complement and $\emptyset $ the
empty set. We denote by $(\Omega ,\mathcal{S}_{y},\mathbb{P})$ the
probability triple for the observables to the econometrician, where $%
\mathcal{S}_{y}=\sigma (Y)$, is the sigma-field generated by $Y$ and $%
(\Omega ,\mathcal{F},\{\mathcal{F}\}_{t\geq 0},\mathbb{P})$ the
corresponding filtered probability space.  Finally we denote by $\odot$ and $\oslash$ the 
Hadamard product and division respectively.

\section{Robustness in Partial Equilibrium}
We first motivate our approach by illustrating the special case of observing a single household 
receiving  a random exogenous endowment $y_{i,t}$ and making consumption-savings decisions $(c_{i,t},s_{i,t})$ in a partial equilibrium context. As in \citet{Zeldes1989},  wealth $w_{i,t}$ earns a riskless return $R$ and there is a borrowing limit at zero:
\begin{eqnarray*}
&&\max_{\{c_{i,t}\}_{t=1}^{\infty }}\mathbb{E}_{0}\sum_{t=1}^{\infty }\beta
^{t}\frac{c_{i,t}^{1-\omega }-1}{1-\omega }\\
s.t. &y_{i,t}&=s_{i,t}+c_{i,t}\\
&w_{i,t+1} &= Rw_{i,t}+s_{i,t}\\
&w_{i,t+1} &\geq 0
\end{eqnarray*}

The corresponding consumption Euler equation is distorted by the non-negative Lagrange multiplier on the
occasionally binding liquidity constraint, denoted by $\lambda_{i,t}$:
\begin{eqnarray}
U'(c_{i,t};\theta) & = & \beta\mathbb{E}_{t}(1+r_{t+1})U'(c_{i,t+1};\theta)+\lambda_{i,t}\label{eq:zeldes}
\end{eqnarray}
Whether the borrowing constraint is binding depends on the  household's expectations about future income, and therefore on its information set. Consequently, condition \eqref{eq:zeldes} is a conditional moment inequality, and for any conformable variable $z_{i,t}$ that belongs to the household's information set, the following unconditional moment inequality holds:
\begin{eqnarray}
\mathbb{E}\left[(U'(c_{i,t};\theta)-\beta(1+r_{t+1})U'(c_{it+1};\theta))\otimes z_{i,t}\right] & \geq & 0 \label{eq:unczeldes}
\end{eqnarray}
 Given the inequality, there is no unique vector $(\theta,\beta)$ that satisfies it, and we therefore no longer have 
point identification. In order to derive explicit identification regions, we adopt the approximation of \citet{Hall1978}, which for CRRA utility and constant interest rate implies
the following law of motion for consumption:
\begin{eqnarray}
c_{i,t+1}&=&\mu c_{i,t}+\epsilon_{i,t+1}+\tilde{\lambda}_{i,t}\label{eq:HallZeldes}
\end{eqnarray} 
where $\tilde{\lambda}_{i,t}\equiv-(U''(c_{i,t+1}))^{-1}\lambda_{i,t}$. Since $\mu$ may be exceed one if $\beta R>1$ and vice verca\footnote{In fact, $\mu$ is equal to $\left(\beta(1+r)\right)^{-\frac{u'(c)}{cu''(c)}}>0$, and it will be constant for CRRA utility.}, we re-express \ref{eq:HallZeldes}
in terms of consumption growth, $\Delta c_{i,t+1}=\tilde{\mu}c_{i,t}+\epsilon_{t+1}+\tilde{\lambda}_{i,t}$. 

Since the sign of $\mathbb{E}_{t}\tilde{\lambda}_{t}$ is still positive, the identified set for $\mu_{0}$ is : 
\begin{eqnarray}
\mu_{ID,1}&:=& \left(0,1+\frac{\mathbb{E}z_{i,t}\Delta c_{i,t+1}}{\mathbb{E}z_{i,t}c_{i,t}}\right] = \left(0,1+\hat{\mu}\right] \label{eq:bound1}
\end{eqnarray}
As we can readily see from \eqref{eq:bound1}, the inherent unobservability of $\tilde{\lambda}_{i,t}$ is the root of the loss of point identification of $\mu_{0}$. 
An interesting question is whether additional data can be informative about, or even point identify, $\mu_{0}$. The answer in this case is affirmative. 

Suppose that we observe the dichotomous response of the household over time to a survey question that asks whether the household is (or expects to be) financially constrained. An honest household will answer positively whenever $\lambda_{i,t}>0$. What this implies is that the time series of responses $\{\tilde{\chi}_{i,t}\}_{t\leq N}$ for $\tilde{\chi}_{i,t}\equiv\mathbf{1}('\lambda_{i,t}>0')$ can be used to estimate $\mathbb{P}_{t}(\lambda_{i,t}>0)$. The conditional probability distribution of $\Delta c_{i,t+1}$, $\mathbb{P}_{t}(\Delta c_{i,t+1}<u)$ for any $u\in \mathbb{R}$ is equal to
\begin{eqnarray} 
\mathbb{P}_{t}( \epsilon_{i,t+1}<u -\tilde{\mu}c_{i,t})\mathbb{P}_{t}(\lambda_{i,t}=0)+\mathbb{P}_{t}( \epsilon_{i,t+1}<u -\tilde{\mu}c_{i,t}-\lambda_{i,t})\mathbb{P}_{t}(\lambda_{i,t}>0) &&\label{eq:probit1}
\end{eqnarray}
For $\epsilon_{t+1}\sim N(0,\sigma_{\epsilon}^{2})$, \eqref{eq:probit1} simplifies further to,
\begin{eqnarray} 
\Phi_{0,\sigma_{\epsilon}^{2}}(u -\tilde{\mu}c_{i,t})\mathbb{P}_{t}(\lambda_{i,t}=0)+\Phi_{0,\sigma_{\epsilon}^{2}}(u -\tilde{\mu}c_{i,t}-\lambda_{i,t})\mathbb{P}_{t}(\lambda_{i,t}>0)&&\label{eq:probit2}
\end{eqnarray} 
A key observation is that the distortion $\lambda_{i,t}$ is positive only when the state variables $(w_{i,t},y_{i,t})$ lie in a particular area of their domain $\mathcal{A}$ which is determined by an unknown endogenous threshold $(w^{*}_{i},y_{i}^{*})$. A \textit{complete} model would pin down, analytically or numerically this domain. Observing $\tilde{\chi}_{i,t}$ dispenses us with the need to characterize $\mathcal{A}$. Once the threshold condition is satisfied, $\lambda_{i,t}$ will be a linear(ized) function of $(w_{i,t},y_{i,t})$. 

 Moreover, it is also reasonable to assume that $\lambda_{i,t}$ is a function of the unobserved exogenous shocks $\epsilon_{t+1}$. We can therefore substitute $\lambda_{i,t}$ for $\lambda_{1}c_{i,t}+\lambda_{2}\epsilon_{i,t+1}$, in \eqref{eq:probit2}, which yields the following quantile restriction:\small 
\begin{eqnarray} 
\nonumber \quad \mathbb{P}_{t}(\Delta c_{i,t+1}<u) &=&\Phi_{0,\sigma_{\epsilon}^{2}}(u -\tilde{\mu}c_{i,t})\mathbb{P}_{t}(\lambda_{i,t}=0)+\Phi_{0,\sigma_{\epsilon}^{2}(1+\lambda_{2})^{2}}(u -\tilde{\mu}^{c}_{ols}c_{i,t})\mathbb{P}_{t}(\lambda_{i,t}>0)\\
 &\geq&\Phi_{0,\sigma_{\epsilon}^{2}}(u -\tilde{\mu}c_{i,t})\mathbb{P}_{t}(\lambda_{i,t}=0)+\Phi_{0,\sigma_{ols}^{2}}(u -\tilde{\mu}_{ols}c_{i,t})\mathbb{P}_{t}(\lambda_{i,t}>0)\label{eq:probit4}
\end{eqnarray} 
\normalsize where  $\tilde{\mu}^{c}_{ols}$ is equivalent to the least squares estimate in the case of borrowing constraints with probability one. The inequality is derived for a class of functions $\mathbb{P}_{t}(\lambda_{i,t}>0)$.
 
 How does this additional inequality refine the identified set in \eqref{eq:bound1}? Denote by $\mu_{ID,2}$ the set of $\mu$ consistent with the unconditional moment inequality constructed using $z_{t}$ and  \eqref{eq:probit4}.  The simplest way to show the refinement of \eqref{eq:bound1} is to show that there exists a $\mu\in\mu_{ID,1}$ that does not belong to $\mu_{ID,2}$. For simplicity, let $\tilde{\mu}_{0}=0$, which implies that $\Delta c_{i,t+1}= \tilde{\chi}_{i,t}(\lambda_1c_{i,t}+\lambda_2\epsilon_{i,t+1})+(1-\tilde{\chi}_{i,t})\epsilon_{i,t+1}$. This is equivalent to considering an economy in which $\beta(1+r)<1$ and the consumer has very high risk aversion ($\omega\to \infty$). In Appendix A we show that using $\mu_{ols}$ as the test point, \eqref{eq:probit4} is not satisfied for $p_t=1(c_{i,t}<0)$, and therefore $\mu_{ols}\notin \mu_{ID,2}$. The upper bound in $\mu_{ID}$ should thus be lower than $\mu_{ols}$. In words, as long as the agent has some positive unconditional probability of not being constrained (here $\frac{1}{2}$), observing her responses provides additional information. 

The identified set for the reduced form, $\mu_{ID}=\mu_{ID,1}\cap\mu_{ID,2}$, is informative in different ways. \textit{First}, as already mentioned, given $\mu_{ID}$, we can recover the implied identified set for the risk aversion parameter, $\omega_{ID}$. As is typical in the empirical macro literature,  $(r,\beta)$ are not identified from the dynamic properties of the model but rather from other information, i.e steady states. In the case when $\beta(1+r)=1$, risk aversion is unidentified, $\omega_{ID} =\mathbb{R}^{+}$, as consumption follows a random walk. For $\beta(1+r)\neq 1$,
\begin{eqnarray*}
\omega_{ID} & := & \left\{\omega \in \mathbb{R}^{+}\cap\left[ \frac{\log(\beta(1+r))}{\log\left(1+\mu^{u}_{ID}\right)}, \frac{\log(\beta(1+r))}{\log\left(1+\mu^{l}_{ID}\right)}\right]\right\}
 \end{eqnarray*}
This interval has a very intuitive interpretation, and actually reflects restrictions on preferences implied by the presence of non-diversifiable income risk and observed behavior. For the sake of illustration let us focus on the identified set implied by \eqref{eq:bound1}. If $\beta(1+r)<1$, the household is impatient and does not accumulate wealth indefinitely.  Using $z_{i,t}=y_{i,t}$, since income and consumption growth are negatively correlated, which is usually the case when liquidity constraints are present\footnote{\citet{Deaton}}, $\omega_{ID,1}= \left\{\omega \in \mathbb{R}^{+}: \omega < \frac{\|\log(\beta(1+r))\|}{\log \mathbb{E}y_{i,t}c_{i,t}-\log\left({\mathbb{E}y_{i,t}c_{i,t}-{\|\mathbb{E}y_{i,t}\Delta c_{i,t+1}\|}{}}\right)}\right\}$. The stronger the negative correlation,  the lower the upper bound on risk aversion, indicating the fact that the less risk averse household is not accumulating enough wealth to fully insure against income risk. On the contrary, weak negative correlation permits higher estimates of risk aversion. Similar to \cite{ARELLANO2012256}, the degree of under-identification therefore depends on this correlation, which is itself a function of $(\omega,\beta,r)$.

\textit{Second}, the upper and lower bounds for $\mu_{ID}$ can be used to do inference on other interesting quantities, including the model itself. For example, we can construct set estimates of distortions in consumption implied by liquidity constraints by plugging-in the identified set and averaging over the observations. In population, this produces an identified set for distortions: 
\begin{eqnarray}
\mathbb{E}\lambda_{i,t} &\in &\left[\mathbb{E}c_{i,t+1}-\mu^{u}_{ID}c_{i,t},\mathbb{E}c_{i,t+1}-\mu^{l}_{ID}c_{i,t}\right] \label{eq:testing}
\end{eqnarray}
Another key observation is that although we have motivated the sign of $\mathbb{E}_{t}\lambda_{i,t}$ by looking at the exogenous constraint of no borrowing, the sign is robust to different mechanisms, as long as they prevent the household from smoothing consumption.  For example,  a non zero constraint on next period wealth, i.e. $w_{i,t+1}\geq \underbar{b}$ or a constraint that is endogenous, i.e. $w_{i,t+1}\geq \underbar{b}(y_{i,t})$ where $y_{i,t}$ is the relevant state variable, lead to an Euler equation which is distorted in the same direction. Therefore, the identified set of distortions will be a priori consistent with different mechanisms generating liquidity constraints. 

A mechanism will not be rejected as long as it generates distortions that will statistically belong to the identifed set for $\mathbb{E}\lambda_{t}$. Moreover, we will see that in general equilibrium, survey data will generate additional moment restrictions which can potentially add information. The identified set is expected to be refined, and this refinement can potentially provide power to reject alternative models that cannot be rejected using bound \eqref{eq:bound1}.

In the rest of the paper we will gradually characterize the case for general equilibrium models and what information aggregated surveys can provide, based on similar assumptions to those made in the simple example. To do so, we will provide below the benchmark general equilibrium model.
\section{Towards General Equilibrium}
In the last section we considered liquidity constraints for a single household in a partial equilibrium context. In this section,  we consider additional types of frictions
by including capital and a representative firm in the economy and set up the general equilibrium framework. Building up from the previous section, consider a simple Real Business Cycle (RBC) model, in which dynastic households with Constant Relative Risk Aversion (CRRA)
preferences form expectations about key state variables and make
consumption - savings decisions. We denote by $\Lambda_{t}(i)$ the distribution of agents at time $t$, and the corresponding aggregate
variables by capital letters i.e. $X_{t}\equiv \int {x_{i,t}d\Lambda_{t}(i)}$.
Households rent capital $(k_{i,t})$ to a representative
firm, which is used for production, and receive a share $(\eta _{i,t})$ of profits made by the firm. Profits 
are simply the production of output using capital intensive technology with random productivity $(Z_{t}K_{t}^{\alpha })$ minus the rents paid to all households $(R_{t}K_{t})$, that is,  $pr_{i,t}=\eta _{i,t}(Z_{t}K_{t}^{\alpha }-R_{t}K_{t})$. Household income is therefore $y_{i,t}=R_{t}k_{i,t}+pr_{i,t}$. Individual investment decisions $( \iota_{i,t})$ increase the availability of capital for next period, up to a certain level of depreciation. The household's optimization problem is therefore as follows:
\begin{eqnarray*}
&&\max_{\{c_{i,t},k_{i,t+1},w_{i,t+1},\iota_{i,t}\}_{1}^{\infty }}\mathcal{E}_{i,0}\sum_{t=1}^{\infty }\beta
^{t}\frac{c_{i,t}^{1-\omega }-1}{1-\omega }\\
s.t.  &w_{i,t+1} &= R_{t}w_{i,t} + y_{i,t}-c_{i,t}-\iota_{i,t} \\
      &k_{i,t+1} &=(1-\delta )k_{i,t}+  \iota_{i,t}
\end{eqnarray*}

Note that individual expectations $\mathcal{E}_{i,0}$ are not
necessarily formed with respect to the objective probability measure, nor with respect to the same information set. Denote by $\xi _{t}$ the aggregation residual, and by $\tilde{X}$ the percentage deviations of any aggregate variable $X_{t}$ from aggregate steady state $x_{ss}$. Then the system of aggregate linearized\footnote{Note that here we follow the linear approximation to the first order conditions as widely used in the DSGE literature, while in the last section we followed \citet{Hall1978}, where the variables are not in percentage deviation from steady state. The analysis is nevertheless largely unaffected.}
 equilibrium and market clearing conditions are as follows:  
 \begin{eqnarray*}
-\omega \tilde{C}_{t}-\xi _{C,t}-\beta r_{ss}\bar{\mathcal{E}_{t}}\tilde{R}%
_{t+1}+\omega \bar{\mathcal{E}_{t}}\tilde{C}_{t+1} &=&0 \\
y_{ss}\tilde{Y_{t}}-c_{ss}\tilde{C}_{t}-i_{ss}\tilde{I}_{t} &=&0 \\
\tilde{K}_{t+1}-(1-\delta )\tilde{K}_{t}-\delta \tilde{I}_{t} &=&0 \\
\tilde{Y}_{t}-\tilde{Z}_{t}-\alpha \tilde{K}_{t}-\xi _{Y,t} &=&0 \\
\tilde{R_{t}}-\tilde{Z}_{t}+(1-\alpha )\tilde{K}_{t}-\xi _{R,t} &=&0
\end{eqnarray*}
Under approximate linearity, that is, when $\xi _{t}$ is negligible 
\footnote{Absence of frictions is one of the reasons for which approximate linearity holds.}, 
we can obtain the aggregate decision rules for X$_{t}$, which will depend on predetermined capital $\tilde{k%
}_{i,t}$ and on aggregate conditional
expectations on productivity for some $j$ periods ahead, $\bar{\mathcal{E}}_{t}\tilde{Z}_{t+j}$.

In the previous section, we focused on the consumption decision of agents by taking the interest rate as exogenous and fixed, that is $R_{t}=R$. In this section, $R_{t}$ is endogenous as it is determined by the marginal productivity of capital, and is subject to aggregate risk. 
With regard to the aggregate implications of the borrowing constraints introduced in section 2, the characterization of individual behavior is identical. In general equilibrium, what needs to be taken into account is the determination of $R_{t}$. We will formally deal with general equilibrium in the next section.  Moreover, in this section the idiosyncratic risk to income is $\eta_{i,t}$ while there is also aggregate risk, $Z_{t}$. Compared to the last section, the introduction of aggregate risk does not alter the observable implications on individual consumption behavior. This is for the reason that although aggregate risk complicates forecasting for households as predicting $R_{t}$ requires tracking $\Lambda_{t}(i)$\footnote{\cite{krusellsmith}}, the inequality in \eqref{eq:unczeldes} is valid irrespective of the presence of aggregate uncertainty.  

In what follows we will analyze the equilibrium law of motion for investment , where we can introduce additional types of frictions, in particular those that arise from the production side of the economy. As in the case of consumption, we first analyze the 
aggregate investment decision rule in the frictionless case, which should satisfy the following second order difference equation: 
\begin{equation}
\frac{\bar{s}_{0}}{\alpha}\bar{\mathcal{E}_{t}}\tilde{I}_{t+1}-\left(1+\frac{(1-\alpha)(1-\bar{s}_{0})+\bar{s}_{0}\omega}{\alpha\omega}\right)\tilde{I}_{t}+\tilde{I}_{t-1}=\frac{1}{\alpha}\tilde{Z}_{t}
\label{eq:agg_inv_lag}
\end{equation}
where $\bar{s}_{0}\equiv\frac{I_{ss,0}}{Y_{ss,0}}$, the frictionless savings ratio.

A sufficient condition for a saddle point solution for the latter is $\alpha>\bar{s}$, that is, the steady state savings rate should be smaller than the elasticity of output with respect to capital. Denoting by $\rho_{1,0},\rho_{2,0}$ the two roots of the corresponding lag polynomial and letting $\rho_{2,0}$ be the unstable root, the aggregate frictionless investment has the following law of motion:
\begin{eqnarray}
\tilde{I}_{t}&=&\rho_{1,0}\tilde{K}_{t}+\frac{1}{\rho_{2,0}\bar{s}_{0}}\sum_{j=1}^{\infty}\bar{\mathcal{E}_{t}}\rho_{2,0}^{-j}\tilde{Z}_{t+j}
\label{eq:agg_inv}
\end{eqnarray}
Consider now the case in which the true model has no real or nominal frictions and agents
have rational expectations ($\mathcal{P}_{t,i} =\mathbb{P}_{t,i}$) where the corresponding information set is $\mathcal{F}_{i}\equiv(\eta_{i,t},k_{i,t},Z_{t})$.
Following the literature, we assume that only a subset of $\left\{\mathcal{F}_{i}\right\}_{i\leq N}$ is observed by the econometrician, i.e. $\mathcal{F}_{t}=K_{t}$.  The decision rule used by
the econometrician can then be rewritten as: 
\begin{equation}
\tilde{I}_{t}=\rho_{1,0}\tilde{K}_{t}+\frac{1}{\rho_{2,0}\bar{s}_{0}}\sum_{j=0}^{\infty}\rho_{2,0}^{-j}{\mathbb{E}_{t}}(\tilde{Z}_{t+j}|\sigma(\mathcal{F}_{t}))+\frac{1}{\rho_{2,0}\bar{s}_{0}}e_{t}  \label{eq:agg_inv_metr}
\end{equation}%
where 
\begin{eqnarray*}
e_{t}&=&\sum_{j=0}^{\infty }\rho_{2,0}^{-j}\int\mathbb{E}(\tilde{Z}_{t+j}|\mathcal{F}_{i,t})d\Lambda_{t}(i)-\sum_{j=0}^{\infty }\rho_{2,0}^{-j}\mathcal{\mathbb{E}}(\tilde{Z}_{t+j}|\mathcal{F}_{t})\\
&=&\sum_{j=0}^{\infty }\rho_{2,0}^{-j}\mathbb{E}(\tilde{Z}_{t+j}|\tilde{Z}_{t},\tilde{K}_t)-\sum_{j=0}^{\infty }\rho_{2,0}^{-j}\mathcal{\mathbb{E}}(\tilde{Z}_{t+j}|\tilde{K}_{t})
\end{eqnarray*}

The second equality holds due to the linearity of conditional expectations and aggregation. Then, for any $\tilde{K}_{t-j,j\geq0}$-measurable function $\phi (.)$, and using the law of total expectations,  the following moment equality holds\footnote{Note that we have not used any knowledge for the exogeneity of $\tilde{Z}_{t}$ for this result to hold i.e. it would hold trivially in the case of an exogenous productivity shock.}:
\begin{eqnarray*}
\mathbb{E}\tilde{K}_{t-j}e_{t}&=&0
\end{eqnarray*}%
The case of interest in this paper is therefore when $\int\mathcal{E}(.|\mathcal{F}_{i,t})d\Lambda_{t}(i)\neq 
\mathbb{E}^{M_{f}}(.|\mathcal{F}_{t})$, where $M_{f}$ is the frictionless probability model- or the benchmark model in the general case\footnote{To avoid confusion, I use frictionless and benchmark model interchangeably.}.  Here the agents and the econometrician not only
have different information, but they also have a different structure of the
economy in mind. $\mathbb{E}^{M_{f}}(.|X_{t})$ is consistent with frictionless behavior and
we thus implicitly assume that the frictionless part of the model is well specified.
The mismatch between the agents' expectations and the econometrician's
prediction could be due to differences in the models and/or information
sets, which can lead to different decision rules, and these differences do not vanish on average.
\vspace{-0.5 cm}

Therefore, denoting aggregate optimal investment in the presence of frictions by $\tilde{I}_{con,t}$, which is the data generating process, and by $I_{t}^{\ast }$ the frictionless investment rule used by the econometrician, we can represent these differences in terms of the econometrician's observables as follows: 
\begin{eqnarray*}
	\tilde{I}_{con,t}&=&\rho_{1,0}\tilde{K}_{t}+\frac{1}{\rho_{2,0}\bar{s}}\sum_{j=0}^{\infty}\rho_{2,0}^{-j}{\mathbb{E}_{t}}(\tilde{Z}_{t+j}|\sigma(\mathcal{F}_{t}))+\frac{1}{\rho_{2,0}\bar{s}}e_{t} +\tilde{\lambda}_{t} \\
	\tilde{\lambda}_{t} &=&\tilde{I}_{con,t}-I_{t}^{\ast }\label{eq:agg_inv_lambda}
\end{eqnarray*}  
or equivalently, as
\begin{eqnarray}
	\mathbb{E}\left( \tilde{I}_{con,t}-\rho_{1,0}\tilde{K}_{t}-\frac{1}{\rho_{2,0}\bar{s}}\sum_{j=0}^{\infty}\rho_{2,0}^{-j}\tilde{Z}_{t+j}-e_{t} -\tilde{\lambda}_{t}|\sigma(\mathcal{F}_{t})\right)&=& \mathbb{E}\left(\tilde{\lambda}_{t}|\sigma(\mathcal{F}_{t})\right)\label{eq:agg_inv_nest}
\end{eqnarray}
Similar to the case of consumption with liquidity constraints, real, nominal or informational frictions therefore generate a
\textquotedblleft wedge\textquotedblright\ $\tilde{\lambda}_{t}$ with non zero conditional mean and therefore a moment inequality. We will consider a general example of constraints in adjusting capital, which can be rationalized by ad hoc adjustment costs in capital accumulation, occasionally binding constraints i.e. capital irreversibility and financial frictions. The specific examples together with an example on non-rational expectations can be found in the Appendix B (online). Below, we present the general example.

\begin{subsubsection}{\textbf{{Capital Adjustment Constraints}}} 

Assuming full capital depreciation, \cite{Wang2012} show that collateral constraints can be represented by an aggregate function $\psi(\iota_{t})$ of the investment rate, $\iota_{t}\equiv\frac{I_{t}}{K_{t}}$. Note that this representation can be rationalized with different setups, including those in which there is a distribution of investment efficiency shocks across firms. Moreover, such setups can be further microfounded by different mechanisms i.e. financial frictions à la \citet{Bernanke_Gertler}. Focusing on exogenous collateral constraints, i.e. the borrowing limit does not depend on the asset value of the collateral\footnote{Endogenous collateral constraints have a different representation than exogenous constraints. Nevertheless, as \cite{Wang2012} note, endogenous collateral constraints imply a form of aggregate investment externality which implies insufficient level of investment relative to the case of exogenous constraints. This simply reinforces the distortions we derive based on exogenous constraints.},  they imply the following subsystem of linearized aggregate equilibrium conditions:
\begin{eqnarray}
\quad \tilde{K}_{t+1}&=& \psi_{ss}'\tilde{I}_{t}+(1- \psi'_{ss})\tilde{K}_{t}\\
\tilde{Q}_{t}&=&- \psi_{ss}^{'-1}\psi_{ss}''(\tilde{I}_{t}-\tilde{K}_{t})\\
\tilde{Q}_{t}&=& \omega\tilde{C_{t}}-\omega\mathbb{E}_{t}\tilde{C_{t+1}}+\Psi_{ss,1}\mathbb{E}_{t}\tilde{R}_{t+1}+\mathbb{E}_{t}\tilde{Q}_{t+1}+(\psi_{ss}'-1)\mathbb{E}_{t}(\tilde{I}_{t+1}-\tilde{K}_{t+1}) 
\end{eqnarray}
where $\Psi_{ss,1}\equiv 1+\beta(\psi_{ss}'-1)$. 

Denoting by $\Psi_{ss,2}\equiv \omega(\psi_{ss}'\bar{s}(\alpha-\bar{s})+\bar{s})-\bar{s}(1-\bar{s})\psi_{ss}^{'-1}\psi_{ss}''$, $\gamma_1\equiv \Psi_{ss,2}^{-1}(\omega\bar{s}- \psi_{ss}^{'-1}\psi_{ss}''\bar{s}(1-\bar{s}))$, $\gamma_2\equiv \Psi_{ss,2}^{-1}\omega$ and $\gamma_3\equiv \Psi_{ss,2}^{-1}\Psi_{ss,1}(1-\bar{s})$ the first order condition \eqref{eq:agg_inv_lag} becomes 
\begin{eqnarray}
	\gamma_1\mathbb{E}_{t}\tilde{I}_{t+1}-(1+\gamma_1+(1-\alpha)\gamma_3)\tilde{I}_{t}+\tilde{I}_{t-1}=\gamma_2\tilde{Z}_{t} && \label{eq:agg_inv_1}
\end{eqnarray} whose solution generalizes to:  
\begin{eqnarray*}
\tilde{I}_{con,t}&=&\rho_{1}(\gamma)\tilde{K}_{t}+\frac{1}{\rho_{2}(\gamma)\bar{s}(\gamma)}\tilde{Z}_{t}+\frac{1}{\rho_{2}}e_{t}+\tilde{\lambda}_{t}
\end{eqnarray*}
For $\rho_{i}(\gamma_0)\equiv\rho_{i,0},\bar{s}(\gamma_0)\equiv\bar{s}_{0}$, the corresponding distortion due to capital adjustment constraints is equivalent to
\begin{eqnarray*}
	\tilde{\lambda}_{t}&=& \left(\rho_{1}(\gamma)-\rho_{1}(\gamma_{0})\right)\tilde{K}_{t}+\left(\frac{1}{\rho_{2}(\gamma)\bar{s}(\gamma)}-\frac{1}{\rho_{2}(\gamma_0)\bar{s}(\gamma_0)}\right)\tilde{Z}_{t}+\left(\frac{1}{\rho_{2}(\gamma)}-\frac{1}{\rho_{2}(\gamma_0)}\right)e_{t}
\end{eqnarray*}
  
  In order to determine the sign of $\mathbb{E}\left(\tilde{\lambda}_{t}|\sigma(\mathcal{F}_{t})\right)$, it is sufficient to look at the sign of the coefficient of $\tilde{K}_{t}$. Using a mean value expansion around zero, the conditional mean is  \vspace{0.3 cm}
 	$\mathbb{E}\left(\tilde{\lambda}_{t}|\sigma(\mathcal{F}_{t})\right)=\rho_1'(\bar{\gamma})\tilde{K}_{t}$
  where $ \rho_1(\bar{\gamma})'\equiv\underset{(-)}{\rho_1(\bar{\gamma})'}\underset{(+)}{(\gamma_1-\gamma_{1,0})}
 	+\underset{(+)}{\rho_1'}(\bar{\gamma})\underset{(-)}{{(\gamma_3-\gamma_{3,0})}}\leq0 $.

 Frictions in capital accumulation therefore generate negative endogenous distortions, which implies that 
$\mathbb{E}\tilde{\lambda}_{t}\tilde{K}_{t}\leq 0$\footnote{The signs hold uniformly over $\gamma$, for any twice differentiable function $\psi$. This is straightforward for the derivatives, while for $\gamma_{i}-\gamma_{i,0}$, this gives rise to a quadratic inequality in $\psi'_{ss}$ which holds for all plausible $\psi'_{ss}\in(0,1]$.}. 
\end{subsubsection}

The preceding example dealt with distortions in the decision rules, which directly affect- and are therefore informative for-
the transmission of shocks and welfare. Moreover, they can be easily linked to additional data. Nevertheless, frictions are always justified using the first order necessary conditions for agents' optimal decisions. In the next
section, we present the general case together with a representation result, which proves useful in translating distortions to the first order conditions
into observationally equivalent general equilibrium distortions to decision rules. One can therefore directly use the latter. Because we work with linearized models, second or higher order effects
will be ignored. However, moment inequalities, would also appear in
nonlinear models. The way we treat frictions does not depend on assumptions regarding
the approximation error.

\section{Perturbing the Frictionless Model}

The general framework involves a system of expectational equations . Denote by $x_{i,t}$ the endogenous individual state, by $z_{i,t}$ the
exogenous individual state, and by $X_{t}=\int x_{i,t}d\Lambda (i)$ and $%
Z_{t}=\int z_{i,t}d\Lambda (i)$ the corresponding aggregate states. The optimality conditions characterizing the individual
decisions are as follows: 
\begin{eqnarray}
{G}(\theta )x_{i,t} &=&F(\theta )\mathcal{E}_{i,t}\left( \left( 
\begin{array}{c}
x_{i,t+1} \\ 
X_{t+1}%
\end{array}%
\right) |x_{i,t},z_{i,t},X_{t},Z_{t}\right) +L(\theta )z_{i,t}
\label{eq:exp_eq_ind} \\
z_{i,t} &=&R(\theta )z_{i,t-1}+\epsilon _{i,t}  \notag
\end{eqnarray}%
where $\mathbb{E}(\epsilon _{i,t})=0$. 
We assume that the
coefficients of the behavioural equations are common across agents. Relaxing this
assumption would make the notation more complicated, but would not change
the essence of our argument. We could also specify equilibrium conditions
that involve past endogenous variables but this is unnecessary as we can
always define dummy variables of the form $\tilde{x}_{i,t}\equiv x_{i,t-1}$
and enlarge the vector of endogenous variables to include $\tilde{x}_{i,t}$.
Aggregating across individuals, the economy can be characterized by the following system: 

\begin{eqnarray}
{G}(\theta )X_{t} &=&F(\theta )\int \mathcal{E}_{i,t}\left( \left( 
\begin{array}{c}
x_{i,t+1} \\ 
X_{t+1}%
\end{array}%
\right) |x_{i,t},z_{i,t},X_{t}, Z_{t}\right) d\Lambda_{t}(i)+L(\theta )Z_{t}
\label{eq:exp_eq_agg} \\
Z_{t} &=&R(\theta )Z_{t-1}+\epsilon _{t}
\end{eqnarray}%

We will refer to the economy with frictions as the triple $(H(\theta
),\Lambda ,\mathcal{E}_{i})$ where 

\begin{eqnarray*}
H(\theta)&\equiv& (vec(G(\theta )^{T}),vec(F(\theta ))^{T},vec(L(\theta
))^{T},vec(R(\theta ))^{T},vech(\Sigma _{\epsilon })^{T}).
\end{eqnarray*}
We partition the vector $\theta $ into two subsets, $(\theta _{1},\theta
_{2})$ where $\theta _{2}$ collects the parameters characterizing the presence and intensity of frictions. Thus,
setting $\theta_2=0$, reduces the model to the frictionless economy and this is without loss of generality. Furthermore, in an economy with no frictions, prices efficiently aggregate all the
information. Thus there is no need to distinguish between individual and
aggregate information when predicting aggregate state variables. 
\vspace{-0.2 cm}

For $w_{i,t}\equiv(x_{i,t},X_{i,t}, z_{i,t}, Z_{i,t})$, when agents are rational, and the model is linear ( or linearized) aggregate expectations for $(x_{i,t+1}^T,X_{t+1}^T)^T$ become as follows: \vspace{-0.2 cm}

\begin{eqnarray*}
\mathbb{E}x_{i,t+1}&=&\iint x_{i,t+1}p_{i}(x_{i,t+1} ,X_{t+1}|w_{i,t})d(x_{i,t+1}
,X_{t+1})d\Lambda_{t}(i)\\
&=&\iint x_{i,t+1}p(x_{i,t+1}|w_{i,t})d(x_{i,t+1})d\Lambda_{t}(i) \\
&=&\int P_{i,1}x_{i,t} d\Lambda(i) +P_2 X_{t} +\int P_{i,3}z_{i,t}
d\Lambda_{t}(i) +P_4 Z_{t}
\end{eqnarray*}
\begin{eqnarray*}
\mathbb{E}X_{t+1}&=&\iint X_{t+1}p_{i}(x_{i,t+1} ,X_{t+1}|w_{i,t})d(x_{i,t+1}
,X_{t+1})d\Lambda_{t}(i)\\
&=&\int X_{t+1}p(X_{t+1}|X_t,Z_t)d(X_{t+1}) \\
&=&P_5 X_{t} +P_6 Z_{t} \equiv\mathbb{E}(X_{t+1}|X_t,Z_t)
\end{eqnarray*} 

where $P_{j,j={1..6}}$ are the coefficients of the linear projection. By
Rational expectations and since the coefficients $(G,F,L)$ are common across i, equilibrium consistency requires $P_{i,1}=P_1,P_{i,2}=P_2$ and therefore $P_1+P_2=P_5$ and $P_3+P_4=P_6$. Thus, as expected, aggregate conditional expectations collapse to $\mathbb{E}(X_{t+1}|X_t,Z_t)$, and the frictionless economy, $(H(\theta _{1},0),\Lambda ,\mathbb{E})$, can be
summarized by the equilibrium conditions: 
\begin{eqnarray}  \label{eq:exp_sys}
{G}(\theta _{1},0)X_{t} &=&F(\theta _{1},0)\mathbb{E}%
_{t}(X_{t+1}|X_{t},Z_{t})+L(\theta _{1},0)Z_{t}  \label{eq:exp_eq_no_frict}
\\
Z_{t} &=&R(\theta )Z_{t-1}+\epsilon _{t}
\end{eqnarray}

Using the decision rule in the expectational
system \eqref{eq:exp_sys} and solving for the undetermined coefficients \footnote{See
	for example \citet{Marimon_Comp}.}, a Rational Expectations equilibrium for $(H(\theta _{1},0),\Lambda ,\mathbb{E})$ holds under the following conditions:
\begin{defn*}\textbf{ASSUMPTION-EQ}\newline  
There exist unique matrices $\underset{n_{x}\times n_{x}}{P^{\ast }(\theta
_{1},0)},\underset{n_{x}\times n_{z}}{Q^{\ast }(\theta _{1},0)}$ satisfying: \small
\begin{eqnarray*}
(F(\theta_{1},0)P^{\ast }(\theta _{1},0)-G^{\ast }(\theta _{1},0))P^{\ast
}(\theta _{1},0)&=&0 \\
({R(\theta )^{T}}\otimes F(\theta _{1},0)+I_{z}\otimes (-F(\theta
_{1},0)P^{\ast }(\theta _{1},0)+G^{\ast }(\theta _{1},0)))vec(Q(\theta
_{1},0))&=&-vec(L(\theta _{1},0))
\end{eqnarray*}%
\normalsize
such that $X_{t}=P^{\ast }(\theta _{1},0)X_{t-1}+Q^{\ast }(\theta
_{1},0)Z_{t}$ is a competitive equilibrium.
\end{defn*}
Since we know the model up to $\theta _{2}=0$, we
rearrange the equations of the economy with frictions into the known and the unknown part of the specification. Adding and subtracting the first order conditions of the frictionless economy we get that: \small
\begin{eqnarray}
G(\theta _{1},0)X_{t}&=&F(\theta _{1},0)\mathbb{E}_{t}(X_{t+1}|X_{t})+L(\theta
_{1},0)Z_{t}+{\mu _{t}} \label{eq:exp_sys_pert}
\end{eqnarray} 
\begin{eqnarray*}
\mu _{t} &\equiv &-(G(\theta )-G(\theta _{1},0))X_{t} +(L(\theta )-L(\theta _{1},0))Z_{t}\\
&&+(F(\theta )-F(\theta _{1},0))\int \mathcal{E}_{it}\left( \left( 
\begin{array}{c}
x_{i,t+1} \\ 
X_{t+1}%
\end{array}%
\right) |x_{i,t},X_{t}\right) d\Lambda_{t}(i) \\
&&+F(\theta _{1},0)(\int \mathcal{E}_{it}\left( \left( 
\begin{array}{c}
x_{i,t+1} \\ 
X_{t+1}%
\end{array}%
\right) |x_{i,t},X_{t}\right) d\Lambda_{t}(i)-\mathbb{E}_{t}(X_{t+1}|X_{t})) \normalsize
\end{eqnarray*} 
 \normalsize
This system of equations cannot be solved without knowing $\mu _{t}$. Nevertheless, we characterize the relationship between $\mu _{t}$
and a set of candidate decision rules that depend on the endogenous
states and some unobserved process, $\lambda _{t}$. Proposition 1
states sufficient conditions such that decision rules are
consistent with $\mu _{t}$. 
\bigskip
\begin{prop}
Given:
\begin{enumerate}
\item The perturbed system of equilibrium conditions \eqref{eq:exp_sys_pert}
where $\mathbb{E}_{t}\mu_{t}\geq 0$
\item A distorted aggregate decision rule $X_{t}^{\ast
}=X_{t}^{f,RE}+\lambda_{t}$ where $X_{t}^{f,RE}=\int x_{i,t}^{RE}d\Lambda
(i) $ is the Rational Expectations equilibrium of $(H(\theta_1,0),\Lambda,%
\mathbb{E})$ and
\item A ${\lambda }_{t}$ vector process such that ${\lambda }_{t}={\lambda }%
_{t-1}\Gamma +\nu _{t}$ for some real-valued $\Gamma $:
\end{enumerate}
If there exists a non-empty subset of $\Theta _{1}$ that satisfies 
\begin{eqnarray}
\mathbb{E}_{t}(F(\theta _{1},0)\Gamma -G(\theta _{1},0)){\lambda }_{t}&=&-%
\mathbb{E}_{t}\mu _{t}  \label{eq:repres}
\end{eqnarray}%
The following condition is satisfied for almost all subsets of $\sigma
(Y_{t-1})$ : 
\begin{eqnarray*}
\mathbb{E}_{t}(G(\theta _{1},0)X_{t}^{\ast }-F(\theta _{1},0)(\theta
)X_{t+1}^{\ast }-L(\theta _{1},0)Z_{t})&\geq& 0 
\end{eqnarray*}
\end{prop}
\begin{proof}
See Appendix
\end{proof}

Proposition 1 states that if for an admissible parameter
vector $\theta _{1}\in \Theta _{1}$ condition (24) is true, the
decision rule $X_{t}^{\ast }=X_{t}^{f,RE}+\lambda _{t}$ generates the same
restrictions as those implied using the perturbed (by $\mu _{t}$) first
order equilibrium conditions. We focus on parameter vectors
that yield determinate and stable equilibria in the frictionless economy, which implies a restriction on the stochastic behavior of $\lambda _{t}$.

While it is true that in many cases the sign of $\mathbb{E}_{t}\lambda_t$ can be directly deduced from the distortions to equilibrium conditions, Proposition 1 can also be useful in practice, as we can deduce the sign of the distortions to decision rules without solving for the otherwise unknown expectational system using \eqref{eq:repres}\footnote{In more complicated cases where knowledge of $\Gamma$ is strictly required, condition \eqref{eq:repres} can always be checked ex post.}. Moreover, Proposition 1 states conditions under which we can guarantee that the moment inequality restrictions are going to be consistent with the rest of the model, and therefore the implied reduced form. \vspace{-0.3 cm}
\begin{example}\textbf{Analytical example for Capital Adjustment constraints }\newline
	As illustrated in the previous section, distortions arise in the capital accumulation equation, Tobin's Q and the Euler equation. Reducing the system to its minimal representation in terms of capital and investment i.e. $X_{t}\equiv (\tilde{I}_{t},\tilde{K}_{t+1}$), we have that:
	\begin{equation*}
	\left[ 
	\begin{array}{cc}
	1& \gamma_{1,0}+(1-\alpha)\gamma_{3,0}\\
	1 & -1\\
	\end{array}%
	\right] \left[ 
	\begin{array}{c}
	\tilde{I}_{t}\\
	\tilde{K}_{t+1}\\
	\end{array}%
	\right] =\left[ 
	\begin{array}{cc}
	1&0\\
	0 & 0\\
	\end{array}%
	\right] \left[ 
	\begin{array}{c}
	\mathbb{E}_{t}\tilde{I}_{t+1}\\
	\mathbb{E}_{t}\tilde{K}_{t+2}\\
	\end{array}%
	\right]  +\left[ 
	\begin{array}{c}
	\gamma_{2,0}\\
	0 \\
	\end{array}%
	\right] \tilde{Z}_{t} +\left[ 
	\begin{array}{c}
	-\mu_{I,t}\\
	0\\
	\end{array}%
	\right]
	\end{equation*}%
		where the Investment Euler equation has a positive distortion as the marginal cost of deferring consumption increases. We
	can then derive $(\tilde{\lambda}_{I,t},\tilde{\lambda}_{K,t+1})$, and consequently distortions to all other endogenous variables i.e. consumption.
    Since $\tilde{K}_{t+1}=\tilde{I}_{t}$, $\Gamma_{11}=\Gamma_{22}<1$. For  $\Omega:=1-\Gamma_{11}+\gamma_{1,0}+(1-\alpha)\gamma_{3,0}$,  computing $\mathbb{E}_{t}(F(\theta _{1},0)\Gamma
	-G(\theta _{1},0))\tilde{\lambda}_{t}+\tilde{\mu}_{t})=0$ for the reduced system and solving for the equilibrium distortions, leads to the following result: \small
	\begin{equation*}
	\left[ 
	\begin{array}{c}
	\tilde{\lambda}_{I,t} \\ 
	\tilde{\lambda}_{K,t+1} \\ 
	\end{array}%
	\right] =-\frac{1}{\Omega }\left[\begin{array}{c}
	\tilde{\mu}_{I,t}\\ 
	\tilde{\mu}_{I,t} \\
	\end{array}%
	\right]
	\end{equation*}\normalsize
  Although trivial, this example illustrates how cross equation restrictions determine the sign of all distortions, can then be utilized to construct moment inequalities.
\end{example} \vspace{-0.3 cm}
It is useful to note that in certain situations one may be interested in characterizing
frictions over time (and not just on average), and therefore we need to obtain a
conditional model that generates $\lambda _{t}$. This requires imposing more
restrictions on the stochastic behavior of $\lambda _{t}$. In the supplementary material, we adopt a model uncertainty approach in the spirit of \citet{HansenSargent2005,Hansen_nobel} and we show that 
restricting the class of distributions for aggregate shocks is enough to obtain a unique conditional model for $\lambda _{t}$. For the rest of this paper we focus on the identification using a set of unconditional moment restrictions.

\section{Identification And Estimation} \vspace{-0.4 cm}
In this section we
provide a formal treatment of identification in a linear(ized) Dynamic
Stochastic General Equilibrium (DSGE) model based on moment inequality restrictions.
First, we illustrate how our statistical representation of a DSGE model
relates to the state space representation that is typically used for
estimation. Building on \citet{ECTA1171}, we will show necessary and
sufficient conditions for partial identification of the model arising from
the theoretical moment inequalities. We also show conditions under which \textit{conditionally} over-identifying inequalities provide a more informative (smaller) identified set.

We base the analysis on the innovation representation of the solution to the
linear(ized) DSGE model. This is the natural representation to use when
there are differences in information between economic agents and the
econometrician, as it takes into account that not all the state variables
relevant the decision of agents are observable. We consider the following
class of models: 
\begin{eqnarray*}
\hat{X}_{t+1|t} &=&\underset{n_{X}\times n_{X}}{A(\theta )}\hat{X}_{t|t-1}+%
\underset{n_{X}\times n_{X}}{K_{t}(\theta )}a_{t} \\
Y_{t}^{o} &=&\underset{n_{Y}\times n_{x}}{C(\theta )}\hat{X}_{t|t-1}+a_{t}
\end{eqnarray*}%
where $K_{t}(\theta )$ is the Kalman gain and $a_{t}$ is the one-step ahead
forecast error which could be derived from the corresponding state space representation: 
\begin{eqnarray*}
X_{t+1} &=&\underset{n_{X}\times n_{X}}{A(\theta )}X_{t}+\underset{%
n_{X}\times n_{X}}{B(\theta )}\epsilon _{t+1} \\
Y_{t} &=&\underset{n_{Y}\times n_{X}}{C(\theta )}X_{t}+\lambda _{t}=\underset%
{n_{Y}\times n_{X}}{C(\theta )}\underset{n_{X}\times n_{X}}{A(\theta )}%
X_{t-1}+\underset{n_{Y}\times n_{X}}{C(\theta )}\underset{n_{X}\times n_{X}}{%
B(\theta )}\epsilon _{t}+\lambda _{t} \\
&=&\underset{n_{Y}\times n_{X}}{\tilde{C}(\theta )}X_{t-1}+\underset{%
n_{Y}\times n_{X}}{D(\theta )}\epsilon _{t}+\lambda _{t}
\end{eqnarray*}
where $\epsilon _{t}$ is the innovation to the shock vector $Z_{t}$.  

By construction, 
\begin{eqnarray*}
a_{t}&=&\lambda _{t}+C(\theta )A(\theta )({X}_{t-1}-\hat{X}_{t-1|t-1})+C(%
\theta )B(\theta )\epsilon _{t}
\end{eqnarray*}
Therefore, the forecast error is a combination of the true aggregate innovations to
the information sets of the agents, $\epsilon _{t}$, the estimation error of
the state variable, ${X}_{t-1}-\hat{X}_{t-1,t-1}$, and the frictions, $%
\lambda _{t}$. Let $N(\theta )\equiv vec(A(\theta )^{\prime },B(\theta
)^{\prime },C(\theta )^{\prime })$, and assume that $\mathbb{E}(\epsilon
_{t}|\sigma (\mathcal{F}_{t}))=0$, and $\mathbb{E}(\epsilon _{t}\epsilon
_{s}^{\prime }|\sigma (\mathcal{F}_{t}))=\mathbf{1}(s=t)\Sigma _{\epsilon
_{t}}$, where $\Sigma _{\epsilon _{t}}\succ 0$.
Given $\mathbb{E}(\lambda _{t}|\sigma (\mathcal{F}_{t-1}))\geq 0$, we define
the following conditional moment restriction: 
\begin{equation}
\mathbb{E}(Y_{t}^{o}-\underset{n_{Y}\times n_{x}}{C(\theta )}\hat{X}%
_{t|t-1}|\sigma (\mathcal{F}_{t-1}))\geq 0 \label{eq:mom_ineq}
\end{equation}%
 For any inequality preserving function $\phi (.)$ of a random vector $\mathbf{Y}_{t-1}$ that
belongs to the information set of the econometrician, the following holds: 
\begin{equation*}
\mathbb{E}(Y_{t}^{o}-\underset{n_{Y}\times n_{x}}{C(\theta )}\hat{X}%
_{t|t-1})\phi (\mathbf{Y}_{t-1})=\mathbb{E}\mathcal{V}(\mathbf{Y}_{t-1})\phi
(\mathbf{Y}_{t-1})\geq 0
\end{equation*}%
for a random function $\mathcal{V}(\mathbf{Y}_{t-1})\in \lbrack 0,\infty ]$.

In order to study identification through estimating equations we need to
make assumptions about the local identification of $\Theta _{0},$ given the
value of $\mathbb{E}\mbox{\ensuremath{\mathcal{V}}(\ensuremath{%
\mathbf{Y}_{t-1}})\ensuremath{\phi}(\ensuremath{\mathbf{Y}_{t-1}})}$. We
resort to sufficient conditions that make the mapping from $\theta $ to the
solution of the model regular, and thus assume away population
identification problems (see, for example \citet{CanovaSala}). We assume
that $\Theta $ belongs to a compact subset of $\mathbb{R}^{n_{\theta }}$.
Since certain parameters are naturally restricted, e.g. discount factors,
persistence parameters or fractions of the population, and others cannot take
excessively high or low values, assuming compactness is innocuous. We also
need to acknowledge that due to \textsl{cross - equation restrictions},
which we denote by $\mathfrak{L}(\theta )=0$, the number of observables used
in the estimation need not be equal to the cardinality of $\Theta $, i.e., $%
n_{y}<n_{\theta }$ . \citet{ECTA1171} provide the necessary and sufficient
conditions for local identification of the DSGE model from the
auto-covariances of the data. We reproduce them below, with the minor
modification that Assumption \textbf{LCI-6} holds for any element of the
identified set $\Theta _{0}$.
\begin{defn*}{\textbf{\textsl{ASSUMPTION -LCI (Local Conditional Identification)}}}
	\normalsize
\begin{enumerate}
\item $\Theta$ is compact and connected
\item (Stability) For any $\theta\in\Theta$ and for any $z\in\mathbb{C}$, $%
det(zI_{n_{X}}-A(\theta))=0$, implies $|z|<1$
\item For any $\theta\in\Theta$, $D(\theta)\Sigma_{e}D(\theta)^{\prime }$ is
non-singular
\item For any $\theta\in\Theta$, (i) The matrix $(K(\theta)\,
A(\theta)K(\theta)\,..,A(\theta)^{n_{X}-1}K(\theta))$ has full row rank and $%
(C(\theta)^{\prime }\, A(\theta)^{\prime }C(\theta)^{\prime
}\,..,A(\theta)^{\prime n_{X}-1}C(\theta)^{\prime })^{\prime }$ has full
column rank.
\item For any $\theta\in\Theta$, the mapping $N:\theta\mapsto N(\theta)$ is
continuously differentiable
\item Rank of $\Delta ^{NS}(\theta )$ \footnote{ \small For $\delta ^{NS}(\theta ,T) =(vec(TA(\theta )T^{-1})^{T},vec(TK(\theta))^{T},vec(C(\theta )T^{-1})^{T},vech(\Sigma _{\alpha }(\theta ))^{T})^{T}$, $\Delta ^{NS}(\theta _{0}) =(\frac{\partial \delta ^{NS}(\theta ,I_{n_{x}})}{\partial \theta },\frac{\partial \delta ^{NS}(\theta ,I_{n_{x}})}{\partial vecT})|_{\theta =\theta _{0}}$} is constant in a neighborhood of $%
\theta _{0}\in \Theta _{I}$ and is equal to $n_{\theta }+n_{x}^{2}$ 
\end{enumerate}
\normalsize
\end{defn*}
\begin{lem}
Given $\mathcal{V}_{i}(.)\in \lbrack \underbar{\ensuremath{\mathcal{V}}}(.),%
\bar{\mathcal{V}}(.)]$, and Assumption \textbf{LCI}, $\theta $ is locally
conditionally identified at any $\theta _{0}$ in $\Theta _{I}$ from the
auto-covariances of $Y_{t}$. 
\end{lem}
\begin{proof}
See Appendix
\end{proof}
Given the maintained assumptions, we next characterize the identified set implied by the restrictions using macroeconomic data. \bigskip

\subsubsection{\textbf{Characterization of the Identified Set using Macroeconomic Data}}
As already shown, the identified set is defined by the conditional moment inequalities in \ref{eq:mom_ineq},
\begin{eqnarray*}
\Theta_{I}&\equiv&\left\{\theta\in\Theta: \mathbb{E}_{t}(Y_{t}^{o}-{C(\theta )}\hat{X}_{t|t-1}|\sigma (\mathcal{F}_{t-1}))\geq 0\right\}
\end{eqnarray*}

In all of the identification results we assume the existence of appropriate
random variables to construct unconditional moment restrictions from the
conditional moment inequalities. Such instruments can be either past data or past state
variables constructed with the Kalman filter. By construction, the latter
are uncorrelated with current information, but they might be noisy. In the case of additional moment conditions, a.k.a \textit{supernumerary}\footnote{We borrow this term from \citet{ECTA:ECTA1218}.}, we need to show conditions under which
they can further reduce the identified set, as $\Theta_{I}'$ will now satisfy multiple moment inequalities. Although the work of  \citet{ECTA:ECTA1218} deals with linear models, we cannot use their results 
here for several reasons. First, the characterization of the identified set is
done through the support function of $\Theta _{I}$\footnote{%
The support function, which is $\sup (q^{T}\Theta ),\forall q\in \mathbb{R}%
^{n_{\theta }}$, can fully characterize any convex set $\Theta $.}, which
requires the identified set to be convex\footnote{One can alternatively work with the Aumann expectation of the non-convex random set, which is always convex. Nevertheless, this requires minimizing the support function of the random set for any value of $\theta$, which is costly in high dimensional settings like ours \citep{doi:10.1146/annurev-economics-063016-103658}.}. In our case, since the stability conditions and
the cross equation restrictions introduce on $\mathfrak{L}%
(\theta )$ and therefore nonlinear restrictions on $\Theta$, the identified
set is not necessarily convex and its geometry is difficult to be known a priori. Second, we deal with moment
conditions of general form and therefore additional conditions arising from more instruments in an IV setting is just a special case. We nevertheless also provide an adapted \textit{Sargan} condition. Third, point identification cannot in general occur despite the moment inequalities as we deal with typically continuously distributed random variables.

Recall that the number of moment conditions we use for estimation depend on
the number of observables. Assumption \textbf{LCI-6} requires that there has
to be enough (or the right kind) of observables such that a rank condition
is satisfied. In our case, the number of observables used determines the number
of first order conditions used for estimation. The minimum number ($r$) of
observables required such that conditional identification is achieved (
Lemma 4 is satisfied) maps to the necessary first order conditions. For
example, if we have $Y_{1}$ and $Y_{2}$ to estimate the model, and we only
need $Y_{1}$ to conditionally identify $\theta $, then the $n_{\theta
}\times 1$ first order conditions arising from $Y_{1}$ will be the necessary
conditions. The rest of the conditions, i.e. those arising from $Y_{2}$ are then supernumerary.

Notice that given the structure of the class of models we consider in this paper, it is straightforward to find a re-parameterization that restores moment equalities, that is, for every observable $Y^{o}_t$, $\exists U_{t}:\mathbb{E}U_t\phi(\mathbf{Y_{t-1}})\equiv U\in\mathbb{R}^{+}$ that satisfies the restriction $\mathbb{E}(Y_{t}^{o}-{C(\theta )^{*}}\hat{X}_{t|t-1})\phi(\mathbf{Y_{t-1}})-U=0$.  Then, for $\mathbf{U}\equiv(U_1,U_2..U_r)'$, the following set determines the map from $U_{I}$ to $\Theta_{I}$:  
$  $ \begin{eqnarray}
 (\Theta_{I},\mathbf{U}_{I})&:=&\left\{(\theta,u)\in (\Theta,\underset{{r}}{\bigtimes}\mathbb{R}^{+}):\mathbb{E}(Y_{t}^{o}-{C(\theta )^{*}}\hat{X}_{t|t-1})\phi(\mathbf{Y_{t-1}})-U=0\right\} \label{eq:joint}
 \end{eqnarray}

\begin{prop} {\textbf{$\Theta_{I}$ when $n_{y}=r$}}\\
 Given correct specification of $M_{f}$, Lemma 4 and optimal instrument $\phi^{*}(\mathbf{Y_{t-1}})$, definition \ref{eq:joint} implies the following:
  \begin{itemize}
  	\item  $\exists!$ invertible mapping $\mathcal{G}:\Theta_{I}=\mathcal{G}^{-1}(U)\cap \Theta$
  	\item $\Theta_{I}$ is sharp 
  \end{itemize}
\end{prop}
\begin{proof} See Appendix
	\end{proof}

Note that the proposition applies to the infeasible case of using optimal instruments. If the case of non-optimal but valid instruments, the set is again as sharp as possible, up to the information loss implied by the non-optimality of the instrument\footnote{Recent work in the literature proposes constructing instrument functions to avoid this information loss, see for
	example \citet{andrews2013inference}. We nevertheless do not pursue this in
	this paper.}. It is also important to note that our characterization is based on the fact that the unobserved endogenous variables are integrated out using the Kalman filter. This implies that the only unobservables we need to deal with in the characterization of $\Theta_{I}$ is the aggregate shocks in the economy. Moreover, the set is robust to individual heterogeneity as long as it vanishes on aggregate. More importantly, given correct specification, $U$ is non zero if and only if there is a positive mass of agents that indeed face frictions.  

Having established the determination of $\Theta_{I}$ using the minimum number of observables, we turn to the case of using a larger number of observables, including non-macroeconomic variables.
Let $m_{\alpha ,t}(\theta )$ and $m_{\beta ,t}(\theta )$ denote the
necessary and supernumerary moment functions for identifying $\Theta $,
where $\mathbb{E}(m_{\alpha ,t}(\theta )|\mathbf{Y}_{t-1})=\mathcal{V}%
_{\alpha }(\mathbf{Y}_{t-1})\in \lbrack \underbar{\ensuremath{\mathcal{V}}}%
_{\alpha }(\mathbf{Y}_{t-1}),\bar{\mathcal{V}}_{\alpha }(\mathbf{Y}_{t-1}))$
and $\mathbb{E}(m_{\beta ,t}(\theta )|\mathbf{Y}_{t-1})=\mathcal{V}_{\beta }(%
\mathbf{Y}_{t-1})\in \lbrack \underbar{\ensuremath{\mathcal{V}}}_{\beta }(%
\mathbf{Y}_{t-1}),\bar{\mathcal{V}}_{\beta }(\mathbf{Y}_{t-1}))$. For notational brevity, we drop the dependence on $\mathbf{Y}_{t-1}$ hereafter.

Comparing these general bounds to the ones implied by the model equilibrium restrictions, $\underbar{\ensuremath{\mathcal{V}}}_{\alpha }=\underbar{\ensuremath{\mathcal{V}}}_{\beta }=0$
and $\bar{\mathcal{V}}_{\alpha }=\bar{\mathcal{V}}_{\beta
}=\infty $. Due to the boundedness of $\Theta $ and the
cross-equation and stability restrictions, the effective lower and upper bounds
are likely to lie strictly within $[0,\infty)$ for every moment condition.
Let $\phi (.)$ be any $\mathbf{Y_{t-1}}-$ measurable function for which $%
\hat{m}_{\alpha ,t}(\theta ):=m_{\alpha ,t}(\theta )\phi_{t} $
and $\hat{m}_{\beta ,t}(\theta ):=m_{\beta ,t}(\theta )\phi_{t}$, $\mathbf{\hat{m}}_{\alpha }(\theta )$ and $\mathbf{\hat{m}}_{\beta
}(\theta )$ the corresponding vectors, and $\mathbf{\bar{m}}_{\alpha
}(\theta )$ and $\mathbf{\bar{m}}_{\beta }(\theta )$ the vector means.
 
Let $W$ be a real
valued, possibly random weighting matrix, diagonal in $(W_{\alpha },W_{\beta
})$. Furthermore, denote by $Q_{\alpha }$ the conditional expectation operator when conditioning on $W_{\alpha
}^{\frac{1}{2}T}\mathbf{\hat{m}}_{\alpha }(\theta )$ and by  $Q_{\alpha}^{\bot}$ the residual. 

The following proposition specifies the additional restrictions that need to be satisfied by the additional moment conditions such that the size of the identified set becomes smaller. This is a general result, and it applies both to conventional cases where $n_{Y}>r$ or when data on non-macroeconomic variables is used. 

\begin{prop}
\textbf{The Identified Set with Multiple Conditions}
Given
\begin{enumerate}
	\item $\Theta_{I}\neq \emptyset $	
	\item $U_t\in [\underbar{U},\bar{U}] $, with $\underbar{U}\equiv(W_{\alpha}^{\frac{1}{2}T}\underbar{\ensuremath{\mathcal{V}}}_\alpha +W_{\beta}^{\frac{1}{2}T} \underbar{\ensuremath{\mathcal{V}}}_\beta )\phi_{t} $, $\bar{U}\equiv(W_{\alpha}^{\frac{1}{2}T}\bar{\mathcal{V}}_{\alpha _{t}}+W_{\beta}^{\frac{1}{2}T}\bar{\mathcal{V} }_\beta )\phi_{t}$
	\end{enumerate}

Then $\Theta^{\prime }_{I}\subset\Theta_{I}$ \text{iff}
\begin{equation}
	\mathbb{E} Q_{\alpha}^{\bot}\left(W_{\beta}^{\frac{1}{2}T}\hat{m}%
	_{\beta,t}(\theta)-\hat{U}_t\right)  =  0  \label{eq:sargan}
\end{equation}
\end{prop}
\begin{proof}
See the Appendix
\end{proof}
The main argument behind Proposition 5 is the following. Suppose that the
necessary moment conditions have no common information with the
supernumerary conditions and that $W=I_{n_{\alpha }+n_{\beta }}$. From the
minimization of $\mathbb{E}\frac{1}{2}(\mathbf{\bar{m}}-\bar{U}_{t})^{T}(%
\mathbf{\bar{m}}(\theta )-\bar{U}_{t})$ where $\mathbf{\bar{m}}(\theta
)\equiv (\mathbf{\bar{m}}_{\alpha }(\theta ),\mathbf{\bar{m}}_{\beta
}(\theta ))^{T}$, the first order condition is 
\begin{equation*}
\mathbb{E}(\mathbf{\bar{m}}_{\alpha }(\theta )+\mathbf{\bar{m}}_{\beta
}(\theta )-\bar{U}_{t})=0
\end{equation*}%
which for $\bar{U}_{\alpha }\equiv\overline{Q_{\alpha }^{\bot }U_{t}}$
 can be rewritten as $\mathbb{E}((\mathbf{\bar{m}}_{\alpha }(\theta )-\bar{U}_{\alpha })+(\mathbf{\bar{m}}_{\beta }(\theta )-\bar{U}^{\bot}_{\alpha }))=0$.
By construction the two parts of the left hand side of the expression are
independent, and therefore both have to be zero. 
\begin{eqnarray}
\mathbb{E}(\mathbf{\bar{m}}_{\alpha }(\theta )-\bar{U}_{\alpha }%
)=0 &&  \label{eq:ma} \\
\mathbb{E}(\mathbf{\bar{m}}_{\beta }(\theta )-\bar{U}^{\bot}_{\alpha })=0 &&
\end{eqnarray}%
Notice that, by construction, the set of necessary moment conditions in 5.3.
must have full rank, and this establishes a one-to-one mapping from $\Theta $
to the domain of variation of $U_{t}$, $[\underbar{U}(\mathbf{Y}_{t-1}),\bar{%
U}(\mathbf{Y}_{t-1})]$. Thus, there exists an inverse operator $\mathcal{G}%
_{\alpha }$ such that $\theta =\mathcal{G}_{\alpha }(U_{t},\mathbb{P})$.
Plugging this expression for $\theta $ in 5.4, we get that 
\begin{equation*}
\mathbb{E}\mathbf{\bar{m}}_{\beta }(\mathcal{G}_{\alpha }(U_{t},\mathbb{P}))=%
\mathbb{E}\bar{U}^{\bot }_{\alpha }
\end{equation*}
This is a restriction on the values that $U_{t}$ can take in addition to the
ones implied by the necessary conditions. A restriction on $U_{t}$ implies a
restriction on the admissible $\Theta _{I}$ given the one-to-one
relationship in 5.3. Notice that when the supernumerary conditions do not
add any additional information, i.e. $m_{\alpha ,t}(\theta )\equiv m_{\beta
,t}(\theta )$, the restriction collapses to $Q_{\alpha }=Q_{\alpha }^{\bot }=%
\frac{1}{2}$.
Given these identification results, we can analyze identification arising from any inequality restriction, and therefore any additional type of information can be potentially analyzed. Below, we discuss and formally show how qualitative survey data can provide additional restrictions that are informative about aggregative models of economies with frictions. Such information constrains further the stochastic properties of $\lambda_{t}$, and therefore the size of $\Theta_{I}$. 

\section{The use of additional information}

Additional data can be potentially linked to dynamic equilibrium models by augmenting the observation equation 
with moment restrictions. The latter can be motivated by the fact that not all data are explicitly modeled using structural equations, but they can nevertheless be used to sharpen economic and econometric inference, especially for unobserved processes. An example of such a process is the distortion $\lambda_{t}$.  We focus on qualitative survey data because, as we illustrate, they contain distributional information that can be linked to the aggregate model and are informative about  $\lambda_{t}$. This is especially true for economies with frictions, as the proportion of agents whose behavior is  distorted is an important statistic. An example of a distributional statistic that has been routinely used to judge the validity of \textit{complete} models is the proportion of firms which cannot change prices in New Keynesian models featuring a Calvo adjustment mechanism. Our focus is on more general information that is contained in qualitative surveys, and given that we deal with \textit{incomplete}  models, this information is able to discard a set of economic mechanisms, and not just one. We thus view our method as a significant generalization of the treatment of microeconomic information in dynamic equilibrium models.

Qualitative survey data are usually available in the form of aggregate statistics,
where aggregation is performed over categories of answers to particular
questions. As in micro-econometric studies, the categorical variable is a function of a
continuous latent variable. In a structural context, there is a measurable mapping from the categories of answers to
the random variables relevant to the decision of each agent. 
For example, if the question is of the type "How do you expect
your financial situation to change over the next quarter" and the answer is
trichotomous, i.e "Better", "Same" and "Worse", then the answers map to a set of
partitions of the end of period assets $a_{t+2}$: $a_{t+2}\in
\lbrack a_{t+1}-\epsilon ,\infty )$, $a_{t+2}\in (a_{t+1}-\epsilon
,a_{t+1}+\epsilon )$ and $a_{t+2}\in (-\infty ,a_{t+1}-\epsilon ]$.
For this interpretation to hold, we need
to assume that agents report their states or beliefs truthfully. Furthermore, denote by $\{S_{i,k,t}\}_{i\leq N}$ the survey sample over a period of
length $T,$ for the $i$th respondent. Let $C_{l}^{k}$ be the $l^{\text{th}}
$ categorical answer to question $k$ and $\hat{\xi}_{i,t,k}$ the
respondent's choice. Given some weights on each category $w_{l}$, the available statistics are of the form: 
$
\hat{\mathcal{B}}_{t}^{k}=\sum_{l\leq L}w_{l}\sum_{i\leq N}1(\hat{\xi}%
_{i,t}\in C^{k})
$.  Given truth telling, we can map the answer to agent
beliefs i.e. there exists a mapping $h:\hat{\xi}_{i,t}\rightarrow \xi _{i,t}\equiv \mathcal{E}%
_{i,t}(x_{i,t+1}|x_{i,t},X_{t})$. Moreover,  since the conditional
expectation is a function of the information set, let ${B}_{l}\equiv \{(x_{i,t},X_{t})\in \mathbb{R}%
^{2n_{x}}:h(x_{i,t},X_{t})\in C_{l}^{k}\}$, that is, ${B}_{l}$ belongs to the
partition of the support of individual $x_{i,t}$ (or aggregate $X_{t})$ that
corresponds to category $C_{l}^{k}$. Consequently, the survey statistic has the following theoretical form, which can be linked to the model: 
\begin{eqnarray}\hat{\mathcal{B}}_{t}^{k}  &=& \sum_{l\leq L}w_{l}\sum_{i\leq N}w_{i}1(\mathcal{E}_{i,t}(x_{t+1,i}|x_{i,t},X_{t})\in B_{l}) \label{eq:repr_survey}\end{eqnarray}

For every $Y_{t}^{o}$ with a corresponding  model based conditional expectation $Y_{t}^{m}\equiv \mathbb{E}(Y_{t}^{o}|M,\mathcal{F}_{t})$,  Proposition 6 illustrates the additional restrictions implied by $\hat{\mathcal{B}}_{t}^{k} $. For these to hold, we need to make the following assumptions: 
\begin{defn*}{\textbf{\textsl{ASSUMPTION -S}}}
	\begin{enumerate} \small \item Let $\hat{%
			\mathcal{B}}_{t}^{k}$ defined as in \eqref{eq:repr_survey} and an aggregate
		shock vector $Z_{t}$ of length $p>0$, with innovation $\epsilon _{t}|\mathcal{F}_{t-1}\sim (0,\Sigma_{e})$.
		\item  $\lambda_{i,t}$ is conditionally linear.
		\item Truthful response such that $\hat{\mathcal{B}}_{t}^{k}\to B_{t}\in (0,1]$.
		\item $\hat{R}^{m}_{t}:=\left(\frac{\mathbb{E}(x^{m}_{i,t-1}|x_{i,t}\in\mathcal{X}^{\star,C}_{t})}{\mathbb{E}(x^{m}_{i,t-1}|x_{i,t}\in\mathcal{X}^{\star}_{t})}-1\right)\mathbb{P}(x_{i,t}\in \mathcal{X}^{\star,C}_{t})$ is monotone in  $\mathcal{X}^{\star,C}_{t}$ for all $m=1..n_{x}$.
	\end{enumerate}
\end{defn*} \normalsize

The assumption of conditional linearity for $\lambda_{i,t}$ enables us to claim that once the distortion is activated, it is a linear function of the states and the shocks, e.g. $\lambda_{i,t}:=\lambda_{1}x_{i,t-1}+\lambda_{2}z_{i,t}$. We motivate assumption S-4 after presenting the following proposition:
\bigskip
\begin{prop} Under Assumption-$\mathbf{S}$, the following restriction holds, $\mathbb{P}_{X_{t-1}}-a.s.$ 
	\begin{eqnarray*}
		\mathbb{E}(Y_{t}^{o}|X_{t-1})& \leq & P(\theta_{1},0)X_{t-1}+  \lambda^{1}X_{t-1}\odot \mathbb{E}( B_{t}|X_{t-1})
	\end{eqnarray*}	
where $P(\theta_{1},0)\equiv C(\theta)$.
\end{prop}
\begin{proof}
	See Appendix A
\end{proof}

 Notice that in the proof, we define  $(\mathcal{X}^{\star}_{t},\mathcal{Z}^{\star}_{t})$ as the time dependent subset of the support such that $\lambda_{i,t}\neq 0$. In models of heterogeneous agents these boundaries usually depend on aggregate states and shocks. Therefore, this inequality will in general be strict unless $B_{t}=0$.

Moreover, in the proof we establish two
facts. First, since the equilibrium conditions of the model with frictions
depend on subjective conditional expectations, and the model with frictions
is a smooth perturbation of the frictionless model, any probability
statement on the subjective expectations translates to a probability
statement on $\mu_{i,t}$. Second, any probability statement on $\mu _{i,t}$
is a probability statement on the solution of the model, and therefore on $%
\lambda _{i,t}$. Given representation \eqref{eq:repr_survey}, qualitative
survey data have information on the quantiles of subjective conditional
expectations of the agents. Therefore, survey data relate directly to the
conditional probability of observing a friction, $\mathbb{P}_{t}(\lambda _{i,t}\geq 0)$, which implies the above restrictions. The latter cannot in general deliver point identification, as the vector $\hat{\mathbf{R}}_{t}=(\hat{R}^{1}_{t},\hat{R}^{2}_{t}...\hat{R}^{n_{x}}_{t})$ is hard to pin down unless specific assumptions are made. Monotonicity, which is a mild assumption, provides a bound. 

Given Proposition 6,  we can also establish the following corollary result regarding identification using survey data. 
\begin{cor}{Identification with Survey Data }
\begin{enumerate}
	\item If $\hat{%
		\mathcal{B}}_{t}^{k}\to B_{t}\in(0,1)$, $\Theta_{I}'\subset\Theta_{I}$.
	\item Impossibility of point identification: When $B_{t}\neq 0$, $\Theta_{I}$ is not a singleton
\end{enumerate}
\end{cor}
\begin{proof}
	See Appendix
\end{proof}

Corollary 7 has several implications. First, the set of admissible structures becomes smaller, and we can therefore make more precise statements regarding parameters and conditional predictions. Second, using identified set and the definition of $\lambda (Y_{t},\theta )$, a plug-in set estimate of the average
distortion (wedge) in a macroeconomic variable $Y_{t}$ is $\mathbb{E}\lambda(Y_{t},\theta _{I}(Y))$. These estimates of distortions can be used in several ways, which we explore in the next section.  

\section{Inference based on $\lambda(\Theta_{I})$}  
 We first discuss briefly how $\Theta_{I}$ and $\lambda(\Theta_{I})$  can be obtained in practice. Given the quasi-structural framework,  performing inference for $\Theta_{I}$ rather than $\theta_{0}$ seems a natural thing to do. Many of the parameters in dynamic macroeconomic models are semi-structural, and therefore $\theta_{0}$ itself does not have a very specific economic interpretation. We therefore use methods appropriate for constructing consistent estimators for the identified set (IdS) and confidence sets (CS) for the IdS. In the frequentist literature, several methods have been proposed for a general criterion function like subsampling \citet{CHT,ECTA:ECTA998} or the bootstrap for moment inequality models \citet{ECTA:ECTA1025}. Nevertheless, in macroeconomic models parameter dimensions are high and pointwise testing is a vastly inefficient way to construct  confidence sets, so we focus on computationally tractable methods to perform statistical inference using Markov Chain Monte Carlo algorithms (MCMC). 
 
When $\theta_{0}$ is point identified and root-$n$ estimable, \citet{Chernozhukov2003293} have proposed to use simulation methods to do inference on $\theta$ using  the quasi-posterior distribution, which is constructed using the relevant loss function $L_{n}(\theta)$ that defines $\theta_{0}$. Generalized
Method of Moment (GMM) class of estimators can be easily embedded
in this case. Denoting by $\pi(\theta)$ the prior distribution, simulation draws $\left\{\theta^{j}\right\}_{j\leq L}$ are obtained from  \vspace{-0.2 cm} \small \begin{eqnarray*} \Pi_{n}(\theta|\mathbf{Y}) &\equiv& \frac{\exp({TL_{n}(\theta)})}{\int_{\Theta}\exp{(TL_{n}(\theta))}d\pi(\theta)} \end{eqnarray*} \normalsize and upper and lower $100(1-\alpha)/2$ quantiles are used to conduct inference that has valid frequentist properties.  Nevertheless, once point identification fails, one has to consider adjusting the method to accommodate for this. 

More particularly, let $L_{n}(\theta)=n^{\frac{1}{2}}q_{n}(\theta)_{+}'W_{n}n^{\frac{1}{2}}q_{n}(\theta)_{+}$
be the criterion function to be minimized, where $q_{n}(\theta)$
are the moment functions to be used and $q_{n}(\theta)_{+}\equiv\max(q_{n}(\theta),0)=\min(-q_{n}(\theta),0)$. If $L_{n}(\theta)$  is stochastically equicontinuous, that is, there exists a $\Delta_{n}(\theta_{0})$ and $J_{n}(\theta_{0})$
	such that $L_{n}(\theta)$ admits a quadratic expansion \footnote{\begin{eqnarray*}
		L_{n}(\theta) & =L_{n}(\theta_{0})+ & (\theta-\theta_{0})'\Delta_{n}(\theta_{0})-\frac{1}{2}(\theta-\theta_{0})'nJ{}_{n}(\theta_{0})(\theta-\theta_{0})+R_{n}(\theta)
\end{eqnarray*}}  
then this allows us to assume a Central Limit Theorem (CLT) on $\Delta_{n}(\theta_{0})$. Redefining $Q_{n}(\theta)\equiv q_{n}(\theta)-b$
where $b$ is the bias term and $\tilde{L}_{n}(\theta,b)=n^{\frac{1}{2}}Q_{n}(\theta)'W_{n}n^{\frac{1}{2}}Q_{n}(\theta)$, we assume that the following CLT holds \vspace{-0.2 cm}
\begin{eqnarray}
V(\theta_{0},b_{0})^{-\frac{1}{2}}\tilde{\Delta}_{n}(\theta_{0},b_{0}) & \overset{d}{\to} & N(0,I) \label{eq:stoche}
\end{eqnarray} \vspace{-0.2 cm}
Estimating equations arising in DSGE models involve smooth functions. This is particularly true if we focus on determinate equilibria, so that model specification is uniform across the parameter space. Combining them with smooth instrument functions is enough to guarantee stochastic equicontinuity. 

Given \ref{eq:stoche}, \citet{liao2010} show that when using a cutoff rate $\nu_{n}$ such that $ 1\prec\nu_{n} \prec n$, and defining \small $\mathcal{A}_{n}:=\left\{ \max_{\vartheta} ln(\Pi_{n}(\vartheta|\mathbf{Y}))- ln(\Pi_{n}(\theta|\mathbf{Y})) \leq  \nu_{n} \right\}$ \normalsize then $d_{H}(\mathcal{A}_{n},\Theta_{I})\to_{p} 0$ where $d_{H}(.,.)$ is the Hausdorff distance between two sets \footnote{$d_{H}(.,.)=\max\left[\sup_{a\in A}\inf_{b\in B}||b-a||,\sup_{b\in B}\inf_{a\in A}||a-b||\right]$}. Moreover, the rate of convergence of the pseudo-posterior density outside the identified region is exponential. As is evident, in finite samples the identified will depend on $\nu_{n}$, so some robustness checks are required. In addition, since our criterion function vanishes  when $\theta\in\Theta_{I}$, we do not need to adjust $\nu_{n} $ as suggested in \citet{CHT}. 

With regards to constructing confidence sets for the identified set, which is what we do in the empirical application, recent work by \cite{mcmcids} provides a computationally attractive procedure to construct  CS for the IdS and functions of it that have correct coverage from a frequentist perspective. Compared to \citet{Chernozhukov2003293, liao2010}, cutoff values are based on quantiles of draws of the loss function $L_{n}(\theta)$ rather than draws from the quasi-posterior of the parameter vector and there is no maximization involved.  Also, results extend to models with singularities, that is models in which a local quadratic approximation involves a non vanishing singular component, and parameters are not root-$n$ estimable. This is useful in case one wants to derive bounds based on extensions to non-parametric treatments of the survey based bounds.  Although we find the method suggested very appealing, we apply it only in the case of estimating frictions. Proving validity for the testing procedure we will propose in this section is not trivial as it involves combining two independent MCMC chains, and we leave this interesting research for the immediate future. 

We next motivate the proposed test that utilizes $\lambda_{Y}(\Theta_{I})$. When a \textit{complete model} is estimated,  the pseudo true vector $\theta_{0}$ is point identified. Given its value, we can estimate the predicted level of friction in each of the macroeconomic variable,  $\mathbb{E}\lambda_{Y}(\theta_{0})$. Misspecification of this model may imply that the predicted friction does not lie in the identified set of distortions. Therefore, the distance from $\mathbb{E}\lambda_{Y}(\Theta_{I})$ becomes a sufficient statistic to judge whether the suggested model is properly specified. We propose a Wald statistic that tests whether the expected distance from the point
estimate for the wedge from the parametric model to the identified set
of wedges is different than zero for all (or some of) the observables. Since survey data provide more information on the wedge, this increases power to reject non-local alternatives. In addition, survey data regularizes the test, as in the absence of additional data the distribution degenerates and requires non-standard inference. \vspace{-0.2 cm}
\subsubsection{\textbf{ A Test for Parametric Models of Frictions}}
Letting $H_{0}:\theta_{p}\in\Theta_{s}$ and $H_{1}:\theta_{p}\notin\Theta_{s}$, the proposed statistic is as follows:  \vspace{-0.2 cm} \small
\begin{equation*}
\mathcal{W}_{t}=\left( \sqrt{T}\inf_{\lambda _{s}\in \lambda (\hat{\Theta}%
	_{s})}||\mathcal{V}^{-\frac{1}{2}}(\lambda _{s}-\lambda _{p}^{\star
})||\right) ^{2}
\end{equation*} \normalsize
where $\lambda $ is the estimated friction obtained using either the
identified point in the parametric model case, $\lambda _{p}$, or the
identified set in the robust case, $\lambda _{s}$. Individual frictions are
weighted by their respective estimate of standard deviation. The statistic
measures the Euclidean distance between the wedge that arises in the
parametric model, and the set of admissible wedges, adjusted for estimation
uncertainty. 
Under the following conditions, the test is consistent and
has asymptotic power equal to one against fixed alternatives.  \vspace{-0.2 cm}
\begin{defn*}{\textsl{\textbf{ASSUMPTION -R} (Regularity conditions)}}
	Let $D$ and $q$ be $Y-$ measurable functions, continuous in $\theta$ w.p.1
	and $\tilde{q}(.;\theta)\equiv \tilde{q}(.;\theta) - \bar{q}(.;\theta)$ such
	that:
	\begin{enumerate}
		\item For any $\mathbf{Y}$, and any $\theta \in \Theta$, $\sqrt T {\tilde{q}(%
			\mathbf{{Y};\theta)}} \to_d \mathcal{N}(0,\Omega)$
		\item $\sup_{\theta\in\Theta} D_n(\theta) \to_p D$ where $D$ is positive
		definite
	\end{enumerate}
\end{defn*} \vspace{-0.3 cm}
\begin{prop}{Consistency and Power}. Under assumption \textbf{R}:
	\begin{enumerate}
	\item $T\mathcal{W}(\theta_{p},\Theta_{s}) \overset{d}{\to} ||\sum_{j=1..p}\omega_{j}\mathcal{N}(0,I_{p})||^{2}$
    \item  Given critical value $c_\alpha$ with $\alpha\in(0,1)$, 
    	\begin{enumerate}
    	\item Under $H_{0}$: $\underset{T\to\infty}{\lim}\mathbb{P}(T\mathcal{W}(\theta_{p},\Theta_{s})\leq c_{\alpha})= 1-\alpha $
    	\item Under $H_{1}$: $ \underset{T\to\infty}{\lim}\mathbb{P}(T\mathcal{W}(\theta_{p},\Theta_{s})\leq c_{\alpha})=0 $ 
   \end{enumerate}
\end{enumerate}
\end{prop}
\begin{proof}
	See Appendix
\end{proof}
\normalsize
In the supplemental material we show that using the non-parametric block bootstrap to compute $c_{\alpha}$ is valid, and we illustrate through a measurement error example that the bootstrap distribution
coincides with the asymptotic distribution. Thorough examination of the performance of this test in small samples is a very interesting avenue of research, which deserves separate treatment. Another equally interesting topic for future research is to investigate which criterion function would deliver robustness of the testing procedure to possible misspecification of the benchmark model, which we have assumed away. See \citet{ECTJ:ECTJ332} for the issue of misspecification in moment inequality models. 
\vspace{-0.2 cm}

Finally, we briefly discuss the \textit{refutability} of the candidate model, in the sense of \cite{Breusch}. The Null hypothesis, tests whether the point identifying restrictions on $\lambda_{Y}$ implied by the candidate model are satisfied on $\lambda_{Y,s}$, where the latter is only set identified. Since $(\lambda_{Y,s},\lambda_{Y,p})$ is constructed using $\theta_{1}$  as identified using the frictionless and complete model respectively (See Proposition 1 for the definition of $\theta_{1}$), $\lambda_{Y,p}$ is a function of the additional parameters indexing the CM, $\theta_{2}$. In the absence of cross parameter (equation) restrictions on $\Theta_{1}\times\Theta_{2}$, $\lambda_{Y,p}$ would necessarily lie within $\lambda_{Y,s}$ as the latter would be consistent with any $\theta_{2}\in\Theta_{2}$. However, in the presence of cross equation restrictions, $(\theta_{1},\theta_{2})$ lie in a strict subset of $\Theta_{1}\times\Theta_{2}$. Thus, the completion $\lambda_{Y,p}$ may not necessarily lie within the family of completions, $\lambda_{Y,s}$. Equivalently, $\exists \lambda\in\lambda_{s}$ that is not observationally equivalent to $\lambda_{p}$.  In this sense, the candidate model is \textit{refutable} only when it imposes restrictions on $\lambda_{y}:=U_{y}(\theta_1,\theta_2)$ for any observable $Y$, which implies restrictions on the reduced form. 
This is indeed true in the context of incomplete general equilibrium models; taking the unconditional expectation in condition \eqref{eq:repres} of  Proposition 1, implies that ($\Theta_{1},\Theta_{2}$) is not variation free and therefore $U_{y}(\theta_1,\theta_2)$ is restricted. With additional information, the set on which $H_{0}$ and $H_{1}$ are defined is reduced, and therefore the set of alternatives is narrower.

\section{Application: Estimating Frictions in the Spanish Economy}
We apply the methodology to the case of Spain, where financial frictions have arguably played a significant role during the last decade. The benchmark economy we use features \textit{some }frictions and since it
is standard, we will directly introduce the log-linearized conditions. We
consider a small open economy with capital accumulation, along the lines
of \citet{10.1257/aer.97.3.586} and \citet{Gali01072005}. There are
households, intermediate good firms, final good firms, government
expenditure, and a foreign sector which is composed by infinitesimal
symmetric economies.

The type of frictions we allow in the baseline model are those we do not have
sufficiently informative survey data to implement our methodology. Thus, we
keep the parametric Calvo type of friction in the wage setting by labor unions and
in the price setting behavior of firms. All other frictions are
going to be semi-parametrically characterized. We thus remove capital
adjustment costs, and therefore Tobin's q becomes constant. This implies that
the arbitrage condition between capital and bonds has no dynamics in the benchmark model.

Let $X_{1}^{o}$ denote the vector of variables that enter the moment
equalities, $X_{2}^{o}$ the vector of variables used in the moment
inequalities and $Z$ the vector of instruments.   Model predictions are denoted with superscript '$m$'. The conditions we use are therefore:
\begin{eqnarray*}
\mathbb{E}((X_{1,t}^{o}-X_{1,t}^{m})\otimes Z_{t}) &=&0 \\
\mathbb{E}((X_{2,t}^{o}-X_{2,t}^{m})\otimes Z_{t}) &\leq &0
\end{eqnarray*}
The
variables employed in estimation are Non Government Consumption expenditure
(C), Hours (H), Inflation ($\pi $), Investment (I), Gross Domestic Product
(Y), Wages (W), and the EONIA rate (R). Real variables are in per capita terms.W use lagged
values of output and consumption as instruments, as both positively depend on the capital stock which is unobserved. We estimate the model using
two different subsets of survey data from Spain, collected by the European
Commission. Our sample period covers $1999Q_{1}$ to $2013Q_{4}$. Detailed
information on this survey data can be found at: %
\url{http://ec.europa.eu/economy_finance/db_indicators/surveys/index_en.htm}

In the set of survey responses, we include responses to quarterly questions
1,2 and 11 from the Consumer Survey, that relate to the financial position
of the household. We plot the time series in Appendix B. Credit constraints imply
negative distortions to household consumption, and given that hours worked are complementary to consumption,
they also imply negative distortions to labor supply and output.  In addition, we use business survey data, in particular questions 8F4
and 8F6, relating to capital adjustment costs and financial constraints to
production capacity. For the case of capital adjustment costs, we
have restrictions on $Y$ and $I$ similar to those of the general example in Section 3. Financial
constraints to productive capacity imply similar restrictions and lead to
lower aggregate investment and output. We therefore choose $X_{1,t}^{o}\equiv (W,\pi ,R)$ and $%
X_{2,t}^{o}\equiv (Y,I,H,C)$. We estimate the model using the 2-block RW-MCMC with uniform priors for all parameters\footnote{We keep the last 200000-300000 draws for inference.}. Since the survey based moment restrictions require estimating the nuisance parameters $\lambda^{1}$ (see Proposition 6), for a given draw of $\Theta_{1}$, we obtain $\hat{\lambda}^{1}$ by ordinary least squares. This approximation can, in theory, weaken the informativeness of the additional restrictions but we do not observe this in our results. 

\subsection{\textbf{Details of the exercise}} We obtain the $95\%$ confidence sets of wedges,
that is, a range for the standardized wedge in each observable that is consistent with
both macroeconomic and survey data, $\lambda_{y}\equiv (\mathbb{E}\mathbb{V}_{y|\hat{x}}\mathbb{E}\mathbb{V}_{\hat{x}})^{-1}\mathbb{E}\hat{X}_{t,t-1}(Y_{t}-C(\Theta_{I}))$, which under linearity is equivalent to $(\mathbb{E}\mathbb{V}_{y|\hat{x}})^{-1}\lambda_{y}$\footnote{We actually compute $\lambda_{t}$ using quadratic programming as in the second section of the online Appendix and then take the average. The result is equivalent to a standardized wedge.}.  These estimates are the empirically relevant wedges that any model featuring financial frictions and adjustment costs should produce over the whole sample. We plot them (in red) in the left panel of Figures 1-6, in Appendix A.  

We employ a complete model (\textbf{CM}) featuring \citet{Bernanke_Gertler} (BGG) type of frictions (similar to \cite{DEGRAEVE20083415}) and estimate it using the same macroeconomic aggregates as the \text{IM} using full information MCMC, where we also rely on \cite{mcmcids} to obtain $\Theta^{cs}_{CM}$.

 Despite that the model features financial frictions on the firm side (the "entrepreneurial sector"),  and not borrowing constrained consumers, its implications can still be tested against the incomplete model (\textbf{IM}), which is robust to both types of frictions.   

We plot (in blue) the corresponding confidence sets for the wedges that are consistent with the complete model. To facilitate the comparison, we also produce an alternative comparison based on QQ plots, which gives information on the support of each of these estimates. Note that in some of the cases the support is completely different. 

In Table 1 (Appendix A), we display the corresponding confidence sets for $\Theta$ using survey data. Confidence sets for $\Theta_{IM}$ without using survey data and the corresponding wedges can be found in Appendix B.

A key observation is the fact that the sign of the wedge confidence sets that correspond to the restrictions imposed to identify the IM are in accordance with the sign of the distortions identified with the CM. While both the CM and IM identify negative distortions to output, the former identifies a significantly lower level of distortions. As is typical, empirical versions of accelerator type of models ignore the output costs (bankruptcy costs) as they are deemed to be numerically unimportant. Some of the neglected variation is indeed captured by the high estimates of $\phi_{p,CM}$ relative to $\phi_{p,IM}$. 

Regarding inflation, using the CM we cannot reject zero frictions to $\pi_t$, while the IM identifies significantly negative distortions. Moreover, the CM identifies a much higher slope of the Phillips curve\footnote{\small $\kappa:=\frac{(1-\iota_{p}\xi_{p})(1-\xi_{p})}{(1+\beta\iota_{p})(\xi_{p}(1+(\phi_{c}-1)\epsilon_{p}))}$ where $\epsilon_{p}$ (Kimball aggregator) is calibrated to 10. \normalsize}, which is consistent with its "inability" to generate negative distortions to inflation.

With regard to the nominal interest rate, both the CM and the IM identify negative distortions (with no evidence of different magnitudes). Mechanically, this is due to the Taylor rule that pins down $r_t$. Nevertheless, there is also a deeper insight which comes from steady state reasoning. In an economy with uninsured idiosyncratic risk on both firms and consumers, the steady state interest rate is lower than in the case of complete markets, and as shown by \cite{ANGELETOS20071}, this does not contradict negative distortions to capital accumulation due to the presence of a risk premium.  

Furthermore, we cannot reject the hypothesis that the CM is consistent with zero distortions to $c_t$, while the IM identifies negative distortions. Moreover, the estimates of risk aversion are much lower in the IM compared to the CM. One of the reasons is the fact that the CM ignores credit constraints to consumers, and therefore requires much higher estimates in risk aversion to match the corresponding negative distortions to consumption. 

With regard to hours worked, both the IM and the CM identify  negative distortions which are statistically different. Since  $\sigma_{c,IM}$ is much smaller than $\sigma_{c,CM}$, negative distortions to consumption are consistent with less distortions to hours in the IM. Moreover, the CM identifies a slightly higher (inverse) Frisch elasticity ($\sigma_{l,CM}$) and does not reject zero distortions to $w_t$, while the IM identifies positive distortions. The latter observation is not robust to excluding the survey data as shown in Table 2.

Finally, the IM and CM identify the same negative distortions to investment. Since no marginal $\Theta_{cs}$ has an extremely narrow support, this result can be possibly explained by the mapping from parameters to investment that is nearly not injective. This is not a problem per se, as the MCMC method we have used to obtain confidence sets in both the IM and the CM is robust to lack of identification.

The non overlapping standardized output, inflation and labour wedges provide enough evidence to reject the hypothesis that the SW-BGG model implies similar distortions to those identified by the IM. This finding also corroborates existing evidence of the lack of uniform fit for such a model over the entire business cycle ( \cite{DELNEGRO2016391} for the case of the US).

\section{Conclusion}

In this paper we propose a new inferential methodology which is robust to
misspecification of the mechanism generating frictions in dynamic
stochastic economies. We characterize wedges in equilibrium decision rules in a
way which is consistent with a variety of generating mechanisms. We use the implied restrictions to partially identify the
parameters of the benchmark model, and to obtain a set of admissible economic
relationships. Moreover, we formally characterize set identification given the number of observables. Regarding the latter, beyond macroeconomic data, we show how qualitative survey data can provide additional distributional information, which is crucial in economies with ex post heterogeneous and frictions. Survey data provide a sufficient statistic for this information that would otherwise have to be generated by a model. We show how to exploit this additional information to test parametrically specified
models of frictions. We apply our methodology to estimate the distortions in the Spanish economy
due to financial frictions and adjustment costs using an small open economy version of %
\citet{10.1257/aer.97.3.586} model and qualitative survey data collected by
the European Commission. We investigate the adequacy of \citet{Bernanke_Gertler} type of financial frictions. 

In general, our work shows that adopting a robust approach to inference and
using the information present in surveys is a fruitful way of dealing with
lack of knowledge about the exact mechanisms generating frictions. 

Our work is limited to dynamic \textit{linear} economies and has focused on deviations of the representative agent
approximation from the underlying heterogeneous agent economy. Possible extensions could study deviations based on a benchmark heterogeneous model. \citet{buera_moll} have recently shown that the distortions generated in an
aggregate equilibrium condition can depend on the type of heterogeneity
present in the economy. Our approach is robust to this criticism. The
methodology does not rest on \textit{observing} residuals from
representative agent frictionless economies - we just impose moment
inequality restrictions theoretically motivated by deviations from the
frictionless economy. In addition, if heterogeneity has additional
implications for the form of these restrictions, they can be taken on board.
In addition, we impose weak moment inequalities. This is important when
heterogeneous distortions in decision rules cancel out. Finally, we impose
restrictions implied by $\mu _{t}$ on all of the variables and thus take
general equilibrium effects into account. Nevertheless, future work could focus on investigating the
robustness of our methodology in environments where some heterogeneity is
ignored when imposing the identifying restrictions.

Finally, we have assumed that aggregated survey data are not systematically biased and strategy proof. Analysing the consequence for inference would be another fruitful avenue of research.

\newpage
\section{Appendix A}
\begin{proof}
\textbf{of Proposition 1}.

Recall the representation for the model with frictions, that is, 
\begin{eqnarray*}
G(\theta_1,0)X_{t}&=&F(\theta_1,0)\mathbb{E}_{t}(X_{t+1}|X_{t})+L(\theta_1,0)Z_{t}+\tilde{\mu _{t}}
\end{eqnarray*}
 Plugging in the candidate distorted decision rule: $X_{t}^{\ast }=X_{t}^{f,RE}+\tilde{\lambda}_{t}$ and using that $\underset{n_{x}\times n_{x}}{F(\theta_1,0)}\underset{n_{x}\times n_{x}}{P^{\ast }(\theta_1,0)}+\underset{n_{x}\times n_{x}}{G^{\ast }(\theta_1,0)}=0$ and $(\underset{n_{z}\times n_{z}}{R(\theta_1,0)^{T}}\otimes F(\theta_1,0) + I_z\otimes(F(\theta_1,0)P^{\ast }(\theta_1,0)+G^{\ast }(\theta_1,0)))vec(Q(\theta_1,0))=-vec(L(\theta_1,0))$
we have the following condition:
\begin{eqnarray*}
\mathbb{E}_{t}(F(\theta_1,0)\lambda_{t+1}-G(\theta_1,0)\lambda_{t} + \tilde{\mu _{t}})&=&0
\end{eqnarray*}
Note that in this proposition we let the econometician's model variables (observables $Y_t$ and unobservables $Z_t$ coincide with a proper subset of $X_t$ and $Z_t$. That is, setting $\theta_2=0$ essentially eliminates some of the elements of $(X_t,Z_t)$. Furthermore, in the proposition we state that ${\lambda}_{t}={\lambda}_{t-1}\Gamma +\nu_{t} $ for some real-valued  $\Gamma\neq \mathbf{0}$. Substituting for $\lambda_t$ we get the condition stated in Proposition 1, that is:
\begin{eqnarray}
\mathbb{E}_{t}(F(\theta_1,0 )\Gamma-G(\theta_1,0 )){\lambda}_{t}+\mu_{t})&=&0 
\end{eqnarray}
To motivate the assumption on the random variable $\lambda_t$ notice that, the condition above essentially links the conditional mean of $\lambda_t$ with that of $\mu_t$. Since $\mu_t$ is by construction a linear function of $X_t$, then $\mu_{t}$ is correlated with subsets of the information set of the agent, $\sigma(\mathcal{F}_{t})$. Then, there exists a projection operator $\mathbf{P}$ such that for any element $H{}_{t}\in\mathcal{F}_{t}$, $\mu_{t}=\mathbf{P}H_{t}+\mathbf{P}_{\perp}H_{t}=\mathbf{P}H_{t}+v_{t}$.
Projecting $\mu_{t+1}$ on $\mu_{t}$:
\small
\begin{eqnarray*}
\mathbf{P}(\mu_{t+1}|\mu_{t}) & = & \mathbf{P}(\mathbf{P}H_{t+1}+v_{t+1}|\mathbf{P}H_{t}+v_{t})\\
 & = & \mathbf{P}(\mathbf{P}H_{t+1}+v_{t+1}|\mathbf{P}H_{t})+\mathbf{P}(\mathbf{P}H_{t+1}+v_{t+1}|v_{t})\\
 & = & \mathbf{P}(H_{t+1}|H_{t})
\end{eqnarray*}
\normalsize

This implies that, there exists
a real valued matrix $\tilde{\Gamma}$ such that $\mu_{t+1}=\tilde{\Gamma}\mu_{t}+u_{t+1}$.
Substituting in the condition 9.1 and collecting all the errors in $w_{t+1}$ (we can do this as they are not uniquely defined) we have that:
\begin{eqnarray*}
(F(\theta_1,0 )\Gamma-G(\theta_1,0 )){\lambda}_{t+1}&=& \tilde{\Gamma}(F(\theta_1,0 )\Gamma-G(\theta_1,0 )){\lambda}_{t} +w_{t+1} 
\end{eqnarray*}
 Denoting by $\tilde{C}$ the generalized inverse of $C:=F(\theta_1,0 )\Gamma-G(\theta_1,0 )$ we have that ${\lambda}_{t+1}= \tilde{\Gamma}\lambda_t+C\tilde{C}w_{t+1}$. Comparing with the proposed representation for $\lambda_t$ we have that $\Gamma = \tilde{\Gamma}$ and a non uniquely defined $\nu_{t}$ that nevertheless satisfies $\mathbb{E}(\nu_t)=0$

\textbf{Note}: The assumption that the system can be casted in the expectational
form to which we apply the method of undetermined coefficients is with no loss of generality. More elaborate methods like a canonical or Schur decomposition,
can be used to obtain forward and backward solutions, and the existence of
 an incomplete rule would require arguments similar to the non zero determinant
stability conditions for $\Gamma(\theta)$ and $\mu_{t}$. 
\end{proof}

\begin{proof}
\textbf{of Lemma 3}\\
Given $\mathcal{V}_{i}(\mathbf{Y}_{t-1}) $,  the condition in 5.1 can be rewritten as $\mathbb{E}(\tilde{Y}_{t}^{o}  -{C(\theta)}\hat{X}_{t|t-1})\phi(\mathbf{Y}_{t-1})=0$
where $\tilde{Y}_{t}^{o}\equiv Y_{t}^{o}- \mathcal{V}(\mathbf{Y}_{t-1})$. Given \textbf{LCI}, the Proposition 2-NS in \citet{ECTA1171} can be applied. Moreover, local identification holds for generic $i$, and therefore holds for any $\theta_0$ in $\Theta_I$ as \textbf{LCI} guarantees a unique map from $\mathcal{V}_{i}$ to $\theta_0$.
\end{proof}
\begin{proof}
	\textbf{of Proposition 4 ($n_{y}=r$)}\\
\begin{enumerate} 
	\item For every observable $Y_{t}$ we can construct $n_{\theta}$ moment
	conditions.	Then,  $r$ is the minimum number of observables such
	that $J(\theta)\equiv\frac{\partial}{\partial\theta}\Lambda(\theta)$, is of full column rank. 	
	Thus, for $\mathcal{H}\equiv J_{r\times n_{\theta}}(\theta){}^{T}$, $Z_{t}\in \mathbf{Y_{t-1}}$ and 
	$m(Z_{t},\theta)\equiv\mathcal{H}\otimes\phi(Z_{t})$, the first order condition becomes $\mathbb{E}(m(Z_{t},\theta)-V(Z_{t})\phi(Z_{t})) = 0$ or more compactly $\mathbb{E}m(Z_{t},\theta) \equiv\mathcal{G}(\theta) = \underset{r\times 1}{U}\in  (-\infty,\bar{U}]$.
   Given Lemma 4, $\mathcal{G}(\theta)$ is well behaved, and since $n_{y}=r$, it is invertible in $U$, which guarantees a unique solution $\theta^{*}(U)$ for all admissible $U$. Thus, $\Theta_{I}=\Theta \cap \theta^{*}(U)$.
	\item Given that $\Theta_{I}=\Theta \cap \mathcal{G}^{-1}(U)$, sharpness of $\Theta_{I}$ is determined by the set of admissible structures $S$ such that $\forall s\in S$, $F^{s}_{(Y_{t},Y_{t-1},Z_{t})}: U \equiv \int U_{s,t} dF^{s}\in (-\infty,0]$.
\begin{enumerate} 
	\item Existence  of joint distribution for some $s$ : 
 With no frictions, that is, for $\lambda_{s}=0,\mathbb{P}_{Y_{t-1}} a.s.$, a joint distribution for $(Y_{t},Y_{t-1},Z_{t})$ exists by construction (parametric benchmark model) and thus $F_{(Y_{t},Y_{t-1},Z_{t})}$ is equal to $F_{(Y_{t}|Y_{t-1},Z_{t})}F_{Y_{t-1}}F_{Z_{t}}$. Given additive separability of $\lambda_{s}(Y_{t-1},Z_{t})$, $F^{s}_{(Y_{t},Y_{t-1},Z_{t})}$ also exists as $F^{s}_{(Y_{t},Y_{t-1},Z_{t})}=F_{(Y_{t}|Y_{t-1},Z_{t})}(v-\lambda_{s}(Y_{t-1},Z_{t}))F_{Y_{t-1}}F_{Z_{t}}$. 
 \item The set of admissible joint distributions: Since a joint distribution exists, the set of admissible structures are then determined by:   $U_{s,t}=\lambda_{s}(Y_{t-1},Z_{t})+CZ_{t}$ for any constant matrix $C$. Under Correct specification, $U\neq0$ iff there are frictions in the economy. This is true for the set of admissible densities $\mathbf{F^{s,C}}\equiv\left\{ F^{s}_{(Y_{t},Y_{t-1},Z_{t})}(v):\int U_{s,t} dF^{s}< 0\right\} \supset\mathbf{F^{s,\lambda}}\equiv\left\{F^{s}_{(Y_{t},Y_{t-1},Z_{t})}:\int \lambda_{s}(Y_{t-1},Z_{t}) dF^{s}<0\right\}$. But,$U=0, \forall F^{s}_{(Y_{t},Y_{t-1},Z_{t})}\in \mathbf{F^{s,C}} \setminus \mathbf{F^{s,\lambda}}$. The set of admissible $U$ is the smallest set compatible with $\lambda_{s}\neq 0$ and thus $\Theta_{I}$ is sharp. 
\end{enumerate}
\end{enumerate}
\end{proof}

\begin{proof}
\textbf{of Proposition 5} 

This is a generalization of the proof in Proposition 5. We denote the moment conditions using $n_{y}$ macroeconomic variables by $q_{1}(\theta,\mathbf{Y}_{t-1})$ and moment conditions using $m-n_{y}$ survey variables by $q_{2}(\theta,\mathbf{Y}_{t-1})$:
 \small
\begin{eqnarray*}
\underset{n_{y}\times1}{q_{1}(\theta,\mathbf{Y}_{t-1})} & = & \mathbb{\mathcal{V}}_{1}(\mathbf{Y}_{t-1})\in[\underbar{\ensuremath{\mathcal{V}}}(\mathbf{Y}_{t})_{1},\bar{\mathcal{V}}(\mathbf{Y}_{t})_{1}]\\
\underset{(n_{y}-k)\times1}{q_{2}(\theta,\mathbf{Y}_{t-1})} & = & \mathbb{\mathcal{V}}_{2}(\mathbf{Y}_{t-1})\in[\underbar{\ensuremath{\mathcal{V}}}(\mathbf{Y}_{t})_{2},\bar{\mathcal{V}}(\mathbf{Y}_{t})_{2}]
\end{eqnarray*}
which again imply the following unconditional moment restrictions:
\begin{eqnarray*}
\mathbb{E}(\phi(Z_{t})(q_{1}(\theta,\mathbf{Y}_{t-1})-\mathbb{\mathcal{V}}_{1}(\mathbf{Y}_{t-1})) & = & 0\\
\mathbb{E}(\phi(Z_{t})(q_{2}(\theta,\mathbf{Y}_{t-1})-\mathbb{\mathcal{V}}_{2}(\mathbf{Y}_{t-1})) & = & 0
\end{eqnarray*}
\normalsize
Recall that  $r$ is the minimum number of observables such
that $J(\theta)\equiv\frac{\partial}{\partial\theta}\Lambda(\theta)$, is of full column rank. Rewrite the moment conditions
such that the first $n_{\theta}$ rows satisfy
the rank condition. Given the total number of moment conditions used,
$m$, the rest of the system has $m-n_{\theta}$ equations. We partition
$J(\theta)\equiv(J_{r}(\theta),J_{n_{y}-r}(\theta))$ and let $\mathcal{H}$
be an $m\times n_{\theta}$matrix where $\mathcal{H}\equiv(J_{r\times n_{\theta}}(\theta){}^{T},J_{(n_{y}-r)\times n_{\theta}}(\theta){}^{T},\Delta_{\dim q_{2}\times n_{\theta}}(\theta){}^{T})^{T}$.
Let $m\equiv\mathcal{H}\otimes\phi(Z_{t})$ and partition $m\equiv(m_{\alpha}^{T},m_{\beta}^{T})^{T}$
where $m_{\alpha}$ contains the first $r$ elements. . Since $m>n_{\theta}$,
and given a general weighting matrix $W$, we have the following first
order condition: 
\begin{eqnarray*}
\mathbb{E}(W_{\alpha}^{\frac{1}{2}T}m_{\alpha}+W_{\beta}^{\frac{1}{2}T}m_{\beta}-V(Z_{t})\phi(Z_{t})) & = & 0\\
\mathbb{E}(W_{\alpha}^{\frac{1}{2}T}m_{\alpha}+W_{\beta}^{\frac{1}{2}T}m_{\beta}-\underset{n_{\theta\times1}}{U}) & = & 0
\end{eqnarray*} 
This is a projection of $m$ on a lower dimensional subspace. Since
$W$ is an arbitrary matrix, and $(m_{\alpha},m_{\beta})$ are possibly
correlated, we reproject the sum onto the space spanned by $W_{\alpha}^{\frac{1}{2}T}\mathbf{m}_{\alpha}$.
Define $Q_{\alpha}:=W_{\alpha}^{\frac{1}{2}T}\mathbf{m}_{\alpha}(\mathbf{m}_{\alpha}^{T}W_{\alpha}^{T}\mathbf{m}_{\alpha})^{-1}\mathbf{m}_{\alpha}^{T}W_{\alpha}^{\frac{1}{2}}$, the
projection, and $Q_{\alpha}^{\perp}$ the orthogonal projection.
Since the original sum satisfies the moment condition, then the two
orthogonal complements will also satisfy it: Therefore, 

\begin{eqnarray*}	
Q_{\alpha}\left(W_{\alpha}^{\frac{1}{2}T}\mathbf{m}_{\alpha}+W_{\beta}^{\frac{1}{2}T}\mathbf{m}_{\beta}-U\right) & = & W_{\alpha}^{\frac{1}{2}T}\mathbf{m}_{\alpha}+Q_{\alpha}\left(W_{\beta}^{\frac{1}{2}T}\mathbf{m}_{\beta}-U\right)=0\\
Q_{\alpha}^{\bot}\left(W_{\alpha}^{\frac{1}{2}T}\mathbf{m}_{\alpha}+W_{\beta}^{\frac{1}{2}T}\mathbf{m}_{\beta}-U\right) & = & Q_{\alpha}^{\bot}\left(W_{\beta}^{\frac{1}{2}T}\mathbf{m}_{\beta}-U\right)=0
\end{eqnarray*}
where $U\in[W_{\alpha}^{\frac{1}{2}T}\underbar{\ensuremath{\mathcal{V}}}(\mathbf{Y}_{t})_{\alpha}+W_{\beta}^{\frac{1}{2}T}
\underbar{\ensuremath{\mathcal{V}}}(\mathbf{Y}_{t})_{\beta},\quad W_{\alpha}^{\frac{1}{2}T}\bar{\mathcal{V}}(\mathbf{Y}_{t})_{\alpha}
+W_{\beta}^{\frac{1}{2}T}\bar{\mathcal{V}}(\mathbf{Y}_{t})_{\beta}]\otimes\phi(\mathbf{Y}_t)$.

As in Proposition 4, the first
set of restrictions identifies a one to one mapping from $U$ to $\Theta_{I}$, and therefore $\theta^{*}(U)\equiv \mathcal{G}^{-1}(U)$. Plugging this in the second set of restrictions eliminates dependence on  $m_{\alpha}$ and imposes further restrictions on the domain of variation of $U$. Thus 
\begin{eqnarray*}
	Q_{\alpha}^{\bot}\left(W_{\beta}^{\frac{1}{2}T}\mathbf{m}_{\beta}(U)-U\right) & = & 0 
\end{eqnarray*}
The admissible set for $U$ is now $\left\{U \in (-\inf,\bar{U}): Q_{\alpha}^{\bot}\left(W_{\beta}^{\frac{1}{2}T}\mathbf{m}_{\beta}(U)-U\right)=0\right\}$. Therefore,  $\exists \theta\in \theta(U): \theta\notin \Theta_{I}(U')$ and consequently $\Theta_{I}'\subset \Theta_{I}$. 
\subparagraph{The case of non-linear moment conditions:}The above proof carries through if we replace the linear projection with conditional expectations. In this case, $Q_{\alpha}$ is the conditional expectation operator, which also implies that any integrable moment function $m$ can be decomposed as $m=Q_{\alpha}m+Q^{\bot}_{\alpha}m$ such that $Q_{\alpha}(m-Q_{\alpha}m)=0$. This is important for the corollary regarding survey based conditions as conditional expectations are no longer linear.
\end{proof}
\begin{proof}
	\textbf{of Proposition 6}:  Recall that agent $i$ has the
	following behavioral equation $G(\theta_1,\theta_2)x_{i,t} = F(\theta_1,\theta_2)\mathcal{E}_{t,i}(x_{i,t+1}|\mathcal{F}_{t-1,i})+L(\theta_1,\theta_2)z_{i,t}$, which can be rewritten as :
	\begin{eqnarray}
	G(\theta_1,0)x_{i,t} & = & F(\theta_1,0)\mathbb{E}_{t}(x_{i,t+1}|\mathcal{F}_{t-1})+L(\theta_1,0)z_{i,t} +\mu_{i,t}
	\end{eqnarray}
	Before formally deriving the bounds, we need to establish some facts which will be used in the derivations: 
	\begin{enumerate}
		\item \textbf{$\mu_t$  is a continuous function of the state variables}\\
		The model with frictions reduces to the frictionless model by either by setting $\theta_2=0$ or by eliminating differences in subjective expectations from rational expectations, $\mathcal{E}_{i,t}(X_{t}|\mathcal{F}_{i,t}),\mathbb{E}(X_{t}|\mathcal{F}_{t,i}\forall i) $. We show that the latter is equivalent to a change in a component of the model. Since, $\mathcal{P}(.|\mathcal{F}_{t,i})$ is absolutely continuous to the Rational Expectations measure $\mathbb{P}(.|\mathcal{F}_{t,i}\forall i)$, there exists a Radon Nikodym derivative $\mathcal{M}_{t,i}:=\frac{d\mathcal{P}_{i,t}}{d\mathbb{P}_t}$ such that $\mathcal{P}_{i,t}(\chi)=\underset{\chi}{\int}\mathcal{M}_{t,i}\mathbb{P}_{t}$ for $\chi \subset X$. For $\mathcal{M}_{t,i}=1$ for all $t,i$ we get the frictionless model. Then given $\theta_2$ and $\mathcal{M}_{t,i}$, $\mu_t$ is a continuous w.r.t state variables.
		\item \textbf{Any probability statement about $\mu_t$ translates to a probability statement on $\lambda_t$.}\\
		Since $x_{i,t}^{\star}$ solves the behavioral functional equation of the agent uniquely, there is a map $h:(G,F,L)\to (P,Q)$ which is a continuous bijection, and by the implicit function theorem, any perturbation to the first order conditions (change in $(G,F,L)$) maps to perturbations of the solution, ($P,Q$). Therefore,  for every univariate decision variable, $\mathbb{P}(\mu_{i,t}\in  [\underline{\mu}_{i},\bar{\mu}_{i}])=\mathbb{P}(h(\lambda_{i,t})\in  [\underline{\mu}_{i},\bar{\mu}_{i}]) = \mathbb{P}(\lambda_{i,t}\in  [\underline{\lambda}_{i},\bar{\lambda}_{i}]) $. Given Proposition 1, same statement holds also for $\mathbb{E}_{t}\lambda_{i,t}$ and $\mathbb{E}_{t}\mu_{i,t}$ .
		\item\textbf{Any probability statement on the subjective conditional expectations translates to a probability statement on $\mu_{i,t}$}\\
		Recall that the enlarged state vector contains also past states, $X_t\equiv(X_t,\tilde{X}_t)$ where $\tilde{X}_t=X_{t-1}$. Observing $\mathcal{E}_{i,t}(x_{i,t+1})\in B_l$ therefore implies observing either an expectation about the future or the present. Given the behavioral equation of the agent, $\mathcal{E}_{i,t}(x_{i,t+1})$ maps deterministically to $x_{i,t}$ and therefore $\mu_{i,t}$. 
	\end{enumerate}
	Given the above, we consider the expected value of the statistic $\mathcal{\hat{B}}_{k,t}$. 
	Given the (joint) measure $\mathbb{P}=\mathbb{P}(t)\times\Lambda(i)$, taking the expectation we have that 
	
	\begin{eqnarray*}
		\mathbb{E}_{t}\hat{\mathcal{B}}_{k,t} & = & \sum_{i\leq N}w_{i}\int\int1(\mathcal{E}_{i,t}(x_{i,t+1}|\mathcal{F}_{t-1,i})\in B)d(P(\mathcal{F}_t|\mathcal{F}_{t-1})\times \Lambda_{t}(i))\\
		& = & \sum_{i\leq N}w_{i}\bar{\mathbb{P}_{t}}(\mathcal{E}_{i,t}(x_{i,t+1}|\mathcal{F}_{t-1,i})\in B)\\
		& = &\bar{\mathbb{P}_{t}}(\mathcal{E}_{i,t}(x_{i,t+1}|\mathcal{F}_{t-1,i})\in B)\\
		& = &\bar{\mathbb{P}_{t}}(\mathbb{E}_t\mu_{i,t}\in [\mathbb{E}_t\underline{\mu}_{i,t},\mathbb{E}\bar{\mu}_{i,t}])\\
		& = &\bar{\mathbb{P}_{t}}(\mathbb{E}_t\lambda_{i,t}\in [\mathbb{E}_t\underline{\lambda}_{i,t},\mathbb{E}\bar{\lambda}_{i,t}])
	\end{eqnarray*}

	In the second to last equality, we use fact \textbf{3} and in the last equality we have used fact \textbf{2} adapted to conditional means. We next derive the main result. 
	
	Let $Y_{t}^{o}=X_{t}$, $Y_{t}^{m}=P(\theta_{1},0)X_{t-1}$. Denote by $(\mathcal{X}^{\star}_{t},\mathcal{Z}^{\star}_{t})$ the time dependent subset of the support such that $\lambda_{i,t}<0 ,\forall (x_{i,t-1},z_{i,t})\in (\mathcal{X}^{\star}_{t},\mathcal{Z}^{\star}_{t})$ and by $(\mathcal{X}^{\star,C}_{t},\mathcal{Z}^{\star,C}_{t})$ their complements in $(\mathcal{X},\mathcal{Z})$.
	Given that $\lambda^{1,2}_{i}$ can possibly differ across $i$, let\footnote{The typical example in this case would be economies with ex post heterogeneity i.e. across net asset positions.} 
	\begin{eqnarray}
	\lambda^{j}_{i,t}&\equiv& 
	\begin{cases}
	\lambda^{1}x_{i,t-1}\text{ or }\lambda^{2}z_{i,t} &\text{    if }(x_{i,t-1},z_{i,t})\in (\mathcal{X}^{\star}_{t},\mathcal{Z}^{\star}_{t})\\
	0 & \text{otherwise}
	\end{cases}
	\end{eqnarray} 	
	\small
	\begin{eqnarray*}
		Y_{t}^{o}	&=& \int \lambda^{1}_{t}d\Lambda_{t}(i)+\int\lambda^{2}_{t}d\Lambda_{t}(i)+Q(\theta_1,0)Z_{t}+Y_{t}^{m} \\
		& = & \lambda^{1}\int_{\mathcal{X}^{\star}_{t}} x_{i,t-1}d\Lambda_{t}(i)+\lambda^{2}\int_{\mathcal{Z}^{\star}_{t}} z_{i,t}d\Lambda_{t}(i)+Q(\theta_1,0)Z_{t}+Y_{t}^{m} 
		\end{eqnarray*}
	Pre-multiplying and post dividing every element of $\int_{\mathcal{X}^{\star}_{t}}x_{i,t-1}d\Lambda_{t}(i)$ and ${\int_{\mathcal{Z}^{\star}_{t}} z_{i,t}d\Lambda_{t}(i)}$ with the corresponding aggregates:
		\begin{eqnarray}
		&&\quad{}Y_{t}^{o}= \lambda^{1}X_{t-1}\odot{\int_{\mathcal{X}^{\star}_{t}}x_{i,t-1}d\Lambda_{t}(i)}\oslash X_{t-1}+\lambda^{2}\odot Z_{t}{\int_{\mathcal{Z}^{\star}_{t}} z_{i,t}d\Lambda_{t}(i)}\oslash{Z_{t}}+Q(\theta_1,0)Z_{t}+Y_{t}^{m} \label{eq:had}
		\end{eqnarray} where $\odot$ and $\oslash$ are 
	the Hadamard product and division respectively, and $\mathbb{P}$ is the probability over vectors $z_{i,t}$ and $x_{i,t-1}$ using $\Lambda_{i}$.
	Without loss of generality, looking at the first element of the vector  ${\int_{\mathcal{X}^{\star}_{t}}x_{i,t-1}d\Lambda_{t}(i)}\oslash X_{t-1}$, 
	\begin{eqnarray*}
		&& {\int_{\mathcal{X}^{\star}_{t}}x^{1}_{i,t-1}d\Lambda_{t}(i)}\oslash X^{1}_{t-1}\\
	&=&\frac{\mathbb{E}(x^{1}_{i,t-1}|x_{i,t}\in\mathcal{X}^{\star}_{t})\mathbb{P}(x_{i,t}\in \mathcal{X}^{\star}_{t})}{\mathbb{E}(x^{1}_{i,t-1}|x_{i,t}\in\mathcal{X}^{\star}_{t})\mathbb{P}(x_{i,t}\in \mathcal{X}^{\star}_{t})+\mathbb{E}(x^{1}_{i,t-1}|x_{i,t}\in\mathcal{X}^{\star,C}_{t})\mathbb{P}(x_{i,t}\in \mathcal{X}^{\star,C}_{t})}\\
	&=&\frac{\mathbb{P}(x_{i,t}\in \mathcal{X}^{\star}_{t})}{\mathbb{P}(x_{i,t}\in \mathcal{X}^{\star}_{t})+\mathbb{P}(x_{i,t}\in \mathcal{X}^{\star,C}_{t})+\left(\frac{\mathbb{E}(x^{1}_{i,t-1}|x_{i,t}\in\mathcal{X}^{\star,C}_{t})}{\mathbb{E}(x^{1}_{i,t-1}|x_{i,t}\in\mathcal{X}^{\star}_{t})}-1\right)\mathbb{P}(x_{i,t}\in \mathcal{X}^{\star,C}_{t})}\\
	&:=&\frac{\mathbb{P}(x_{i,t}\in \mathcal{X}^{\star}_{t})}{1+\hat{R}^{1}_{t}}\frac{B_{t}}{1+\hat{R}^{1}_{t}}\leq \frac{B_{t}}{1+\inf_{\mathcal{X}^{\star}_{t}} \hat{R}^{1}_{t}} = B_{t}	\end{eqnarray*}
        The last inequality holds since $\inf_{\mathcal{X}^{\star}_{t}}\hat{R}^{1}_{t}=0$,  $\sup_{\mathcal{X}^{\star}_{t}}\hat{R}^{1}_{t}=\infty$ and by the monotonicity assumption ($\mathbf{S-4}$). The inequality holds for every element of $x_{i,t}$. Finally, taking conditional expectations of \eqref{eq:had} (using $\mathbb{P}(.|X_{t-1})$:\vspace{-0.5 cm}
    	\begin{eqnarray*}
    \mathbb{E}(Y_{t}^{o}|X_{t-1})& \leq & \lambda^{1}X_{t-1}\odot \mathbb{E}( B_{t}|X_{t-1})+Y_{t}^{m}
    \end{eqnarray*}    
\end{proof} \vspace{-1.3 cm}
\begin{proof}
\textbf{of Corollary 7}\\
\small	(1): When $B_{t}=0$, $\mathbb{P}(x_{i,t}\in \mathcal{X}^{\star}_{t})=0$ and therefore the model has no frictions, which trivially restores point identification. When  $B_{t}=1$, $\int_{\mathcal{X}^{\star}_{t}}x^{1}_{i,t-1}d\Lambda_{t}(i)=X_{t-1}$, and this results to the first type of moment inequality using macroeconomic data. We now show the possibility of refinement of $\Theta$ for $B\in(0,1)$. Defining the conditional errors from the first and second types of inequalities as $\nu_{t}\equiv Y^{o}_{t}-P(\theta_{1},0)X_{t-1}$  and  $\eta_{t}\equiv \hat{Y}^{o}_{t}-P(\theta_{1},0)X_{t-1}-\lambda^{1}X_{t-1}\odot\mathbb{E}( B_{t}|X_{t-1})$ respectively,  $\mathbb{E}(\eta_{t}|\nu_{t},\mathbf{Y})=\int\eta_{t}p(\eta_{t}|\nu_{t},\mathbf{Y})\neq\eta_{t}$ since $p(\eta_{t}|\nu_{t},\mathbf{Y})\neq 1$ and $\eta_{t}-\mathbb{E}(\eta_{t}|\nu_{t},\mathbf{Y})\neq 0$ unless $B_{t}=0$. \\
		(2): Suppose that $\Theta_{I}$ is a singleton. This implies that $B=0$, or $\Lambda_{t}(i)$ has unit mass on one agent which is unconstrained, and thus $B=0$ too. 
\end{proof}

\begin{proof}
	\textbf{of Proposition 8}\\
Consider the quantity $T^{-\frac{1}{2}}\inf_{\theta\in\hat{\Theta}_{s}}||\mathcal{V}^{-\frac{1}{2}}(\hat{f}(\hat{\theta}_{p})-\hat{f}(\hat{\theta}_{s}))||$. 
Given that $\Theta$ is a connected set and $\hat{\Theta}_{s}\in\Theta$,
then $\hat{\Theta}_{s}$ is also connected. For any $\theta,\theta'\in cl(\hat{\Theta}_{s})$,
$d(\theta,\theta')<\epsilon$ for arbitrarily small $\epsilon>0$.
This implies that if $\theta_{p}\in\Theta_{s}$ then there exists
a $\theta_{s}\in\hat{\Theta}_{s}$ such that $||\mathcal{V}^{-\frac{1}{2}}(f(\theta_{p})-f(\theta_{s}))||<\epsilon$.
For every estimating equation, redefine $\hat{f}(\hat{\theta}_{p})=T^{-1}\sum_{t}q(Y;\theta)-T^{-1}\sum_{t}\mathbb{E}q(Y;\theta)+\gamma_{T}=A_{T}(\theta)+\gamma_{T}$
where $A_{T}(\theta)\equiv(A_{T,1}(\theta),A_{T,2}(\theta),...A_{t,p}(\theta))^{T}$
and $\gamma_{T}=(\gamma_{1,T},\gamma_{2,T},...\gamma_{p,T},)^{T}$. By element-wise mean value expansion around $\theta_{l} \in \Theta_{l},l\in{s,p}$, $\hat{f}(\hat{\theta}) = \mathbb{E}q(Y_{t},\theta_{l}) + (\hat{f}(\theta_{l})-\mathbb{E}q(Y_{t},\theta_{l})) + D(\tilde{\theta})(\theta - \theta_{l}))$. Given mild assumptions on the $(2+\delta)$ boundedness of each moment, the second component scaled by $\sqrt{T}$, $\sqrt{T}(\hat{f}(\theta_{l})-\mathbb{E}\hat{f}(\theta_{l}))\to_{d}\mathcal{N}(0,\Omega_{1,l})$ while the scaled third component, $D(\tilde{\theta})\sqrt{T}(\theta - \theta_{l}))\to_{d}\mathcal{N}(0,\Omega_{2,l})$. Consequently, for each $j,j\leq p$, $A_{T,j}\to_{d}\mathcal{N}(0,\Omega^{*}_{j})$. Let $\mathcal{V}(\theta)=AsyVar(T^{\frac{1}{2}}f(\theta_{p})-T^{\frac{1}{2}}f(\theta_{s}))$
and denote by $\mathcal{V}_{d}(\theta)$ be the matrix containing
only the diagonal elements of $\mathcal{V}(\theta)$.  \vspace{-0.2 cm}
\begin{eqnarray*}
T^{\frac{1}{2}}\mathcal{V}_{d}^{-\frac{1}{2}}(\hat{f}(\hat{\theta}_{p})-\hat{f}(\hat{\theta}_{s})) & = & \mathcal{V}_{d}^{-\frac{1}{2}}T^{\frac{1}{2}}(A_{T}(\theta_{p})-A_{T}(\theta_{s})+\gamma_{T,p}-\gamma_{T,s})
\end{eqnarray*} \vspace{-0.2cm}
Under $H_{0}$, $\theta_{p}\in\Theta_{s}$, or $f(\theta_{p})\in f(\Theta_{s})$ and thus $\inf_{\theta \in \Theta_{s}}(\gamma_{T,p}-\gamma_{T,s})=0 $
\begin{eqnarray*}
T\mathcal{W}(\theta_{p},\Theta_{s}) & = & \inf_{f(\theta_{s})\in f(\hat{\Theta}_{s})}||\hat{\mathcal{V}_{d}}^{-\frac{1}{2}}(T^{\frac{1}{2}}A_{T}(\theta_{p})-T^{\frac{1}{2}}A_{T}(\theta_{s})+T^{\frac{1}{2}}(\gamma_{T,p}-\gamma_{T,s}))||^{2}\\
 & = & \inf_{\theta_{s}\in\hat{\Theta}_{s}}||\hat{\mathcal{V}_{d}}^{-\frac{1}{2}}(T^{\frac{1}{2}}A_{T}(\theta_{p})-T^{\frac{1}{2}}A_{T}(\theta_{s})+T^{\frac{1}{2}}(\gamma_{T,p}-\gamma_{T,s}))||^{2}\\
 &\overset{d}{\to}& ||\mathcal{V}_{d}^{-\frac{1}{2}}\mathcal{V}^{\frac{1}{2}}\mathcal{N}(0,I_{p})||^{2} 
 \end{eqnarray*}\vspace{-0.2 cm}
Under $H_{1}$, $\theta_{p}\neq\Theta_{s}$, and
therefore $\underset{\theta_{s}\in\Theta_{s}}{\inf}(\gamma_{T,p}-\gamma_{T,s})=O(1)$. \vspace{-0.2cm}
\begin{eqnarray*}
	T\mathcal{W}(\theta_{p},\Theta_{s}) & = & \inf_{f(\theta_{s})\in f(\hat{\Theta}_{s})}||\hat{\mathcal{V}}^{-\frac{1}{2}}T^{\frac{1}{2}}(A_{T}(\theta_{p})-A_{T}(\theta_{s})+\gamma_{T,p}-\gamma_{T,s})||^{2}\\
	& = & ||\hat{\mathcal{V}}^{-\frac{1}{2}}(T^{\frac{1}{2}}A_{T}(\theta_{p})-T^{\frac{1}{2}}A_{T}(\theta_{s})+\inf_{f(\theta_{s})\in f(\hat{\Theta}_{s})}T^{\frac{1}{2}}(\gamma_{T,p}-\gamma_{T,s})||^{2}\\
	& = & ||O_{p}(1)+O_{p}(T^{\frac{1}{2}})||^{2}= O_{p}(T)
\end{eqnarray*} 
\end{proof}

\begin{proof} \textbf{of $\boldsymbol{\mu_{ols}\notin \mu_{ID,2}}$} in Partial Equilibrium
	
	We first solve for the consumption process. Let $p_{t}:=\mathbb{P}_{t}(\lambda_{i,t}>0)$ and $p:=\frac{\mathbb{E}(c^{2}_{i,t}\mathbb{P}_{t}(\lambda_{i,t}>0))}{\mathbb{E}c^{2}_{i,t}}$.	Moreover, let $\beta(1+r)$ and $\rho\to\infty$, $\beta(1+r)^{\frac{1}{\rho}}=1$ and therefore $\tilde{\mu}_{0}=0$.
	Denote by $\varphi(u)$ the Normal density at $u$ and by $\mathfrak{s}(.)$ the marginal density of $c_{i,t}$, Assuming that $c_{i,t}$ started $\tau$ periods ago, \small
	\begin{eqnarray*} 
		c_{i,t+1}&=&\sum^{\tau}_{j=0}\epsilon_{t+1-j}(1+\chi_{i,t-j}\lambda_2)\Pi^{j-1}_{k=0}(1+\chi_{i,t-k}\lambda_1):=\sum^{\tau}_{j=0}\epsilon_{t+1-j}z_{t+1-j} :=\sum^{\tau}_{j=0}\xi_{t+1-j} 
	\end{eqnarray*}  	 	\normalsize	 	
	where $\mathbb{P}(z_{t+1-j}=0)=0$. Consumption is a convolution of $\tau+1$ independent variables $\xi_{t+1-j}$ whose marginal density is a mixture of $N(0,\omega(\lambda_1,\lambda_2,j))$ with weights that depend on $p$ and $j$. For example, $\xi_{t+1-j}=\epsilon_{t+1-j}(1+\lambda_2)\Pi^{j-1}_{k=0}(1+\lambda_1)$ w.p.$(1-p)^{j+1}$. For each $\xi_{t+1-j}$, symmetry around zero is maintained. Next, we prove the result for $\tau=1$. Evaluating the LHS of \eqref{eq:probit4} at the true $\tilde{\mu}_{0}$ and the RHS at $\mu_{ols}$, setting $u=0$  and taking their difference,\small \begin{eqnarray*}
		\Delta :=  \big(\frac{1}{2} -\Phi_{0,\sigma_{\epsilon}^{2}}( -(\tilde{\mu}_{ols}c_{i,t}))\big)(1-p_{t})+(\Phi_{0,\sigma_{\epsilon}^{2}(1+\lambda_{2})^{2}}(-(\tilde{\mu}+\lambda_1)c_{i,t}) -\Phi_{0,\sigma^{2}_{ols}}( -\tilde{\mu}_{ols}c_{i,t}) )p_{t}&&
	\end{eqnarray*}	
	\normalsize	 	
	For $\mu_{ols}$ to be admissible, $\Delta$ should be weakly positive for all $p_{t}$ such that 
	
	$(\Phi_{0,\sigma_{\epsilon}^{2}(1+\lambda_{2})^{2}}(-(\tilde{\mu}+\lambda_1)c_{i,t}) -\Phi_{0,\sigma_{ols}^{2}}( -\tilde{\mu}_{ols}c_{i,t}) )\geq 0$.  Sufficient conditions for the latter, are $\sigma^{2}_{ols}=\sigma_{\epsilon}^{2}(1+\lambda_{2})^{2}$, which is true for all $(\lambda_{1},\lambda_{2}):\lambda^{2}+2\lambda_{2}-\lambda_{1}\frac{\sigma^{2}_{c}}{\sigma^{2}_{\epsilon}}p(1-p)=0$, and for probability functions  $B_{t}:\lambda_{1}c_{i,t}(1-p)\geq 0$.  We choose $B_{t}=1(c_{i,t}<0)$.
	
	Substituting for the test point $\tilde{\mu}=\mu_{ols}$ and $B_{t}=1(c_{i,t}<0)$, $\Delta>0$ for $c_{i,t}>0$ by construction, while for $c_{i,t}<0$ we have a contradiction: \small
	\begin{eqnarray}
	\Delta= \frac{1}{2} -\Phi_{0,\sigma_{\epsilon}^{2}}( -p\lambda_1c_{i,t})=\frac{1}{2} -\Phi_{0,\sigma_{\epsilon}^{2}}( p\lambda_1|c_{i,t}|)<0&& \label{eq:Probit5}
	\end{eqnarray}	 \normalsize	 
	
\end{proof}
\newpage
\subsubsection{Small Open Economy model for Spain}In what follows variables
with $\ast $ denote the "rest of the world", $y_{t}$ is real output, $c_{t}$
is consumption, $i_{t}$ investment, $q_{t}$ the value of capital, $k_{t}$ is
productive capital, $k_{t}^{s}$ capital services, $z_{t}$ is capital
utilization, $\mu _{t}^{p}$ is the price markup, $\pi _{t}$ is domestic
inflation, $\pi _{cpi}$ is CPI inflation, $r_{t}^{k}$ is the rental rate of
capital, $w_{t}$ is the real wage and  $r_{t}$ is the interest rate
\small
\begin{eqnarray*}
y_{t} & = & c_{y}c_{t}+i_{y}i_{t}+z_{y}z_{t}+nx_{y}s_{t}+\epsilon_{t}^{g} \\
c_{t} & = & c_1\mathbb{E}_{t}c_{t+1}+c_{2}(l_{t}-\mathbb{E%
}l_{t+1})-c_{3}(r_{t}-\mathbb{E}_{t}\pi_{t+1}+\epsilon_{t}^{b}) \\
\mathbb{E}_{t}r_{t+1}^{k} & = & r_{1}(r_{t}-\mathbb{E}_{t}\pi_{t+1}+%
\epsilon_{t}^{b}) \\
y_{t} & = & \phi_{p}(\alpha k_{t}^{s}+(1-\alpha)l_{t}+\epsilon_{t}^{\alpha})
\\
c_{t} & = & c_{t}^{*}+\frac{1}{\sigma_{a}}s_{t} \\
k_{t}^{s} & = & k_{t-1}+z_{t} \\
z_{t} & = & z_{1}r_{t}^{k} \\
k_{t} & = & k_{1}k_{t-1}+(1-k_{1})i_{t} \\
\mu_{t}^{p} & = & \alpha(k_{t}^{s}-l_{t})+\epsilon_{t}^{\alpha}-w_{t} \\
\pi_{t} & = & \pi_{1}\pi_{t-1}+\pi_{2}\mathbb{E}_{t}\pi_{t+1}-%
\epsilon_{t}^{p} \\
\pi_{cpi} & = & \pi_{t}+\nu\Delta s_{t} \\
r_{t}^{k} & = & -(k_{t}-l_{t})+w_{t} \\
\mu_{t}^{w} & = & w_{t}-\sigma_{l}l_{t}+c_{t} \\
w_{t} & = & w_{1}w_{t-1}+(1-w_{1})(\mathbb{E}_{t}w_{t+1}+\mathbb{E}%
\pi_{t+1})-w_{2}\pi_{t}+w_{3}\pi_{t-1}-w_{4}\mu_{t}^{w}+\epsilon_{t}^{w} \\
r_{t} & = & \rho
r_{t-1}+(1-\rho)[r_{\pi}\pi_{cpi,t}+r_{y}(y_{t}-y_{t}^{^{p}})]+r_{\Delta
Y}[(y_{t}-y_{t}^{^{p}})-(y_{t-1}-y_{t-1}^{^{p}})]+\epsilon_{t}^{r}
\end{eqnarray*} \normalsize

Financial Frictions and Capital adjustment costs imply the following changes to the model, where $Q_{t}$ is the price of capital and  $R^{k}_{t}$ the return to capital :
\begin{eqnarray*}
	\mathbb{E}_{t}R_{t+1}^{k} & = & -\chi_{pr}(N_{t}-Q_{t}-k_{t})+r_{t}-\mathbb{E}_{t}\pi_{t+1} \\
	N_{t} & = & \gamma\bar{r}_{k}\left(\left(\frac{\bar{K}}{\bar{N}}\right)\left( R^{k}_{t}-r_{t-1}\right)+r^{k}_{t}+N_{t-1}\right)\\
    i_{t}&=& i_{1}i_{t-1}+(1-i_{1})\mathbb{E}_{t}i_{t+1} + \frac{1}{\phi_{adj}}Q_{t} +\epsilon_{q,t}\\
    Q_{t}&=& q_{1}\mathbb{E}_{t}Q_{t+1}+(1-q_{1})\mathbb{E}_{t}r^{k}_{t+1} -(R^{k}_{t}+\epsilon_{b,t})
     \end{eqnarray*} \normalsize

 \begin{table}[H]
	\begin{centering}
		\caption{Confidence Set for $\Theta_{IM}$ and $\Theta_{CM}$-with survey data}\bigskip
		\begin{tabular}{|c||c|c||c|c||}
			\hline
			&  \textbf{Incomplete Model}&& \textbf{Complete Model}\\\hline 
			Parameter & $q_{2.5\%}$ & $ q_{97.5\%}$ & $q_{5\%}$ & $ q_{95\%}$ \\ \hline 

$\sigma_c$&  12.460 & 13.320&28.6900 & 30.0000 \\ \hline 
$\phi_{p}$& 5.1840& 5.3430&21.5700 & 21.7000\\ \hline
$g_{y}$&0.9668 & 1.0000&0.2432 & 0.4827\\  \hline
$\lambda_{h}$ &0.7550 & 0.7673&0.0009& 0.0009\\  \hline
$\sigma_{l}$& 8.5140 & 8.5760&8.892 & 9.104\\ \hline
$\xi_{p}$&0.3654 & 0.3846&0.0000& 0.0000 \\ \hline
$\xi_{w}$&0.0945 & 0.0958&0.8173 & 0.8628 \\  \hline
$\iota_{w}$& 0.5631 & 0.575 &0.9748 & 0.9999\\ \hline
$\iota_{p}$& 0.4952 & 0.498&0.0004 & 0.0004\\ \hline
$z_{1}$	&0.5360 & 0.5429&0.0009 & 0.0009\\ \hline
$\rho_{a}$	& 0.6447 & 0.6504&0.7521 & 0.9092 \\ \hline
$\rho_{b}$ &0.6729 & 0.6789&\textit{0.0000}&\textit{0.0000}\\ \hline
$\rho_{g}$ & 0.0000 & 0.0023&0.9999 & 0.9999\\ \hline
$\rho_{qs}$&0.3394 & 0.3468&0.3242 & 0.4079\\ \hline
$\rho_{ms}$	& 0.4044 & 0.4137&0.7347 & 0.8015\\ \hline
$\rho_{p}$&0.4802 & 0.4843&0.6329 & 0.7077\\ \hline
$\rho_{w}$ &0.2690 & 0.2719&0.0003 & 0.0003\\\hline
$ma_{p}$&0.9958 & 1.000&0.0129 & 0.0561\\\hline
$ma_{w}$&0.5466 & 0.5487&0.0010 & 0.0011\\\hline
$\nu$& 0.6084 & 0.6104&0.0053 & 0.0205\\\hline
$\rho_y\star$&	0.6080 & 0.6080&0.1683 & 0.2251\\ \hline
$\sigma_a$&9.8980 & 20.000&0.0245 & 0.0339\\ \hline
$\sigma_b$&3.8910 & 7.3910&3.4480 & 3.4800\\ \hline
$\sigma_g$&13.530 & 13.830&4.5820 & 4.6340\\\hline
$\sigma_{qs}$& 9.2380 & 15.23&0.1831 & 0.2064\\\hline
$\sigma_m$& 9.1220 & 10.240&1.0250 & 1.0670\\\hline
$\sigma_p$&	0.9563 & 1.9220&4.2260 & 4.2660\\\hline
$\sigma_w$ &-&-&- & -\\\hline
$\sigma_f$ &- & - &- & -\\\hline
$\phi_{adj}$&-&-&5.9210 & 5.9550\\\hline
$\chi_{pr}$&-&-&0.9862 & 0.9999\\\hline
$\gamma$&-&-& 0.9944 & 0.9950\\\hline
$K/N$&-&-&4.3330 & 4.3490\\\hline
 		\end{tabular}
 		\par\end{centering}
 \end{table}

Note: $\sigma_w$ and $\sigma_f$ are set to 1. Parameters in italic font reached the boundary and thus calibrated to the boundary value.
\newpage

\newpage

\section*{\textbf{Estimated Wedges - with Survey Data}- $5\%$ level} \begin{figure}[H]
	\begin{centering}
		\includegraphics[scale=0.47]{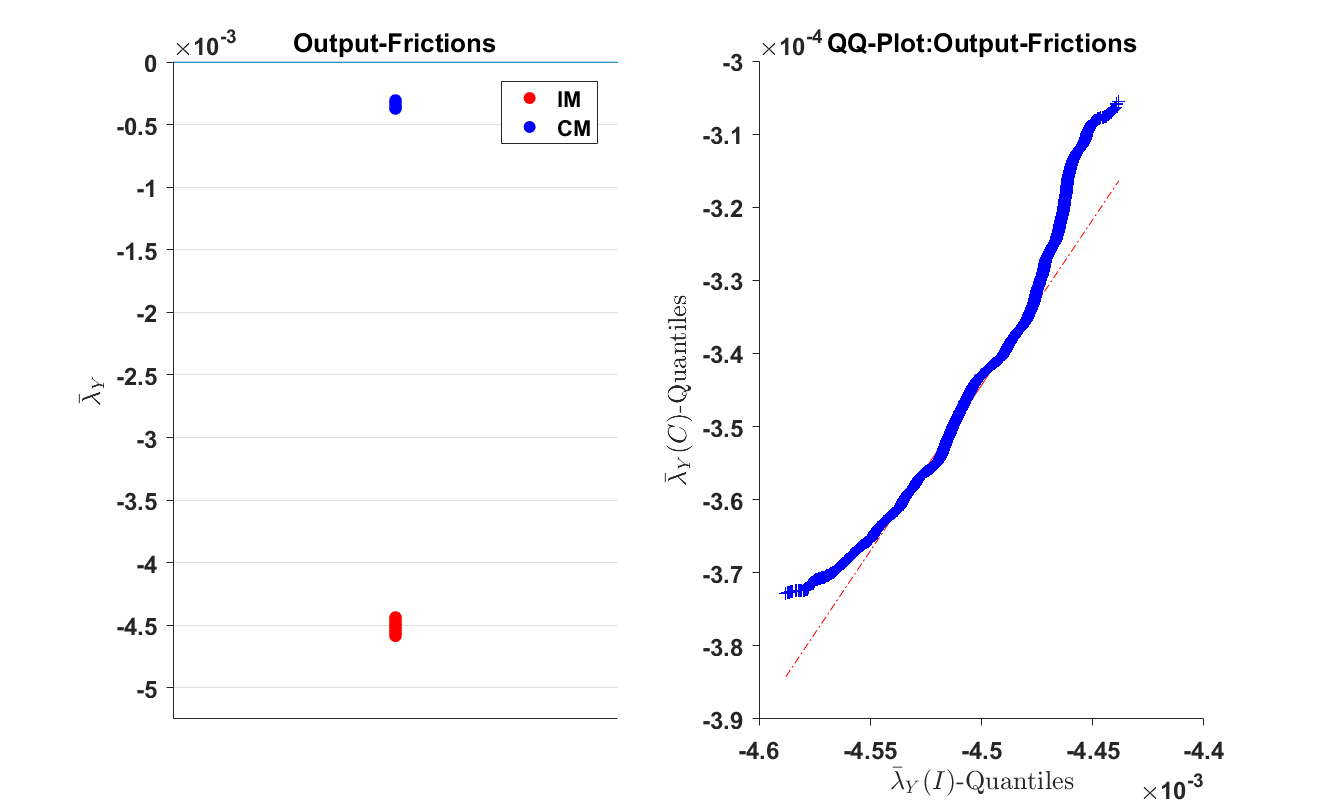}
		\includegraphics[scale=0.47]{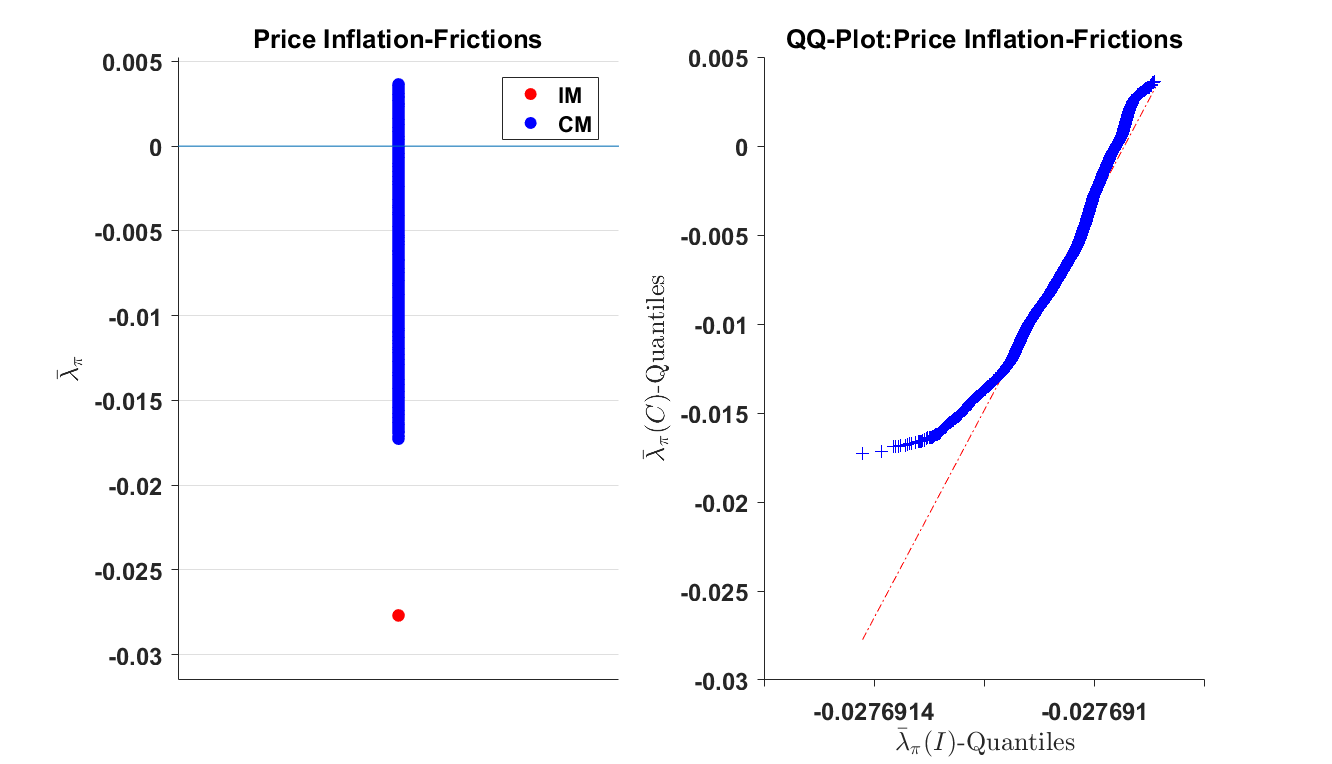}
		\protect\caption{\small Full Sample Wedges (due to Financial Frictions and Capital Adjustment Costs)}
	\end{centering}
\end{figure}

\begin{figure}[H]
	\begin{centering}
		\includegraphics[scale=0.47]{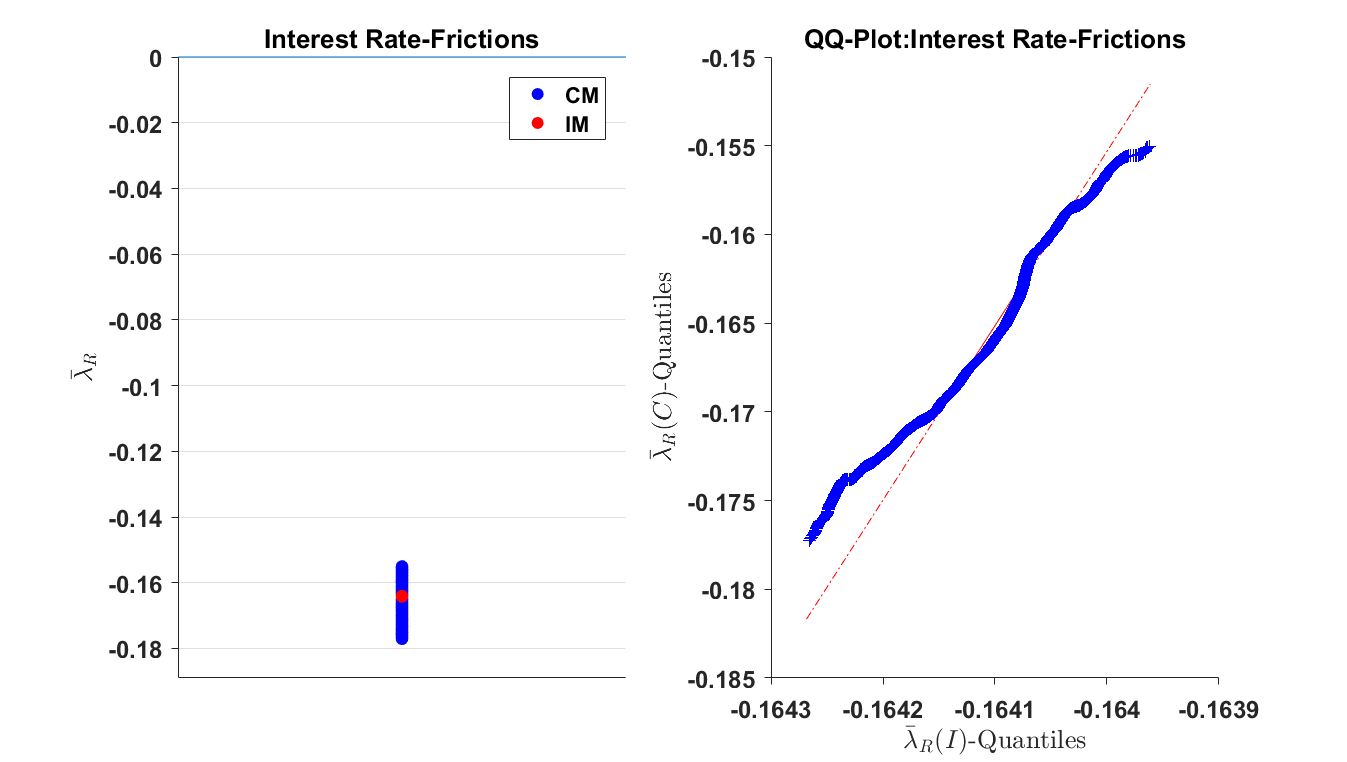}
		\includegraphics[scale=0.47]{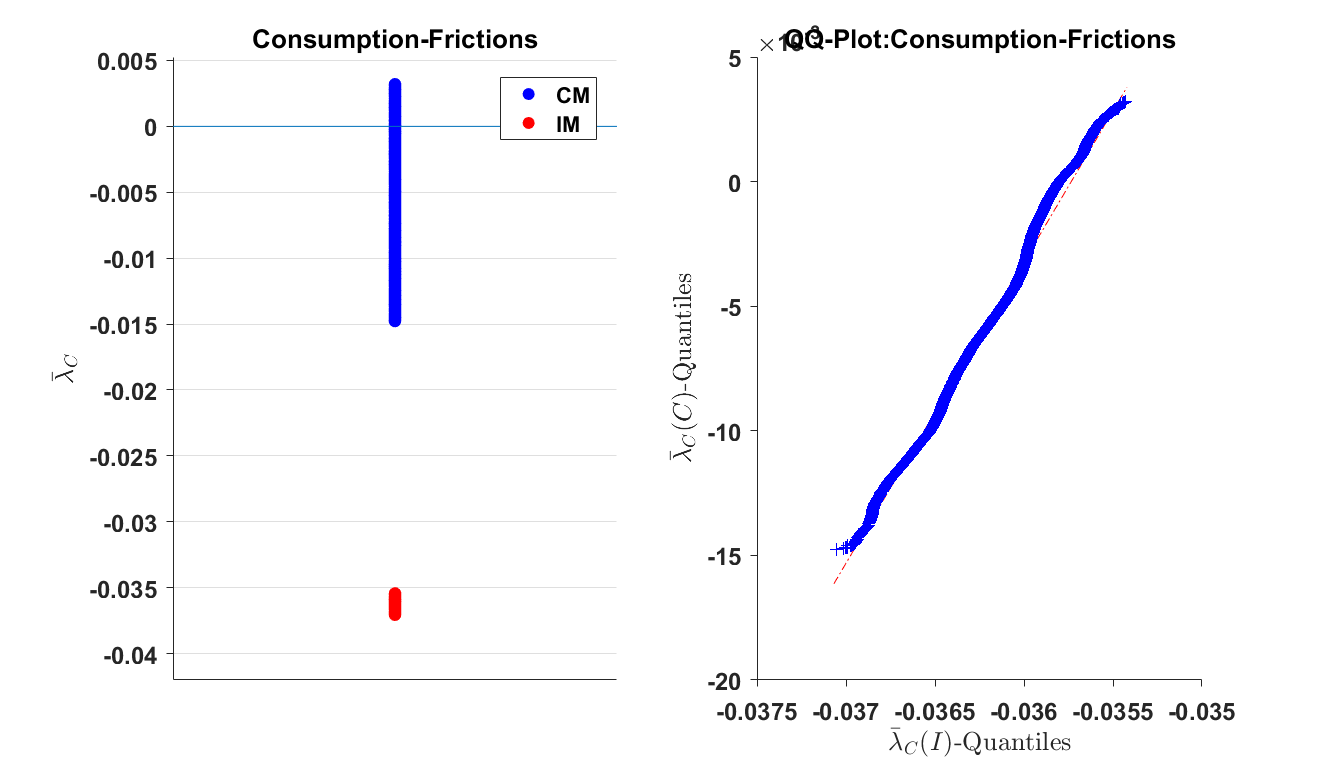}
	 		\protect\caption{Full Sample Wedges (due to Financial Frictions and Capital Adjustment Costs)}
	\end{centering}
\end{figure}
\begin{figure} [H]
	\begin{centering}
	\includegraphics[scale=0.47]{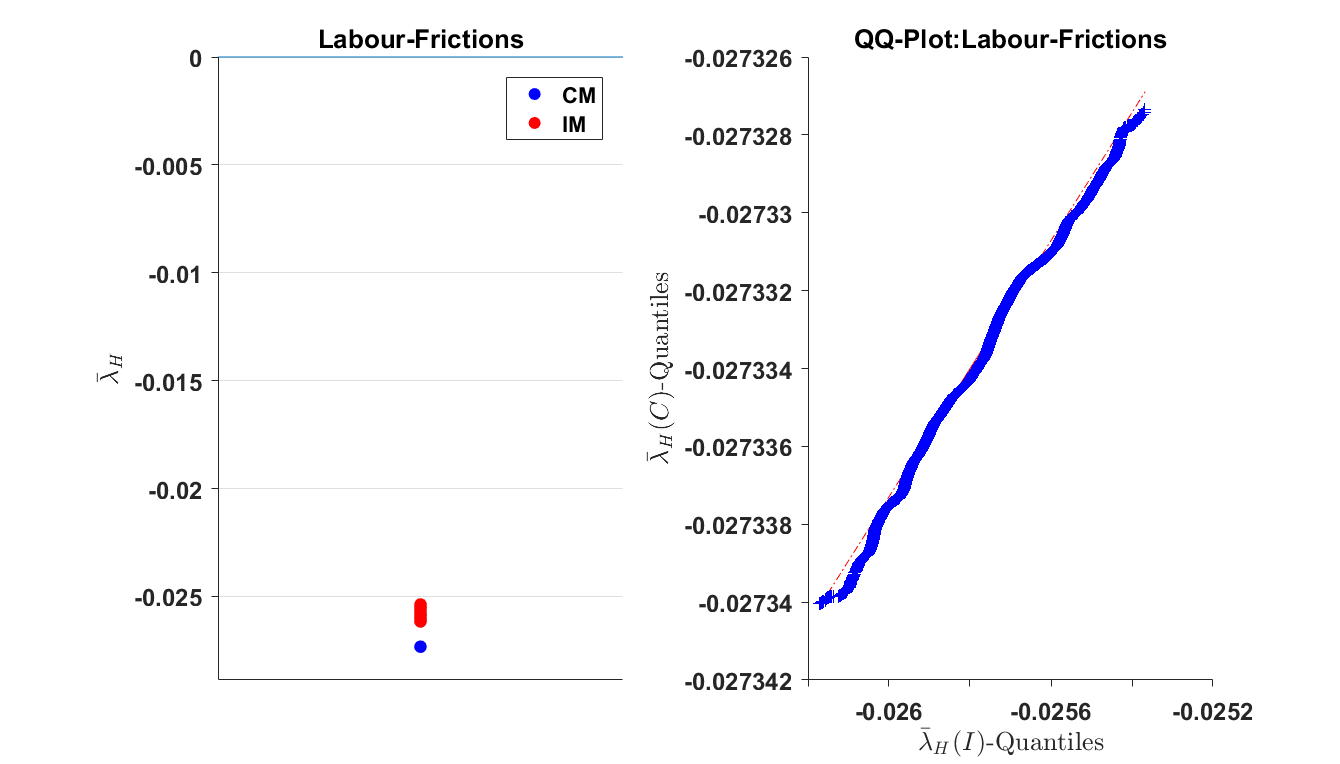}
	\includegraphics[scale=0.47]{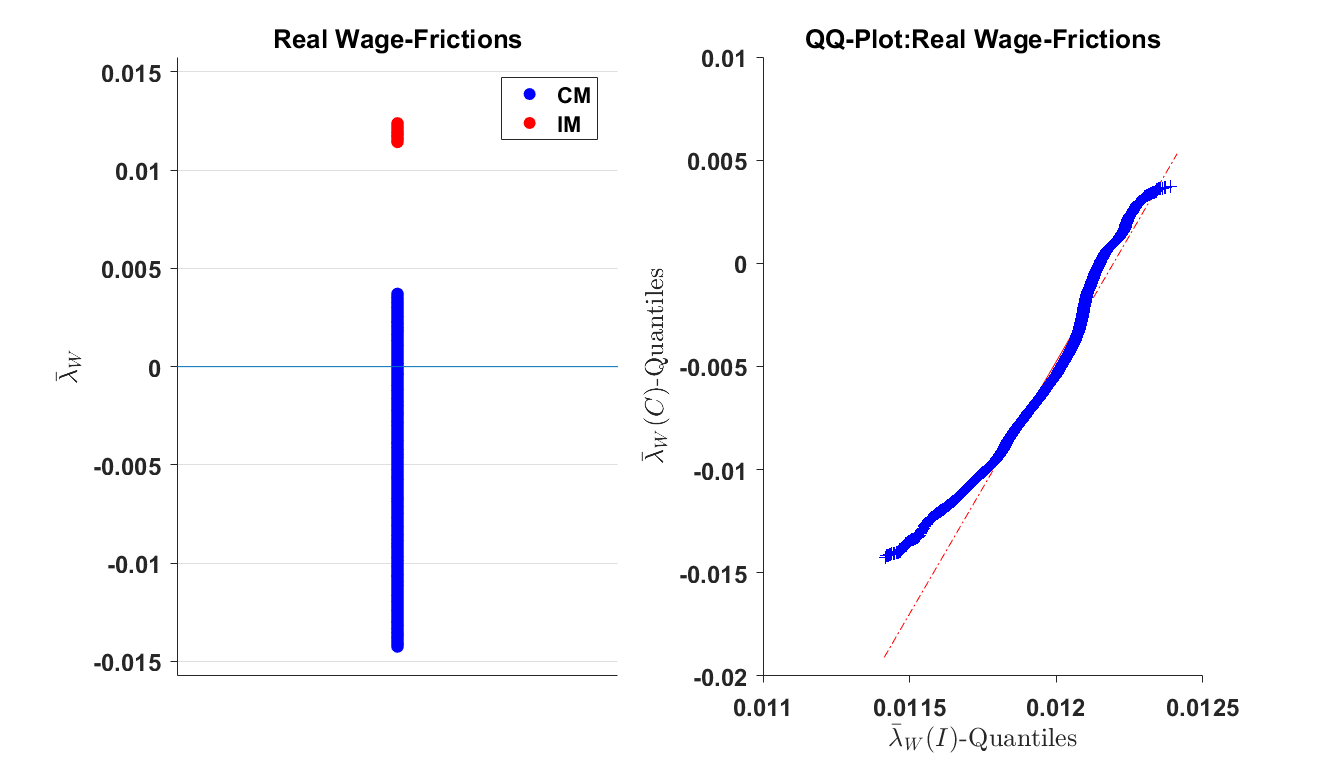}
	\protect\caption{Full Sample Wedges (due to Financial Frictions and Capital Adjustment Costs)}
\end{centering}
\end{figure}

\begin{figure}[H]
	\begin{centering}

		\includegraphics[scale=0.47]{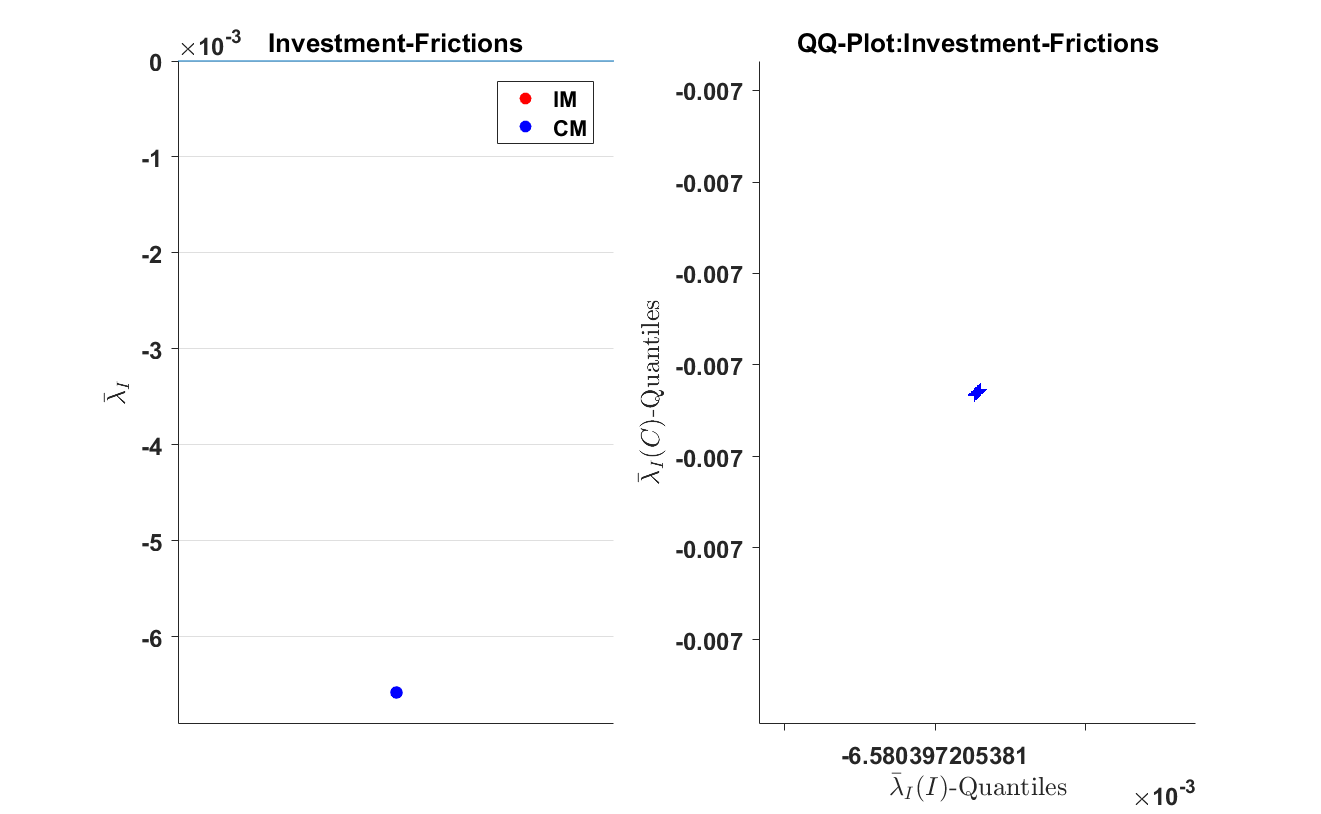}
			\protect\caption{Full Sample Wedges (due to Financial Frictions and Capital Adjustment Costs)}
	\end{centering}
\end{figure}

\newpage
\section*{Appendix B: Supplemental Material}

\section{Towards General Equilibrium - More examples}

\begin{example}
	\textbf{Capital Adjustment Costs}\newline
	Assuming full depreciation, the capital accumulation equation of the
	representative firm is distorted as follows: $K_{t+1}=I_{t}-\frac{\phi }{2}(%
	\frac{K_{t+1}}{K_{t}}-1)^{2}K_{t}$ for $\phi \in (0,1)$.
	
	Using the capital accumulation equation in the linearised Euler equation and
	imposing $\tilde{R}_{t}=\tilde{Z_{t}}-(1-\alpha )\tilde{K_{t}}$ we have: 
	\begin{eqnarray*}
		(\omega +\phi (1+\beta (1-\alpha ))+1-\alpha )\kappa _{t} &=&(\alpha -\phi
		(1-\beta \alpha ))(\alpha -1)\tilde{K}_{t}+(\omega +\beta \phi )\kappa
		_{t+1}-\phi \tilde{Z_{t}} \\
		&=&(\alpha -\phi (1-\beta \alpha ))\tilde{R}_{t}+(\omega +\beta \phi )%
		\mathbb{E}_{t}\kappa _{t+1}+... \\
		&&...-\alpha (1-\beta \phi )\tilde{Z_{t}}
	\end{eqnarray*}%
	where $\kappa _{t}$ is the Lagrange multiplier on the capital accumulation
	equation. Assuming that productivity is iid, and iterating forward we get: 
	\begin{equation*}
	\kappa _{t}=\frac{\gamma _{1}}{1-\gamma _{2}L^{-1}}\mathbb{E}_{t}\tilde{R_{t}%
	}+\gamma _{3}\tilde{Z}_{t}=\zeta (L,\gamma _{1},\gamma _{2})\mathbb{E}_t%
	\tilde{R_{t}}+\gamma _{3}\tilde{Z}_{t}
	\end{equation*}%
	where $\gamma _{1}\equiv \frac{\alpha -\phi (1-\beta \alpha )}{\omega +\phi
		(1+\beta (1-\alpha ))+1-\alpha }$, $\gamma _{2}=\frac{\omega +\beta \phi }{%
		\omega +\phi (1+\beta (1-\alpha ))+1-\alpha }$ and $\gamma _{3}=\frac{%
		-\alpha (1-\beta \phi )}{\omega +\phi (1+\beta (1-\alpha ))+1-\alpha }$ and $%
	L$ the lag operator. Using $\kappa _{t}=-\omega \tilde{C}_{t}$, and letting $%
	{s}:=\frac{C_{ss}}{I_{ss}}$, aggregate investment is: 
	\begin{equation*}
	\tilde{I}_{con,t}=(1+s)\tilde{Y_{t}}+\frac{s}{\omega }\zeta (L,\gamma
	_{1},\gamma _{2})\mathbb{E}_t\tilde{R_{t}}+\frac{s}{\omega }\gamma _{3}%
	\tilde{Z}_{t}
	\end{equation*}%
	When $\phi =0$, 
	\begin{equation}
	\tilde{I}_{t}^{\star }=(1+s)\tilde{Y}_{t}+\frac{s}{\omega }\zeta (L,\gamma
	_{1,\phi =0},\gamma _{2,\phi =0})\mathbb{E}_t\tilde{R_{t}}+\frac{s}{\omega }%
	\gamma _{3,\phi =0}\tilde{Z}_{t}  \label{eq:phi0}
	\end{equation}%
	Notice that $\zeta (L,\gamma _{1},\gamma _{2})$ is increasing in both $%
	\gamma _{1}$ and $\gamma _{2}$ and $\frac{d\gamma _{1}}{d\phi }<0$, $\frac{%
		d\gamma _{2}}{d\phi }<0$ for $\omega>1$ and $\frac{d\gamma _{3}}{d\phi }>0$.
	Therefore, $\zeta (L,\gamma _{1,\phi =0},\gamma _{2,\phi =0})-\zeta
	(L,\gamma _{1},\gamma _{2})>0,\forall \phi $ and the difference the
	investment rules with and without adjustment costs is: 
	\begin{eqnarray*}
		\tilde{\lambda}_{t} &\equiv &\tilde{I}_{con,t}-\tilde{I}_{t}^{\star } \\
		&=&\frac{s}{\omega }(\zeta (L,\gamma _{1,\phi =0},\gamma _{2,\phi =0})-\zeta
		(L,\gamma _{1},\gamma _{2}))\mathbb{E}_t\tilde{R_{t}}+\frac{s}{\omega }%
		(\gamma _{3,\phi =0}-\gamma _{3})\tilde{Z}_{t} \\
		&=&-\frac{s}{\omega }(\zeta (L,\gamma _{1,\phi =0},\gamma _{2,\phi
			=0})-\zeta (L,\gamma _{1},\gamma _{2}))(1-\alpha )\tilde{K}_{t}... \\
		&&...+\frac{s}{\omega }(\zeta (L,\gamma _{1,\phi =0},\gamma _{2,\phi
			=0})-\zeta (L,\gamma _{1},\gamma _{2})+(\gamma _{3,\phi=0}-\gamma _{3}))%
		\tilde{Z}_{t}
	\end{eqnarray*}%
	Hence, $\tilde{\lambda}_{t}$ is negatively related to $\tilde{K}_{t}$.
	Moreover, after some algebra it can be shown that the coefficient of $\tilde{%
		Z}_t$ is also bounded below by a positive number if $\alpha < \frac{1}{2\beta%
	}$. Nevertheless, the sign of the conditional mean of $\tilde{\lambda}_{t}$
	is determined and the following moment inequality holds: $\mathbb{E}\tilde{%
		\lambda}_{t}\tilde{K}_{t-j}\leq 0,\forall j\geq 0$.
\end{example}

\begin{example}
	\textbf{Occasionally binding constraints}\newline
	The case of occasionally binding constraints can be best motivated by
	attaching an aggregate marginal efficiency shock, $\epsilon _{t}$ to
	investment, that is, $k_{i,t}=(1-\delta )k_{i,t-1}+\epsilon _{t}i_{i,t}$.
	For simplicity we assume that this shock is iid and takes two values, $%
	\epsilon _{H}$ and $\epsilon _{L}$. We analyse the case of constraints on dis-investing (capital
	irreversibility), which can be thought of as a restriction on how much
	capital households can withdraw from the firm every period. 
	The optimization problem now includes a new constraint of the form $%
	K_{t}\geq \rho (1-\delta )K_{t-1}$ which is equivalent to $I_{t}\geq -\frac{%
		\tilde{\rho}}{\epsilon _{t}}K_{t}$ where $\tilde{\rho}\equiv {(1-\rho
		)(1-\delta )}$. Denoting by $\nu _{t}$ the Lagrange multiplier on this
	constraint, and $\kappa _{t}$ the Lagrange multiplier on the law of motion
	of capital, the relevant optimality conditions are: 
	\begin{eqnarray*}
		\kappa _{t}+\beta \mathbb{E}\kappa _{t+1}(\tilde{\rho}-(1-\delta ))+\mathbb{E%
		}C_{t+1}^{-\omega }(R_{t+1}+\frac{\tilde{\rho}}{\epsilon _{t+1}}) &=&0 \\
		\nu _{t}-\kappa _{t}\epsilon _{t}-C_{t}^{-\omega } &=&0 \\
		\nu _{t}(I_{t}+\frac{\tilde{\rho}}{\epsilon _{t}}K_{t}) &=&0
	\end{eqnarray*}%
	When $\epsilon _{t}=\epsilon _{H}$ the representative household will choose $%
	I_{H,t}^{\star }$ according to the following Euler equation, which we get by
	setting $\nu _{t}=0$: 
	\begin{equation*}
	C_{t}^{-\omega }=\mathbb{E}\beta C_{t+1}^{-\omega }\frac{\epsilon _{H,t}}{%
		\epsilon _{t+1}}(1-\delta +R_{t+1}\epsilon _{t+1})
	\end{equation*}%
	and linearizing we have: 
	\begin{equation*}
	-\omega \tilde{C}_{t}=-\mathbb{E}\tilde{C}_{t+1}+\tilde{\epsilon}_{H,t}+%
	\mathbb{E}\tilde{R}_{t+1}
	\end{equation*}%
	Solving the Euler equation forward, we get $\tilde{I}_{H,t}^{\star }=\tilde{I%
	}_{t}^{\star }+\tilde{\epsilon}_{H,t}$. When $\epsilon _{t}=\epsilon _{L,t}$%
	, the household dis-invests up to the irreversibility level, that is $%
	I_{L,t}=-\frac{\tilde{\rho}}{\epsilon _{t}}K_{t}$. The corresponding
	linearized rule is $\tilde{I}_{L,t}=-(\tilde{K}_{t}-\tilde{\epsilon}_{t})$
	where we have imposed that $\mathbb{E}\epsilon _{t}=1$. Therefore aggregate
	investment evolves as: 
	\begin{eqnarray*}
		\tilde{I}_{con,t} &=&\tilde{I}_{H,t}^{\star }\mathbb{P}(\epsilon
		_{t}=\epsilon _{H})+I_{L,t}\mathbb{P}(\epsilon _{t}=\epsilon _{L}) \\
		&=&\tilde{I}_{t}^{\star }-(1-\mathbb{P}(\epsilon _{t}=\epsilon _{H}))(\tilde{%
			I}_{H,t}^{\star }-\tilde{I}_{L,t})
	\end{eqnarray*}%
	and the distortion produced by the occasionally binding constraint is: 
	\begin{equation*}
	\tilde{\lambda}_{t}=\tilde{I}_{con,t}-\tilde{I}_{t}^{\star }=-(1-\mathbb{P}%
	(\epsilon _{t}=\epsilon _{H}))(\tilde{I}_{t}^{\star }+\tilde{K}_{t}-\tilde{%
		\epsilon}_{L}+\tilde{\epsilon}_{H})
	\end{equation*}%
	Here $\tilde{\lambda}_{t}$ is negative as by definition $\tilde{\epsilon}_{H}>\tilde{\epsilon}_{L}$. The moment inequality implied by this friction is $%
	\mathbb{E}\tilde{\lambda}_{t}\tilde{K}_{t-j}\leq 0,\quad \forall j\geq 0$.
\end{example}

\begin{example}
	\textbf{Non Rational Expectations}\newline
	When agents employ different models to make predictions, have misperceptions
	or sentiments, the sign of $\mathbb{E}\tilde{\lambda}_{t}\tilde{K}_{i,t-j}$ depends
	on how the model used by the agent relates to the objective probability
	measure. Suppose, for illustration, that agents are unaware and unable to
	estimate the stochastic process for productivity. Suppose that the true
	process is $\tilde{Z}_{t}=\varepsilon _{t}$ where $\varepsilon _{t}\sim
	N(0,1)$. Agents use output realizations to predict future productivity, $%
	\mathcal{E}_{t}\tilde{Z}_{t+j}=\rho ^{j}\tilde{Y}_{t}$ for $|\rho |<1$,
	since $Corr(Z_{t},Y_{t})=\frac{1}{1+\alpha ^{2}\mathbb{V}(K_{t})}>0$. Using
	this conditional expectation in the investment rule 2.1, aggregate
	investment is: 
	\begin{equation*}
	\tilde{I}_{con,t}=A_{1}(\theta )\tilde{K}_{t}+A_{2}(\theta )(1-A_{3}(\theta
	)\rho )^{-1}(\alpha \tilde{K}_{t}+\tilde{Z}_{t})
	\end{equation*}%
	while the investment rule used by the econometrician after substituting the
	true process for $\tilde{Z}_{t}$ in \eqref{eq:agg_inv} is: 
	\begin{equation}
	\tilde{I}_{t}=A_{1}(\theta )\tilde{K}_{t}+A_{2}(\theta )\tilde{Z}_{t}
	\label{eq:agg_inv_ex3}
	\end{equation}%
	The difference between the investment rule under bounded rationality and the
	one used by the econometrician is therefore: 
	\begin{equation*}
	\tilde{\lambda}_{t}=A_{2}(\theta )(1-A_{3}(\theta )\rho )^{-1}(\alpha \tilde{%
		K}_{t}+\rho A_{3}(\theta )\ \tilde{Z}_{t})
	\end{equation*}%
	Assuming full capital depreciation and using the true process for
	productivity in equation \eqref{eq:agg_inv_ex3}, which is identical to
	equation \eqref{eq:phi0} once we substitute for the production function and
	the return to capital, we get that $A_{2}(\theta )=1+\frac{C_{ss}}{I_{ss}}>0$
	and $A_{3}(\theta )=\frac{s\alpha }{(1+s)(\omega +1-\alpha )}<1$. Therefore, 
	$\tilde{\lambda}_{t}$ is a positive function of $\tilde{K}_{t}$ and $\tilde{Z%
	}_{t}$, and the moment inequality in this case is, $\mathbb{E}\tilde{\lambda}%
	_{t}K_{t-j}\geq 0,\quad \forall j\geq 0$.
	
\end{example}

\section{Perturbing the frictionless model}
\subsection{Obtaining unique conditional models}
For every $%
\theta \in \Theta $, there is a multiplicity of corresponding conditional
models which can be constructed with different distributions of shocks but
give rise to the same $\mathbb{E}\lambda _{t}$.

To construct this family of models, we use the fact that the process $%
\lambda _{t}$ causes a change in the measure implied by the frictionless
model. As \citet{Hansen_nobel}, we define a perturbation $\mathcal{M}_{t}$
such that for any measurable random variable $W_{t}$: $\mathbb{E}_{t}%
\mathcal{M}W_{t}=\mathbb{E}_{t}(W_{t}|M_{f})$ (and vice versa, that is $%
\mathbb{E}_{t}W_{t}=\mathbb{E}_{t}(\tilde{\tilde{\mathcal{M}_t}}W_{t}|M_{f})$
). This representation is useful for two reasons. First, as we show in the
proof of Proposition 6, interpreting the distortions as a change of measure
provides a unified way of looking at frictions. Second, we can compute $%
\mathcal{M}_{t}$, for all t, and this can give us estimates of the wedges at
any point of time.

We briefly explain how one can compute $\mathcal{M}_{t}$ - a full
description is in the Appendix. Recall that in the linearised model, $\tilde{%
	\lambda}_{t}$ measures the distance between the prediction of the
frictionless model and the data, where the latter is assumed to be produced
by a model with frictions. As shown in section 2, $\tilde{\lambda}_{t}$ is a
function of the endogenous variables and the shocks. Without loss of
generality assume that $\mathbb{E}\tilde{\lambda}_{t}>0$. We look for a $%
\mathcal{M}_{t}$ that makes this expectation zero. By finding a $\mathcal{M}%
_{t}$ such that $\mathbb{E}\mathcal{M}_{t}\tilde{\lambda}_{t}=0$ we are
identifying $\mathcal{M}_{t}d\mathbb{P}(.)$, which is the density of the
frictionless model that can be derived from the data by distorting the
objective distribution, $\mathbb{P}$. Given $\mathcal{M}_{t}$, we can
decompose $\mathbb{E}\tilde{\lambda}_{t}$ as: $\mathbb{E}\tilde{\lambda}%
_{t}\equiv \mathbb{E}\mathcal{M}_{t}\tilde{\lambda}_{t}+\mathbb{E}(1-%
\mathcal{M}_{t})\tilde{\lambda}_{t}$. Therefore, to impose $\mathbb{E}%
\mathcal{M}_{t}\tilde{\lambda}_{t}=0$ it suffices to set $\mathbb{E}(1-%
\mathcal{M}_{t})\tilde{\lambda}_{t}=\mathbb{E}\tilde{\lambda}_{t}$. The term 
$1-\mathcal{M}_{t}$ determines the distortion at time $t$. To understand why
using this decomposition is useful, notice that $\tilde{\lambda}_{t}$ is
related to endogenous frictions but also to the unobservable shocks. As we
show below, $1-\mathcal{M}_{t}$ is a time varying function of $\mathbb{E}%
\tilde{\lambda}_{t}$. Since the latter is an average, unobservable shocks
are eliminated and $1-\mathcal{M}_{t}$ captures only the endogenous
frictions.

In general, $\mathcal{M}_{t}$ is a positive $\mathcal{F}_{t}-$ measurable random variable, unit expectation martingale,  $\mathbb{E}(\mathcal{M}_{t+1}|\mathcal{F}_{t})=\mathcal{M}_{t}$, $\mathbb{E}\mathcal{M}_{t}=1$. Being a martingale is a necessary condition for the distorted conditional expectation to be consistent with the Kolmogorov definition \citep{HansenSargent2005}.  We stack all $\mathcal{M}_t$ in a vector $\mathcal{M}$ and define $\tilde{\mathcal{M}}=\mathbf{1}-\mathcal{M}$ where $%
\mathbf{1}$ is the unit vector, which corresponds to the frictionless steady
state of $\mathcal{M}_t$. The vector $\mathcal{M}$ satisfies the following
program, where bold letters indicate vectors, $\tilde{\lambda}(Y_t;\theta )$
is the matrix containing the distortions for every $t$ for every variable $j$%
, $\tilde{\lambda}_{j,t}$ and $(\pi _{1},\pi _{2},\pi _{3})$ are the
corresponding Lagrange multiplier vectors.
Moreover, there are multiple $\mathcal{M}_{t}$ that satisfy the
restriction $\mathbb{E}(1-\mathcal{M}_{t})\tilde{\lambda}_{t}=\mathbb{E}%
\tilde{\lambda}_{t}$. This is exactly what we mean by a multiplicity of
conditional models. In order to get a unique conditional model, we need to
impose more restrictions on the stochastic behavior of $\mathcal{M}_{t}$ and
therefore $\tilde{\lambda}_{t}$. 

To do this, we introduce a pseudo-distance
metric $d(\mathcal{M}_{t})$, which we minimize subject to the restriction $%
\mathbb{E}(\mathcal{M}_{t}-1)\tilde{\lambda}_{t}=\mathbb{E}\tilde{\lambda}%
_{t}$. The choice of the metric depends on the modellers' beliefs of the
distribution of the shocks. 

We first state the optimization problem for a general distance \small $d(\mathcal{\tilde{M}})$ 
\begin{eqnarray*}
	\max_{\mathcal{M}} && - d(\mathcal{\tilde{M}}) \\
	\text{subject to} &&\mathbf{1}^{\mathbf{T}}\mathcal{\tilde{M}} = 0\quad (\lambda _{1}) \\
	&&\tilde{\mathcal{M}}^{\mathbf{T}}q_{j}(Y;\theta )+\mathbf{1}^{\mathbf{T}}q_{j}(Y;\theta )=0,\quad j=1,..p\quad ,(\lambda _{2,j}) \\
	&&\tilde{\mathcal{M}}{}^{\mathbf{T}}q_{j}(Y;\theta )+\left[ \mathbf{1}^{%
		\mathbf{T}}q_{j}(Y;\theta )\right] _{+}=0,\quad j=p+1,..r \\
	&&\mathcal{M\geq }\mathbf{0} \quad(\lambda _{3,t},\forall t\in(1..T)
\end{eqnarray*} \normalsize

where $q(Y;\theta )$ is the matrix containing the moment functions and $%
(\lambda _{1},\lambda _{2},\lambda _{3})$ the corresponding Lagrange multiplier vectors. The first constraint imposes unit expectation while the
last constraint imposes non negativity of $\mathcal{M}$. The rest of the constraints
impose the moment equalities and inequalities e.g. $[x]_{+}\equiv \max (x,0)$
Denote by $\tilde{d}(\mathcal{M}_{t})$ as the inverse function of $d(\mathcal{M}_{t})$. The Kuhn Tucker first order necessary conditions are the following \small 

\begin{eqnarray*}
	\mathcal{M}^{*} &=&\mathbf{q(Y;\theta)}\lambda_{2}+\lambda_{1}+\mathbf{\lambda_{3}}\\
	\tilde{d}(\mathbf{q(Y;\theta)}\lambda_{2}+\lambda_{1}+\mathbf{\lambda_{3}}){}^{T}\mathbf{q(Y;\theta)}+\mathbf{1}^{T}\mathbf{q(Y;\theta)} & = & 0,\;(\lambda_{2})\\
	\mathbf{1}^{T}\tilde{d}(\mathbf{q(Y;\theta)}\lambda_{2}+\lambda_{1}+\mathbf{\lambda_{3}}) & = & 0,\;(\lambda_{1})\\
	\lambda_{3,t} &\geq& 0\\
	\lambda_{3,t}(\tilde{d}(q_{t}(Y,\theta)\lambda_{2}+\lambda_{1}+\lambda_{3,t})) & = & 0,\;(\lambda_{3,t})
\end{eqnarray*} \normalsize

Below we illustrate what happens when 
$d(\mathcal{M}_{t})\equiv \frac{1}{2}(\mathcal{M}-1)^{T}(\mathcal{M}-1)$ (chi square distance) which is consistent with the shocks having finite second moments 

\footnote{
	For any random variable $x$ and distorting density $\mathcal{M}_{t}$ by
	Cauchy Schwartz we have that $(\int x_{t}\mathcal{M}_{t}d\mathbb{P}%
	)^{2}\leq \int x_{t}^{2}d\mathbb{P}\int \mathcal{M}_{t}^{2}d%
	\mathbb{P}$. Minimizing the variance of the second term assumes that
	the variance of the first term exists.}. Note that this distance metric is
also used in the classic mean-variance frontiers in portfolio choice theory,
or to compute Hansen-Jagganathan bounds. Intuitively, the minimization
implies that we look for $\mathcal{M}_{t}$ that is consistent with our
moment restrictions and has minimum variance \footnote{%
	This is similar to the approach in generalized empirical likelihood settings
	in econometrics (i.e. \citet{ECTA:ECTA482}).}. Restricting the distribution
of the shocks pins down a unique conditional model corresponding to a value
of $\theta $.

Ignoring the non - negativity constraint we get an analytical solution. Solving the dual problem and
concentrating out the first constraint leads to solutions $(\mathcal{M}%
^{\ast },\lambda _{2}^{\ast })$ that satisfy the following system: 
\begin{equation*}
\left[ 
\begin{array}{cc}
I_{T} & \underset{T\times r}{-(q(Y;\theta )-\bar{q}(Y;\theta ))} \\ 
\underset{r\times T}{q(Y;\theta )^{\mathbf{T}}} & \underset{r\times r}{%
	\mathbf{0}}%
\end{array}%
\right] \left[ 
\begin{array}{c}
\underset{T\times 1}{\tilde{\mathcal{M}}} \\ 
\underset{r\times 1}{\lambda _{2}}%
\end{array}%
\right] =\left[ 
\begin{array}{c}
\underset{T\times 1}{0} \\ 
\begin{array}{c}
-\underset{p\times 1}{\mathbf{1}^{\mathbf{T}}q_{1}(Y;\theta )} \\ 
\underset{(r-p)\times 1}{-\left[ \mathbf{1}^{\mathbf{T}}q_{2}(Y;\theta )%
	\right] _{+}}%
\end{array}%
\end{array}%
\right] 
\end{equation*}%
We therefore have that for $\tilde{q}(Y;\theta )\equiv q(Y;\theta )-\bar{q}%
(Y;\theta )$ and $\mathcal{V}(q(Y_{t};\theta ))\equiv T^{-1}\tilde{q}%
(Y;\theta )^{\mathbf{T}}q(Y;\theta )$ 
\begin{eqnarray*}
	\left[ 
	\begin{array}{c}
		\tilde{\mathcal{M}} \\ 
		\underset{r\times 1}{\lambda _{2}}%
	\end{array}%
	\right] &=&\left[ 
	\begin{array}{c}
		\tilde{q}(Y;\theta )\mathcal{V}(q(Y;\theta )){}^{-1}\left[ 
		\begin{array}{c}
			\bar{q}_{1}(Y;\theta ) \\ 
			{}[\bar{q}_{2}(Y;\theta )]_{+}%
		\end{array}%
		\right]  \\ 
		-\mathcal{V}(q(Y;\theta )){}^{-1}\left[ 
		\begin{array}{c}
			\bar{q}_{1}(Y;\theta ) \\ 
			{}[\bar{q}_{2}(Y;\theta )]_{+}%
		\end{array}%
		\right] 
	\end{array}%
	\right] 
\end{eqnarray*}%

Therefore, 

\begin{eqnarray}  \label{eq:Mtilde}
\tilde{\mathcal{M}}& =& w_{t}\left[ 
\begin{array}{c}
0 \\ 
{}[\bar{q}_{2}(Y;\theta )]_{+} \\ 
\end{array}%
\right]
\end{eqnarray}%

where $w_t\equiv\frac{\hat{Z}^2_t+{\hat\lambda}_t(\hat{\lambda}_t-\mathbb{E}%
	\hat{\lambda}_t)}{\mathbb{V}(\hat{Z}_t)+\mathbb{V}(\hat{\lambda}_t))}$. The
optimal $\mathcal{M}_t$ is a time varying function $w_t$ of the average
distortion over the sample, $ \bar{q}_{2}(Y;\theta) $. The weight $w_t$
is a function of the relative variability of $\hat\lambda$, which is a
function of the endogenous variables, and of $\hat{Z}_t$, which is function
of the shocks $Z_t$.

The solution above has been derived ignoring the non negativity constraint. Looking at $\tilde{\mathcal{M}_t}$
the constraint is violated with positive probability, since $\tilde{q}(Y;\theta)$ can take values lower than minus one. 
Taking into account the non-negativity constraint implies a non analytical solution. There is a variety of algorithms in 
quadratic optimization to deal with this issue. An alternative way is to use a penalty function that penalizes negative values
of $\mathcal{M}_t$. This also typically implies non-closed form solutions. For the adjustment cost example, we re-computed the estimated  $\mathcal{M}_t$ and we report it in Appendix B. As is evident, violations of the constraint can be minimal.
The existence of a unique function $\mathcal{\tilde{M}}\ast (\theta )$
implies that the set of models consistent with moment inequalities should
have a corresponding one-to-one relation to the identified set, the subset of $\Theta$ that satisfies 
those inequalities: {$\theta\in\Theta:\mathbb{E}q_{1}(Y,\theta)=0, \mathbb{E}q_{2}(Y,\theta)\geq 0$}.
This is made clear in Figure 9.1, which
depicts two theory-based moment inequality restrictions on the Euclidean
parameter space. The darker area is the identified set, and for the sake of
illustration, the point of intersection of the two lines is the combination $%
(\theta _{1},\theta _{2})$ that corresponds to the pseudo-true parameter
values of the case of no perturbation $(\mathcal{M}_{t}=1)$. The identified
set contains the true value, which maps one-to-one to the set of admissible
perturbations $\tilde{\mathbb{M}}_{2}$.

Moreover, choices of objective
functional other than $\frac{1}{T}\sum_{t\leq T}\mathcal{\tilde{M}}_{t}^{2}$
leads to different sorts of distortions. A general family of distances that
can account for non-linearities or non-normalities  is the Cressie -
Read divergence, of which Chi square is a special case (%
\citet{almeidagarcia2014,cressie_read}). As in the case of non-negative constraints, 
computing the multipliers might involve numerical optimization.  It is also important
to stress that for any choice of distance functional, the moment inequality constraints are
satisfied. Therefore, the choice of distance functional does not affect the consistency of the 
parameter estimates.

\section{Identification and Estimation }
\begin{example}
	\textbf{Identification in the case of capital adjustment costs}
	
	We assume that the representative firm faces adjustment costs with iid
	probability $B_{t}$ which is a random variable with mean $B$. As long as $B$
	is positive, the conditional mean of $\tilde{\lambda}_{t}$ is the same as
	the one derived in example 1. We focus on identification using the aggregate
	capital accumulation equation. Denote the solution of the frictionless model
	as: $K_{t}=\varphi _{k}(\theta )K_{t-1}+\varphi _{z}(\theta )Z_{t}$.
	Therefore, $\mathbb{E}_{t-1}K_{t}=\varphi _{k}(\theta )K_{t-1}$. Let $\zeta
	_{t}$ denote an instrument and $(\zeta ,K,K_{-1},B)$ the vectors containing
	data on $(\zeta _{t},K_{t},K_{t-1},B_{t})$. The two identifying conditions
	are: 	
	\begin{eqnarray}
	\mathbb{E}\mathbf{\zeta }^{T}(K-\varphi _{k}(\theta )K_{-1})
	&=&v_{1} \\
	\mathbb{E}\mathbf{\zeta }^{T}(K-\varphi _{k}(\theta )K_{-1}+|\lambda_{1}|K_{-1}\odot\mathcal{B})
	&=&v_{2}
	\end{eqnarray}
	\normalsize
	where $v_{1}\leq 0, v_{2}\leq 0$. Rearranging the first equation and letting $\hat{\varphi }_{k}\equiv \mathbb{E}%
	(\zeta ^{T}K_{-1}){}^{-1}\mathbb{E}\zeta ^{T}K,$ we get a lower bound for $%
	\varphi _{k}(\theta )$: 
	\begin{equation*}
	\varphi _{k}(\theta )=\hat{\varphi }_{k}-\mathbb{E}(\zeta
	^{T}K_{-1}){}^{-1}v_{1}\geq \hat{\varphi }_{k}
	\end{equation*}%
	Similarly, from the second equation, 
	Letting: \\ $\hat{{\varphi }}_{s}\equiv |\lambda_{1}|\mathbb{E}(\zeta ^{T}K_{-1})^{-1}\mathbb{E}(\zeta ^{T}(K_{-1}\odot\hat{\mathcal{B}}))$, we have: 
	\begin{equation}
	\varphi _{k}(\theta )=\hat{\varphi }_{k}+\hat{\varphi }_{s}-\mathbb{E}(\zeta ^{T}K_{-1})^{-1}v_{2}\geq 
	\hat{\varphi }_{k}+\hat{\varphi }_{s}  \label{eq:super_id_phi}
	\end{equation}%
	Clearly, as long as $B\in(0,1)$, \eqref{eq:super_id_phi} is more informative, for any $\lambda_{1}$. 	
	\end{example}

\subsection{Using Bootstrap to do Model Validation}
To get a better approximation to the finite sample distribution of
the test statistic, we propose the use of a suitable version of bootstrap. Given
a bootstrap sample $\{Y_{t,l}^{*}\}_{t\leq T,l\leq B}$ obtained with
a block bootstrap scheme we can compute the wedges to each equation,
using the plug-in estimate of $\theta$ under the survey-robust case
and the full model. We consider the re-centered bootstrapped moments,
$(\tilde{f}_{1}(Y*;\theta),\tilde{f}_{2}(Y*;\theta)..\tilde{f}_{k}(Y*;\theta))$
where $\tilde{f}_{j}(Y*;\theta)\equiv{f}_{j}(Y*;\theta)-\bar{f}_{j}(Y*;\theta)$.
We choose to recenter the moments since we deal with an over-identified
case, and therefore sample moments, $\bar{q}(Y;\theta)$, are not
exactly equal to zero. We obtain critical values by computing the
$(1-\alpha)-$quantile of $T\mathcal{W}*(\theta_{p},\Theta_{s})$.
That is, $c_{\alpha}$ is chosen such that $\mathbb{P}_{T}(T\mathcal{W}*(\theta_{p},\Theta_{s})<c_{\alpha})  =  1-\alpha$.
We therefore have that, given the uniform consistency of the bootstrap :
\begin{enumerate}
	\item Under $H_{0}:$ $\underset{T,B\to\infty}{p\lim}\mathbb{P}_{T}(T\mathcal{W}*(\theta_{p},\Theta_{s})<c_{\alpha}|Y_{t,t\leq T})= \mathbb{P}(T\mathcal{W}(\theta_{p},\Theta_{s})<c_{\alpha})= 1-\alpha$
	\item Under $H_{1}:$ $\underset{T,B\to\infty}{p\lim}\mathbb{P}_{T}(T\mathcal{W}*(\theta_{p},\Theta_{s})<c_{\alpha}|Y_{t,t\leq T})=\mathbb{P}(T\mathcal{W}*(\theta_{p},\Theta_{s})<c_{\alpha})=0$
\end{enumerate}
We illustrate below an example with which we show how the bootstrap behaves in large samples. Small sample distortions is an 
interesting topic to pursue, but is the subject af another paper. We use a regression based example, which is unrelated to survey data as 
such, but has the same econometric structure. 
\begin{example}\textbf{Measurement error in Regressors}
	Suppose there one independent measurement of a regressor and the model for the measurement error is $X_{1,t}=X_t^{*}+\nu_{1,t}$. Furthermore assume that $\nu_{1,t}\sim\mathcal{N}(0,0.2^2) $ and $\epsilon_t\sim\mathcal{N}(0,0.1)$. 
	The important assumption in this case is that she  "knows" all parameters apart from $\beta$  and that she \textbf{mistakenly assumes} that $\sigma_{\nu_{1,t}}=0.5$. 
	She uses Simulated Maximum likelihood to estimate $\beta_{m}$ where $Y_t = 0.2 +\beta X_t + \epsilon_t$. 
	A robust approach would be to be agnostic about the distribution of the errors and use the well known fact that
	$B_{0}=\left\{ \beta\in B: \beta_{ols}\leq \beta \right\}$. We can in principle use also the reverse regression to derive an upper bound but it is not necessary for our purposes. We therefore test $H_{0}:\beta_{m}\in B_0$.  
	
	To see the equivalence of this test to the test we propose, notice that under the Null \\$T^{-\frac{1}{2}}\inf_{\beta_0 \in B_0} (X^{T}(Y-\hat\beta_{0} X)-X^{T}(Y-\hat\beta_m X)$ is equal to $T^{-\frac{1}{2}}X^{T}X(\hat\beta_m -\beta_m -(\hat\beta^{*}_{0}-\beta^{*}_{0}))$. The residual in this case, $\sum_t \tilde{q}_2(.)$ is equal to $(\frac{X^{T}X}{T})^{-1}X^{T}\epsilon+\beta_{0}\frac{X^{T}X^{*}}{X^{T}X} + \frac{\tilde{X}^{T}\tilde{X}}{T})^{-1}\tilde{X}^{T}\epsilon+\beta_{0}\frac{\tilde{X}^{T}X^{*}}{\tilde{X}^{T}\tilde{X}}$. We plot below the bootstrap distribution versus the simulated test statistic, 
	which in this case according to our theoretical result is, for $\lambda_{x_1,x_2}\equiv\frac{\sigma_{x_1}^2}{\sigma_{x_2}^2}$, $TW\sim(\beta^{2}((1-\lambda_{x,x^{*}})(\lambda_{x,x^{*}}-2\lambda_{x,\tilde{x}})+\lambda_{x,\tilde{x}}(1-\lambda_{x,\tilde{x}}))+\lambda_{\epsilon,x}(\lambda_{x,x^{*}}-\lambda_{x,\tilde{x}})\chi^{2}(1)$
	
	\begin{figure}[H]
		\begin{centering}
			\includegraphics[scale=0.4]{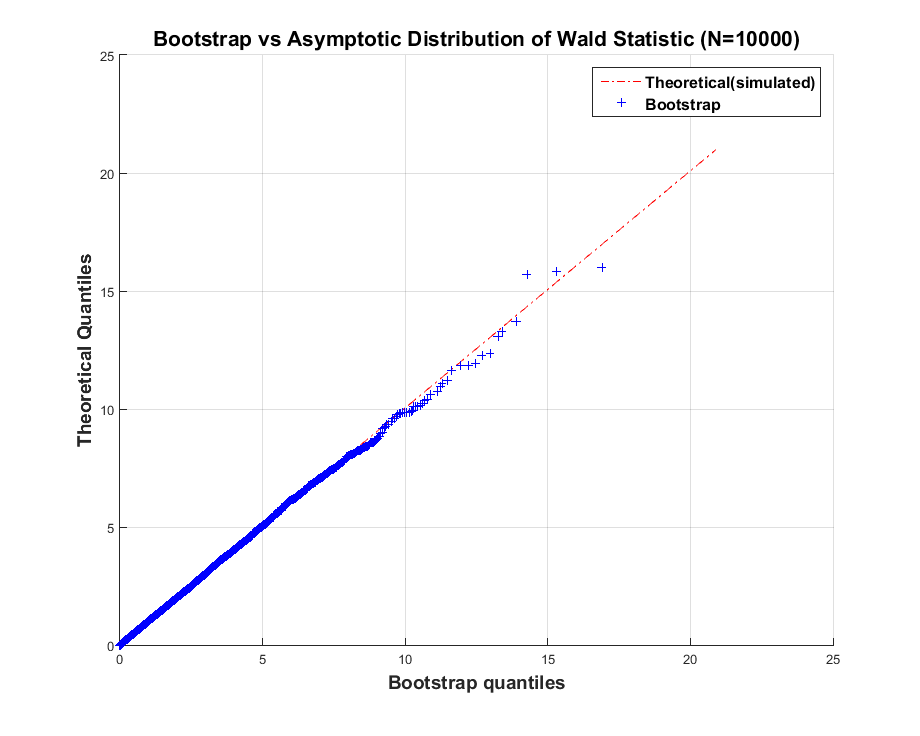}
			\par\end{centering}
		\begin{centering}
			\protect\caption{\small Q-Q plot of Bootstrap versus Asymptotic distribution of $W^{*}$}
		\end{centering}
	\end{figure}
\end{example}

\newpage
\subsection{\textbf{Spanish Survey Data}}

\small Below, we plot the aggregate response to production constraints due to financial issues in the industrial sector. 
\begin{figure}[H]
	\begin{centering}
		\includegraphics[scale=0.45]{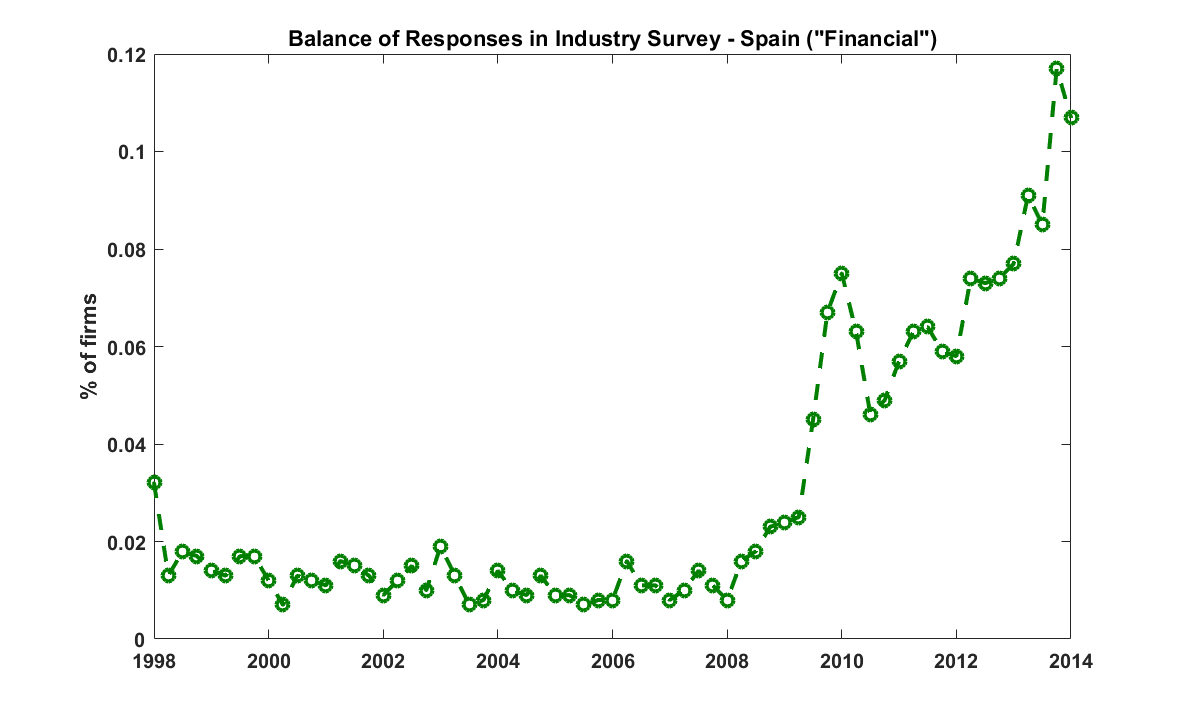}
		\par
	\end{centering}
	\protect\caption{}
\end{figure}
\newpage
\begin{table}[H]
	\begin{centering}
		\caption{Confidence Set for $\Theta_{IM}$ and $\Theta_{CM}$ - without survey data}\bigskip
		\begin{tabular}{|c||c|c||c|c||}
			\hline
			&  \textbf{Incomplete Model}&& \textbf{Complete Model}\\\hline 
			Parameter & $q_{2.5\%}$ & $ q_{97.5\%}$ & $q_{5\%}$ & $ q_{95\%}$ \\ \hline 
			$\sigma_c$&2.1350 & 2.3280&28.6900 & 30.0000 \\ \hline 			
			$\phi_{p}$& 6.4240 & 6.7350&21.5700 & 21.7000\\ \hline			
			$g_{y}$&0.3112 & 0.4583&0.2432 & 0.4827\\  \hline
						
			$\lambda_{h}$ &0.6327 & 0.6674&0.0009& 0.0009\\  \hline	
					
			$\sigma_{l}$& \textit{0.0000} & 0.0461&8.892 & 9.104\\ \hline	
					
			$\xi_{p}$&0.9913 & 0.9999&0.0000& 0.0000 \\ \hline		
				
			$\xi_{w}$& 0.9504 & 0.9883&0.8173 & 0.8628 \\  \hline
			
			$\iota_{w}$& \textit{0.0000} &  \textit{0.0000} &0.9748 & 0.9999\\ \hline
			
			$\iota_{p}$&  0.9949 & 0.9999&0.0004 & 0.0004\\ \hline	
					
			$z_{1}$	&0.7953 & 0.8228&0.0009 & 0.0009\\ \hline	
					
			$\rho_{a}$	&0.3973 & 0.4602&0.7521 & 0.9092 \\ \hline	
					
			$\rho_{b}$ &\textit{0.0000}&\textit{0.0000}&\textit{0.0000}&\textit{0.0000}\\ \hline		
				
			$\rho_{g}$ &0.2649 & 0.3728&0.9999 & 0.9999\\ \hline
						
			$\rho_{qs}$&\textit{0.0000} & 0.05588&0.3242 & 0.4079\\ \hline	
					
			$\rho_{ms}$	&0.0567 & 0.2612&0.7347 & 0.8015\\ \hline
			
			$\rho_{p}$&0.5991 & 0.6485&0.6329 & 0.7077\\ \hline
						
			$\rho_{w}$ &\textit{0.0000}& 0.1314&0.0003 & 0.0003\\\hline	
			
			$ma_{p}$&0.9753 & 0.9999&0.0129 & 0.0561\\\hline	
								
			$ma_{w}$& 0.4935 & 0.4935 &0.0010 & 0.0011\\\hline
			
			$\nu$&0.8040 & 0.8190&0.0053 & 0.0205\\\hline
						
			$\rho_y\star$&0.4162 & 0.5462&0.1683 & 0.2251\\ \hline		
				
			$\sigma_a$&8.8820 & 8.9960&0.0245 & 0.0339\\ \hline	
					
			$\sigma_b$& 0.0952 & 0.1342&3.4480 & 3.4800\\ \hline
			
			$\sigma_g$&19.930 & 19.999&4.5820 & 4.6340\\\hline
			
			$\sigma_{qs}$&7.0010 & 7.1120&0.1831 & 0.2064\\\hline
			
			$\sigma_m$&0.4271 & 0.4933&1.0250 & 1.0670\\\hline
			
			$\sigma_p$&0.0730 & 0.1555&4.2260 & 4.2660\\\hline
			$\sigma_w$ &-&-&- & -\\\hline
			$\sigma_f$ &- & - &- & -\\\hline
			$\phi_{adj}$&-&-&5.9210 & 5.9550\\\hline
			$\chi_{pr}$&-&-&0.9862 & 0.9999\\\hline
			$\gamma$&-&-& 0.9944 & 0.9950\\\hline
			$K/N$&-&-&4.3330 & 4.3490\\\hline
		\end{tabular}
		\par\end{centering}
\end{table}

Note: $\sigma_w$ and $\sigma_f$ are set to 1. Parameters in italic font reached the boundary and thus calibrated to the boundary value.

\newpage

\section*{\textbf{Estimated Wedges-No survey data}- $5\%$ level} \begin{figure}[H]
	\begin{centering}
		\includegraphics[scale=0.47]{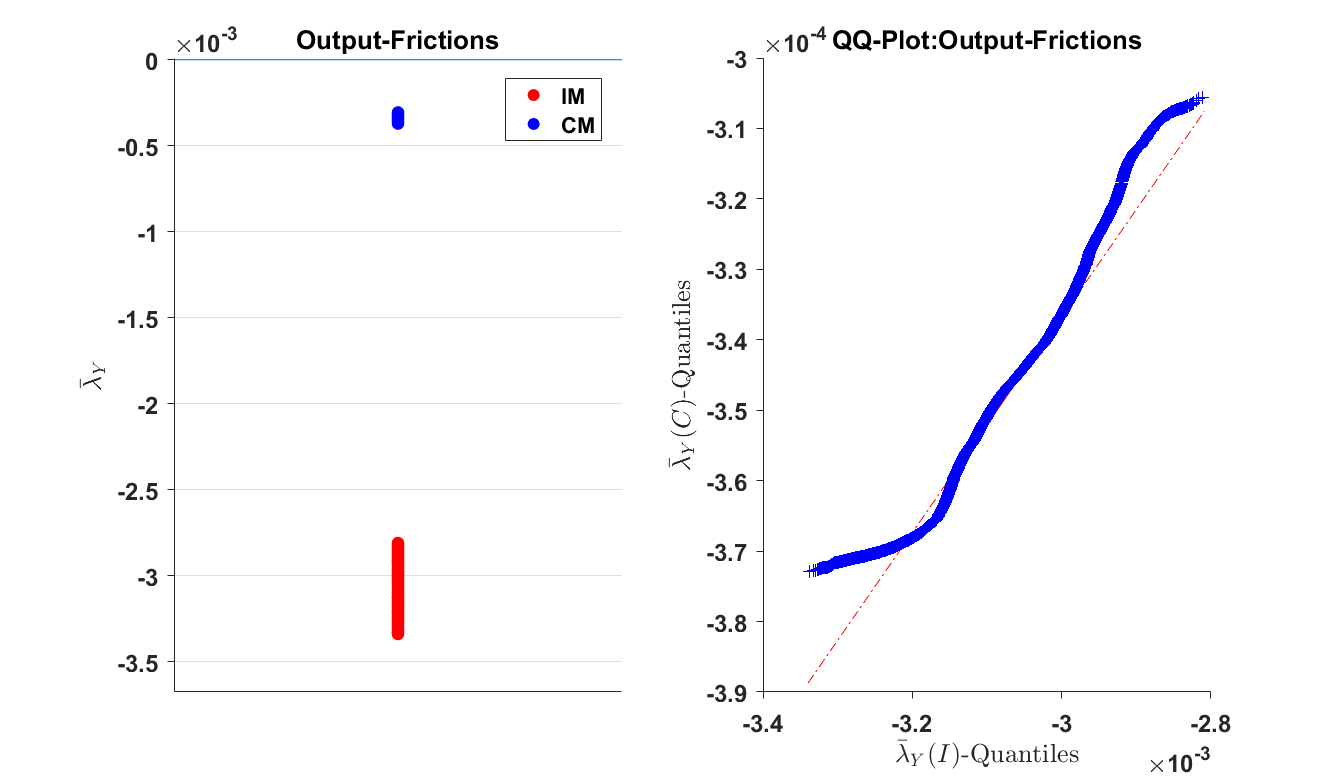}
		\includegraphics[scale=0.47]{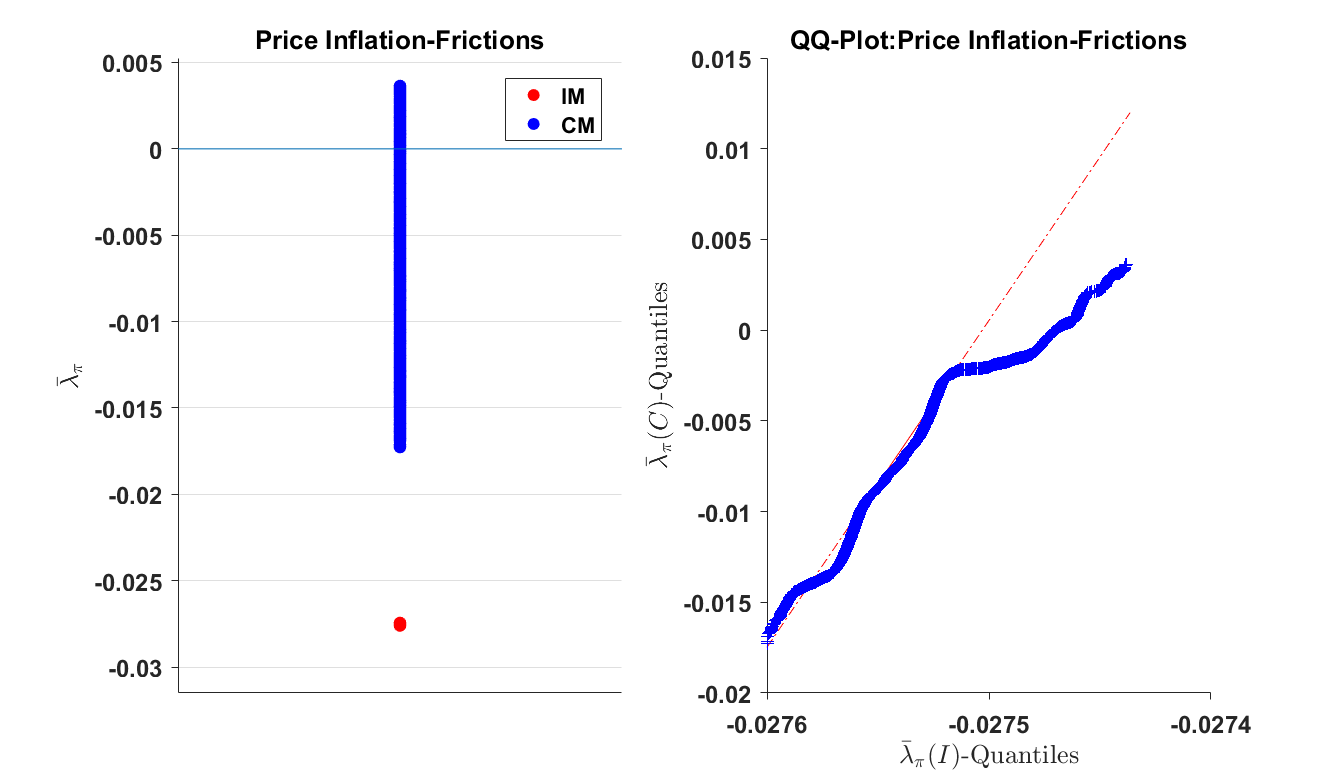}
		\protect\caption{\small Full Sample Wedges (due to Financial Frictions and Capital Adjustment Costs)}
	\end{centering}
\end{figure}

\begin{figure}[H]
	\begin{centering}
		\includegraphics[scale=0.47]{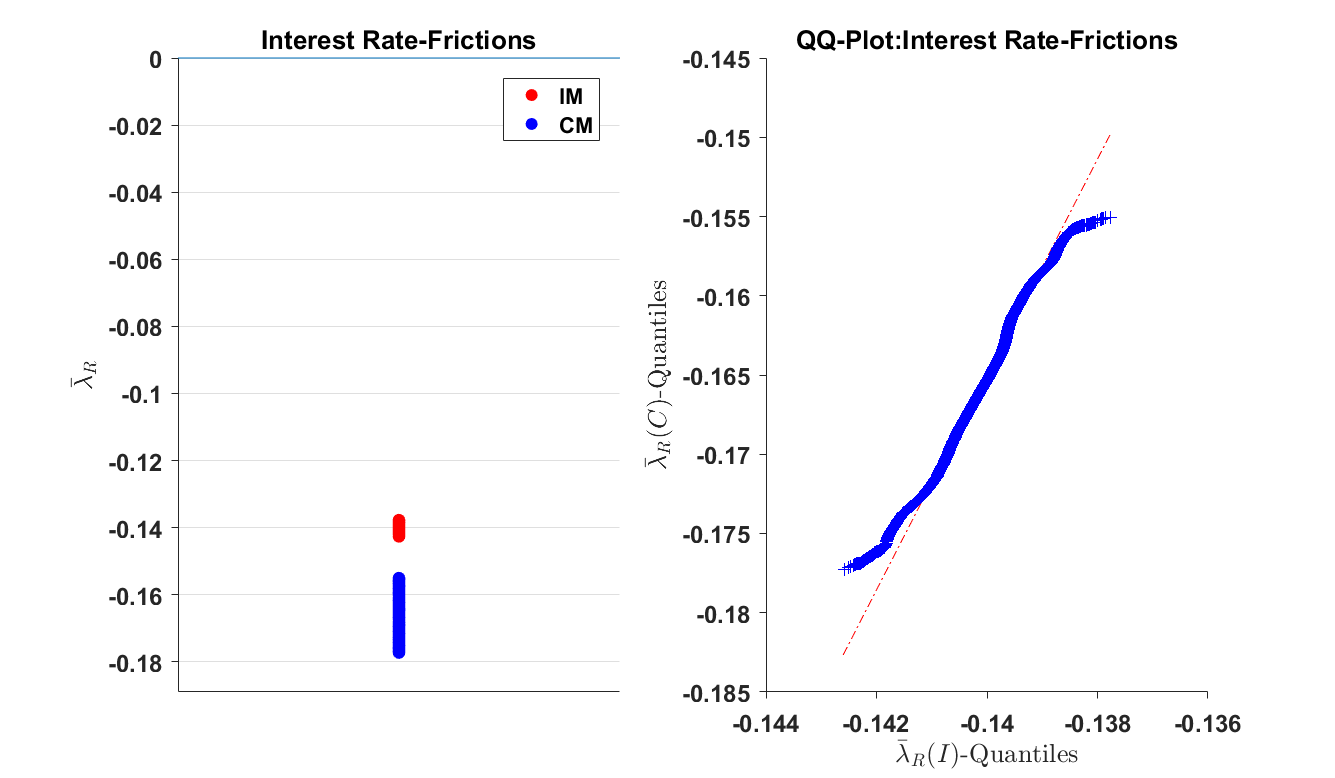}
		\includegraphics[scale=0.47]{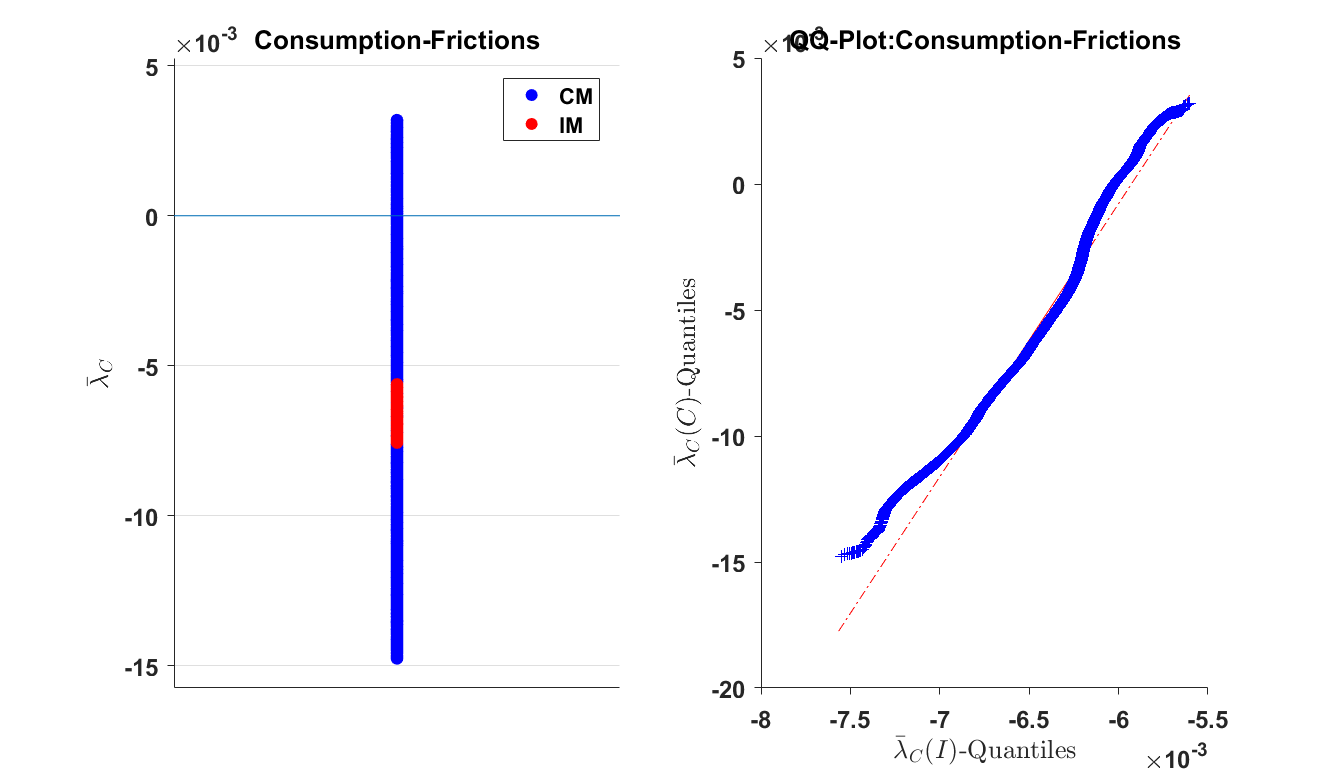}
		\protect\caption{Full Sample Wedges (due to Financial Frictions and Capital Adjustment Costs)}
	\end{centering}
\end{figure}
\begin{figure} [H]
	\begin{centering}
		\includegraphics[scale=0.47]{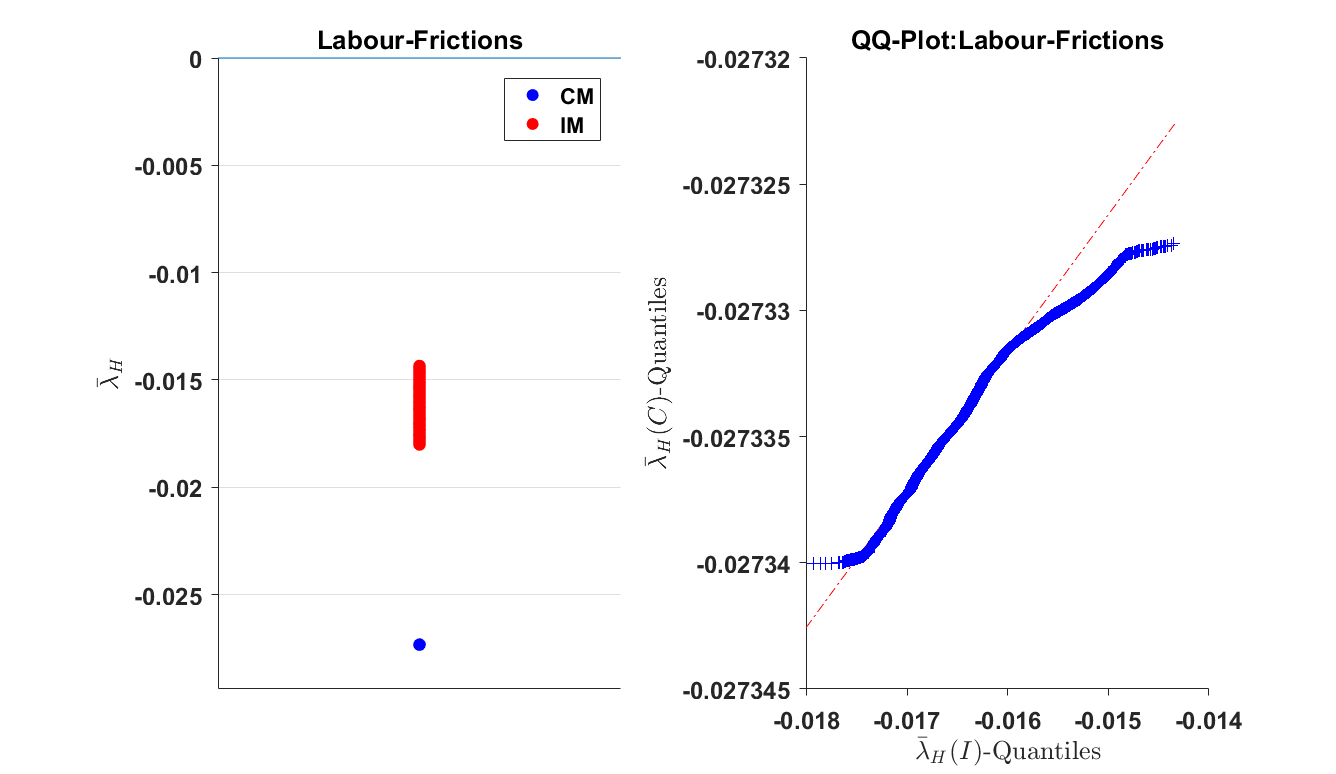}
		\includegraphics[scale=0.47]{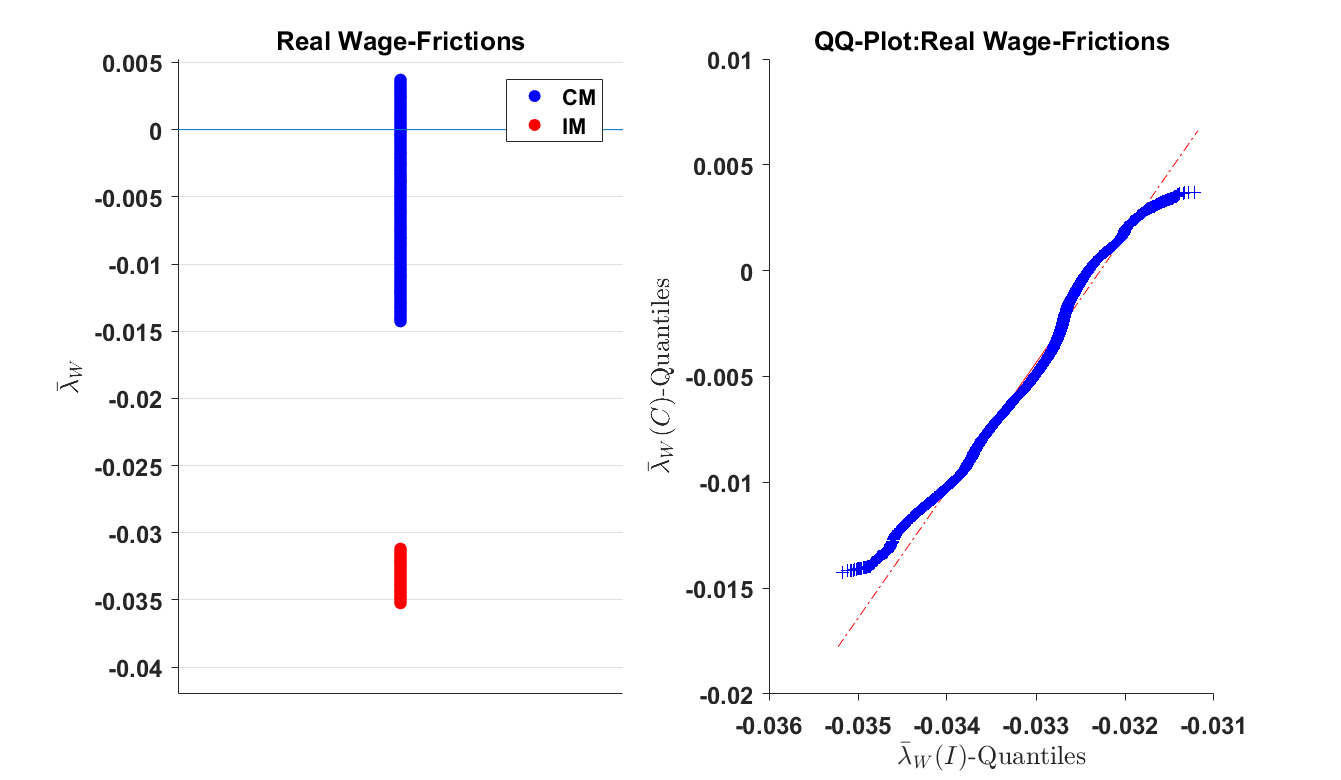}
		\protect\caption{Full Sample Wedges (due to Financial Frictions and Capital Adjustment Costs)}
	\end{centering}
\end{figure}

\begin{figure}[H]
	\begin{centering}

		\includegraphics[scale=0.47]{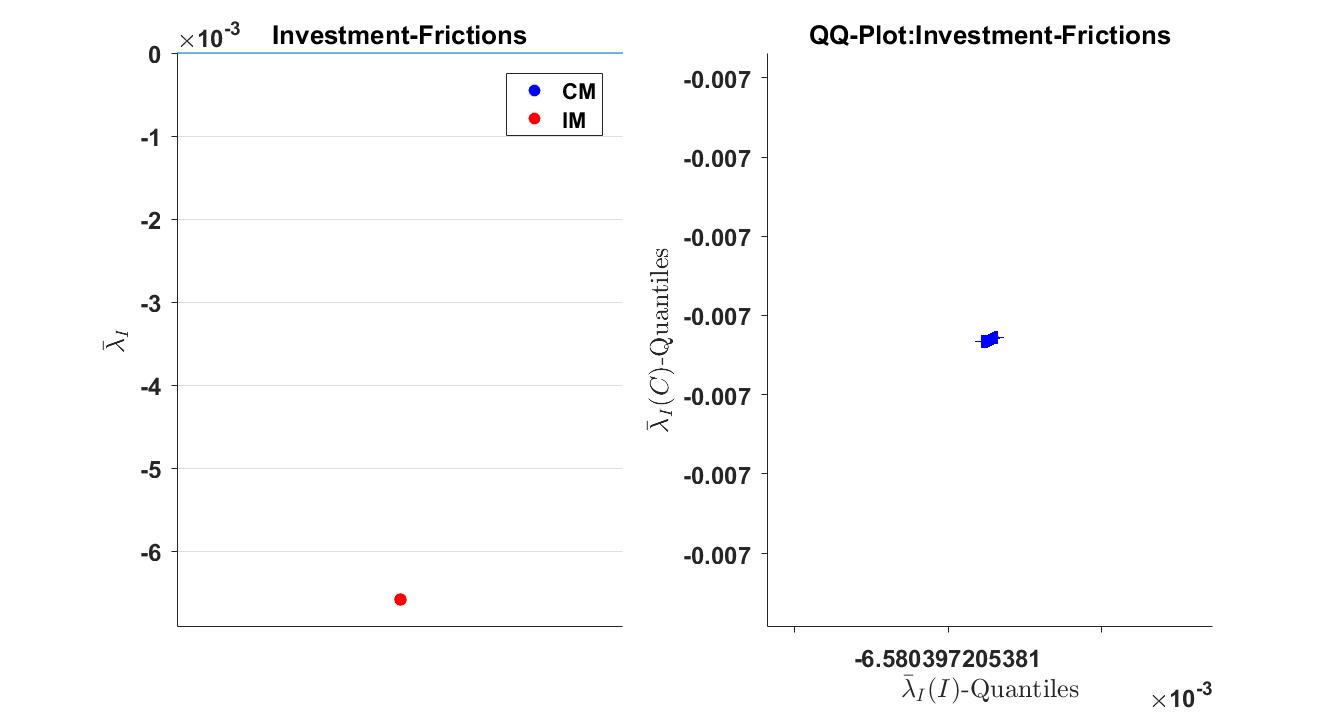}
		\protect\caption{Full Sample Wedges (due to Financial Frictions and Capital Adjustment Costs)}
	\end{centering}
\end{figure}

\bibliographystyle{econometrica}
\bibliography{biblio_thesis}

\end{document}